\documentclass[a4paper,11pt]{article}
\usepackage{times,bm,url,babel}
\usepackage[utf8]{inputenc}
\usepackage{mathrsfs}
\usepackage{amsmath,mathrsfs}
\usepackage{amssymb}
\usepackage[square,numbers]{natbib}
\usepackage{amsthm}
\usepackage{amstext}
\usepackage{array}
\usepackage{stmaryrd}
\usepackage[margin=2cm]{geometry}
\usepackage{thmtools, thm-restate}
\usepackage{url}

\newtheorem{theorem}{Theorem}[section]
\newtheorem{corollary}[theorem]{Corollary}
\newtheorem{lemma}[theorem]{Lemma}
\newtheorem{prop}[theorem]{Proposition}
\newtheorem{defn}[theorem]{Definition}

\newtheorem{manualtheoreminner}{Theorem}
\newenvironment{manualtheorem}[1]{%
  \IfBlankTF{#1}
    {\renewcommand{\themanualtheoreminner}{\unskip}}
    {\renewcommand\themanualtheoreminner{#1}}%
  \manualtheoreminner
}{\endmanualtheoreminner}

\numberwithin{equation}{section}
\numberwithin{theorem}{section}
\counterwithin{figure}{section}
\counterwithin{table}{section}

\usepackage{tikz}
\usepackage{pgfplots}
\usepackage{pgfplotstable}
\usetikzlibrary{pgfplots.groupplots}

\usepackage{subcaption}
\usepackage{physics}
\usepackage{array}
\usepackage{todonotes,soul}
\usepackage{authblk}
\usepackage{float}
\usepackage{mathtools}
\usepackage{centernot}
\usepackage{setspace}
\usepackage{comment}
\usepackage{graphicx}
\usepackage[algosection]{algorithm2e}
\usepackage{hyperref}
\hypersetup{
	breaklinks=true,   
	colorlinks=true,   
}

\SetKwComment{Comment}{/* }{ */}

\let\Im\relax
\DeclareMathOperator{\Im}{Im}

\DeclareMathOperator{\sgn}{sgn}
\DeclareMathOperator{\sinc}{sinc}
\DeclareMathOperator{\supp}{supp}
\DeclareMathOperator{\Conv}{Conv}
\DeclareMathOperator*{\Li}{Li}
\def\Li{\qopname\relax m{Li}}
\DeclareMathOperator{\Cov}{Cov}
\DeclareMathOperator{\Var}{Var}
\DeclareMathOperator{\RD}{RD}

\RestyleAlgo{ruled}

\newcommand{\Prob}[0]{\text{Prob}}
\newcommand{\boldvec}[2][]{\langle \boldsymbol{#2}_{#1}\rangle}

\newcommand{\Ic}{\mathcal{I}}

\DeclareMathOperator{\srsd}{srsd}

\newlength{\bibitemsep}
\newlength{\bibparskip}
\let\oldthebibliography\thebibliography
\renewcommand\thebibliography[1]{
  \oldthebibliography{#1}
  \setlength{\parskip}{\bibitemsep}
  \setlength{\itemsep}{\bibparskip}
}

\begin{document}

\title{Repeated quantum backflow and overflow}

\author[1,2]{Christopher J. Fewster\thanks{\tt chris.fewster@york.ac.uk}}
\author[1]{Harkan J. Kirk-Karakaya\thanks{\tt harkan.kirk-karakaya@york.ac.uk}}
\affil{\small Department of Mathematics,
	University of York, Heslington, York YO10 5DD, United Kingdom}
\affil[2]{\small York Centre for Quantum Technologies, University of York, Heslington, York YO10 5DD, United Kingdom}
\date{November 4, 2025}





\maketitle

\begin{abstract}
 Quantum particles moving in one dimension with rightwards momentum can exhibit the surprising phenomenon of quantum backflow: a net probability transfer to the left-hand half-line over a finite time interval. 
 We generalise this phenomenon by considering the sum of probability differences for $M$ disjoint time intervals. 
 In classical mechanics, the total backflow lies in the interval $[-1,0]$ for all $M$, indicating rightwards probability transfer. 
 By contrast, we show that the maximum $M$-fold quantum backflow is positive and unbounded from above as $M$ increases, demonstrating the existence of repeated backflow.  
     
 For $M\ge 2$, a new phenomenon of `quantum overflow' is discovered: there are states whose total backflow is below $-1$, which is impossible for classical particles. The extent of the backflow and overflow effects is described by a hierarchy of backflow and overflow functions and constants, of which the $M=1$ backflow constant was first studied by Bracken and Melloy. Limiting cases of the backflow and overflow functions are studied, including cases in which two disjoint intervals merge. 
    
 The analytical results are supported by detailed numerical investigations. Using numerical acceleration methods, we obtain a new estimate of the Bracken--Melloy constant of $0.0384506$, slightly lower than the previously accepted value of $0.038452$. 
\end{abstract}

\maketitle

\section{Introduction}
 
    Quantum backflow (QB) is a counterintuitive phenomenon in which a free quantum mechanical particle on the line, in a state with purely rightwards momentum, can exhibit a temporary probability transfer to the left-hand half-line. Here, we understand `purely rightwards momentum' or `purely positive momentum' to mean that the measured value of the momentum would be nonnegative with probability $1$, or equivalently that the support of the wavefunction in momentum space is contained in the positive half-line. The existence of 
    states exhibiting backflow has been known since at least 1969, when Allcock \cite{allcock1969time} showed that there exist states with purely positive momentum whose quantum mechanical current at the origin can be arbitrarily negative.  
    
    The first detailed study of QB was conducted in 1994 by Bracken and Melloy \cite{MR1280379}. They showed that for states that exhibit QB the probability backflow is bounded by a state independent dimensionless value $0<c_\textnormal{BM} < 1$, independent of $\hbar$, which they calculated to be $c_\textnormal{BM} \approx 0.04$~\cite{MR1280379}, by showing that $c_\textnormal{BM}$ is the largest value in the spectrum of a bounded self-adjoint operator. There have since been two further calculations of $c_\textnormal{BM}$ to date, first by Eveson, Fewster and Verch ($c_\textnormal{BM} \approx 0.03845(2)$)\cite{MR2119354} and shortly thereafter by Penz~\emph{et al.}  ($c_\textnormal{BM} \approx 0.0384517)$ \cite{MR2198971}. The first non-trivial upper bound on the backflow constant was calculated only recently by Trillo, Le and Navascu\'es to be $c_\textnormal{BM} < 0.072$ \cite{trillo2023quantum}. Due to the small value of $c_\textnormal{BM}$, experimental verification of QB is an ongoing effort. Goussev has shown that when considering an angular analogue of QB for a particle on a ring, the associated maximum possible backflow is larger than in the linear case with $c_\textnormal{ring} \approx 2.6 \: c_\textnormal{BM}$ \cite{MR4229055}. Angular QB has also been studied in the case of a charged massless fermion \cite{MR4576897}. Palmero~\emph{et al.} have shown how QB could be experimentally verified by measuring density fluctuations in a Bose-Einstein condensate \cite{palmero2013detecting}. Bracken has explored a whole family of QB related problems where the quantum states considered have their momentum restricted to $p > p_0$ for some fixed $p_0$ not necessarily positive \cite{Bracken_2021}. Other variations of QB that have been considered include but are not limited to: QB in two dimensions \cite{MR4572336}, QB in a relativistic setting \cite{ashfaque2019relativistic}, the relation between QB and the local wave number \cite{MR2726698} and QB across a black hole horizon \cite{MR4355013}. Finally, Bostelmann, Cadamuro and Lechner~\cite{MR3711363} studied the extent to which scattering states exhibit the backflow effect for general short-range potentials, while Yearsley and Halliwell showed that the dependency of the maximal QB on $\hbar$ reappears \cite{yearsley2013introduction} once a realistic measurement apparatus is considered.

     Previously, QB has only ever been considered over a single fixed time interval $[t_1,t_2]$. As first shown by Bracken and Melloy, the associated spectrum of backflow values is independent of $t_1$ and $t_2$ \cite{MR1280379}. In this paper (following a suggestion by Reinhard Werner) we will consider the question of quantum backflow over multiple disjoint time intervals. Given an integer $M\geq  1$ and a strictly increasing list of times 
     $\boldvec[M]{t}=\langle t_1,\ldots,t_{2M}\rangle$,
     we study the total backflow 
     \begin{equation}
       \Delta^{(M)}_{\boldvec[M]{t}}(\psi) =
       \sum_{j=1}^M \left(\Prob_\psi(X<0|t=t_{2j})-
       \Prob_\psi(X<0|t=t_{2j-1})   \right) 
     \end{equation}
     exhibited by quantum state $\psi$ over the union of the time intervals
     $[t_1,t_2],...,[t_{2M-1},t_{2M}]$. A basic question is whether or not the total backflow is also limited by the Bracken--Melloy constant, or whether increasing $M$ can increase the amount of backflow beyond $c_{\textnormal{BM}}$. In other words, are there states which exhibit 
     repeated periods of significant backflow? Further, as a sum of $M$ probability differences, $\Delta^{(M)}_{\boldvec[M]{t}}(\psi)$ is, \emph{a priori}, bounded between $\pm M$. Can the maximum total backflow grow unboundedly as $M\to\infty$? Are there states in which one finds
     $\Delta^{(M)}_{\boldvec[M]{t}}(\psi)<-1$ -- i.e., in which 
     more than unit probability is transferred to the right-hand half-line, in total, over the given intervals? If so, is this effect unbounded as $M\to\infty$? Are the answers to the previous questions different for classical particle mechanics?
     
     We will answer all these question affirmatively. 
     For $M$ time intervals of equal duration, separated by gaps of the same duration, we show that there are positive-momentum states $\psi_M$ so that
     $\Delta^{(M)}_{\boldvec[M]{t}}(\psi_M)\to \infty$, at least as fast 
     as $\mathcal{O}(M^{1/4})$. Moreover, we discover the new phenomenon that, for $M\ge 2$, there are indeed states with $\Delta^{(M)}_{\boldvec[M]{t}}(\psi)<-1$, an effect that we call \emph{quantum overflow}. For the intervals of equal duration separated by gaps of the same duration, we will prove that there are positive-momentum states $\varphi_M$
     with $\Delta^{(M)}_{\boldvec[M]{t}}(\varphi_M)\to -\infty$, at least as fast as $\mathcal{O}(M^{1/4})$. Thus, both quantum backflow and overflow are unbounded in the limit of large $M$. By contrast,
     the total backflow of a classical particle is constrained to $[-1,0]$, independent of $M$. Our analytical results are complemented by a 
     numerical investigation of the magnitude of the backflow and overflow effects for $M\leq 4$. As has been pointed out by one of the referees, overflow could be studied profitably without the restriction to positive momentum. Here, we keep the restriction because it turns out that there is some interplay between the backflow and overflow phenomena in the positive momentum case.  The exploration of overflow in greater generality is left for future work.

Before continuing, we give a preview of some numerical results that are discussed in greater depth below. In Section~\ref{sec:numerical_methodology_results}, we describe how sequences of states may be found that display significant backflow or overflow for a given set of intervals. Consider in particular $M$ equally spaced intervals of an equal duration that will be chosen as the unit of time in the following discussion. Fixing the unit of mass to be twice the mass of the quantum particle, the unit of length is then fixed so that $\hbar=1$. The states $\chi^{(M)}_\textnormal{back}$ are drawn from calculations reported in Section~\ref{sec:numerical_methodology_results} and exhibit $M$-fold backflow over the $M$ time intervals $[-M+2k+1/2, -M+2k+3/2]$ for $0\le k\le M-1$,
with total backflow $0.0350766$, 
$0.0543496$, $0.0684048$, and 
$0.0797003$ for $M=1,2,3,4$ respectively. Fig.~\ref{fig:backvecJM1_4}  plots the probability flux at $x=0$ against time, and shows that the flux is negative in the stated time intervals. The position probability density plots for these states (Figs.~\ref{fig:backvecsM13PR} and~\ref{fig:backvecsM24PR} in Section~\ref{sec:numerical_methodology_results}) display $M+1$ main peaks (some with substructure), symmetrically distributed about $x=0$. 

The time evolution of the position probability density in the $M=2$ case is illustrated in Fig.~\ref{fig:backvecM2heatmap} for times $-1.5\le t\le 1.5$.   
In broad outline, the three main peaks move to the right, while nonetheless
reshaping so that probability flows to the left-hand half-line during the time intervals $[-1.5,-0.5]$ and $[0.5,1.5]$. At $t=-1.5$ the leading peak lies in $x>0$ and is higher than the other two; by $t=-0.5$, the leading peak has diminished while the central peak and the trailing peak, still in $x<0$, have grown. By $t=0.5$, the central peak is dominant, lying in $x>0$, and over the time interval $[0.5,1.5]$ this peak diminishes while the trailing peak, still in $x<0$, grows. These reshapings are the result of constructive and destructive interference between the slow-moving main peaks and higher-frequency pulses that move through at a higher velocity. As it evolves, the wavepacket spreads on a characteristic timescale $t_S = m(\Delta X)/\Delta P$, where $\Delta X$ and $\Delta P$ are the dispersions of position and momentum at $t=0$. For $\chi^{(2)}_{\textnormal{back}}$, we computed $t_S\sim 1.7$, showing that backflow takes place on comparable timescales to $t_S$, and indeed Fig.~\ref{fig:backvecJM1_4} shows that the probability flux is positive for times $|t|\gtrsim t_S$.

\begin{figure}
    \centering
    \includegraphics[width=0.8\textwidth]{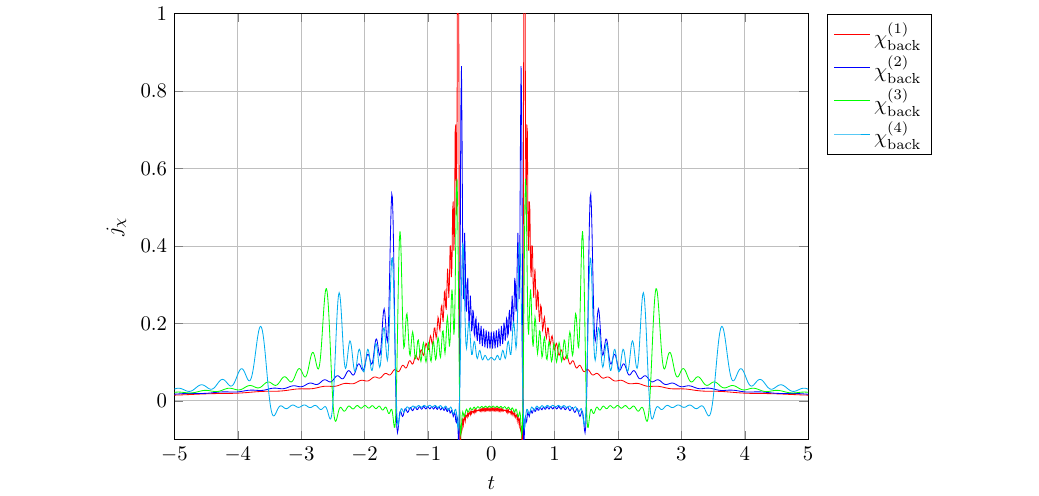}
    \caption{Probability flux at $x=0$ of the $M$-fold backflow states $\chi_{\textnormal{back}}^{(M)}$ computed in Section~\ref{sec:numerical_methodology_results}.}
    \label{fig:backvecJM1_4}
\end{figure}

\begin{figure}
    \centering
    \includegraphics[width=0.8\textwidth]{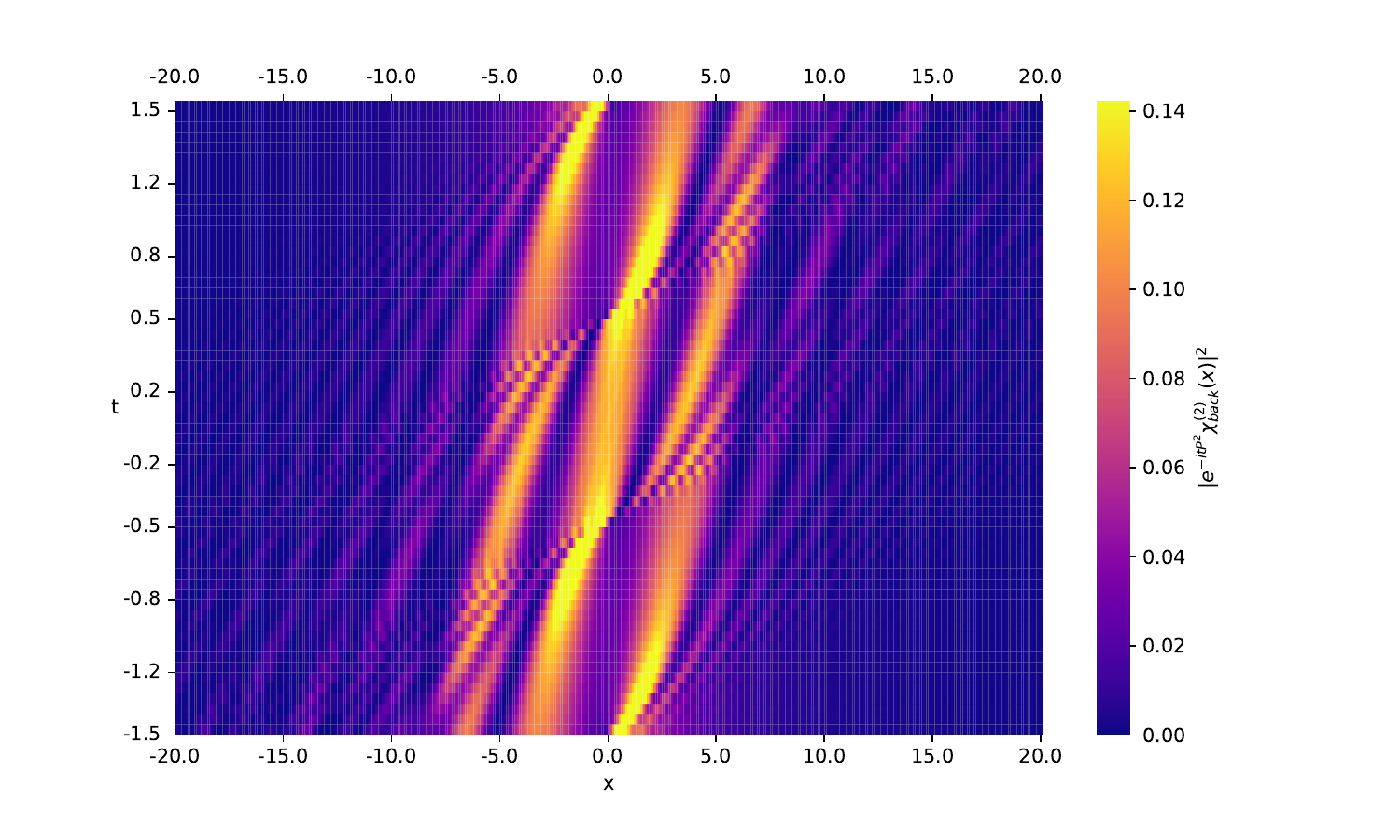}
    \caption{Time evolution of the position probability density for the $M=2$ backflow state $\chi^{(2)}_{\textnormal{back}}$.}
    \label{fig:backvecM2heatmap}
\end{figure}

The calculations in Section~\ref{sec:numerical_methodology_results} also provide states $\chi^{(M)}_\textnormal{over}$ exhibiting significant overflow on the $M$ equally spaced intervals of equal (unit) duration given above, with total backflow equal to $-1.00213$
$-1.00628$, and
$-1.01122$ for $M=2,3,4$ respectively.
Fig.~\ref{fig:overvecM2heatmap} shows the time evolution of the position probability density evolution for $\chi^{(2)}_\textnormal{over}$, which has spreading time  $t_S\sim 1.4$. At $t=-1.5$ the bulk of the probability distribution lies in $x<0$. During the initial time period $[-1.5,-0.5]$ the slow-moving main peak and faster high-frequency peaks move forwards and net probability flows into $x>0$. At $t=-0.5$ the distribution has two main peaks with the dominant one in $x>0$; during the time interval $[-0.5,0.5]$, these peaks reshape and probability flows backwards into the negative half-line, before again flowing forwards in the final period $[0.5,1.5]$. Plots of the flux (given in Supplementary Section~\ref{sec:pos_plots}) show that the overflow states for $M=2,3,4$ all exhibit backflow between the overflow intervals. This is in line with the following elementary argument. Note that 
\begin{equation}
\Delta^{(2)}_{\langle t_1,t_2,t_3,t_4\rangle}(\psi) = \Delta^{(1)}_{\langle t_1, t_4\rangle}(\psi)-\Delta^{(1)}_{\langle t_2, t_3\rangle}(\psi)\ge -1-\Delta^{(1)}_{\langle t_2, t_3\rangle}(\psi),
\end{equation}
because overflow does not occur for $M=1$. Thus, if $\psi$ exhibits $2$-fold overflow for intervals $[t_1,t_2]$ and $[t_3,t_4]$, then $\Delta^{(2)}_{\langle t_1,t_2,t_3,t_4\rangle}(\psi) <-1$ and consequently $\Delta^{(1)}_{\langle t_2, t_3\rangle}(\psi)>0$, which shows that $\psi$ also exhibits backflow over $[t_2,t_3]$. 

\begin{figure}
    \centering
    \includegraphics[width=0.8\textwidth]{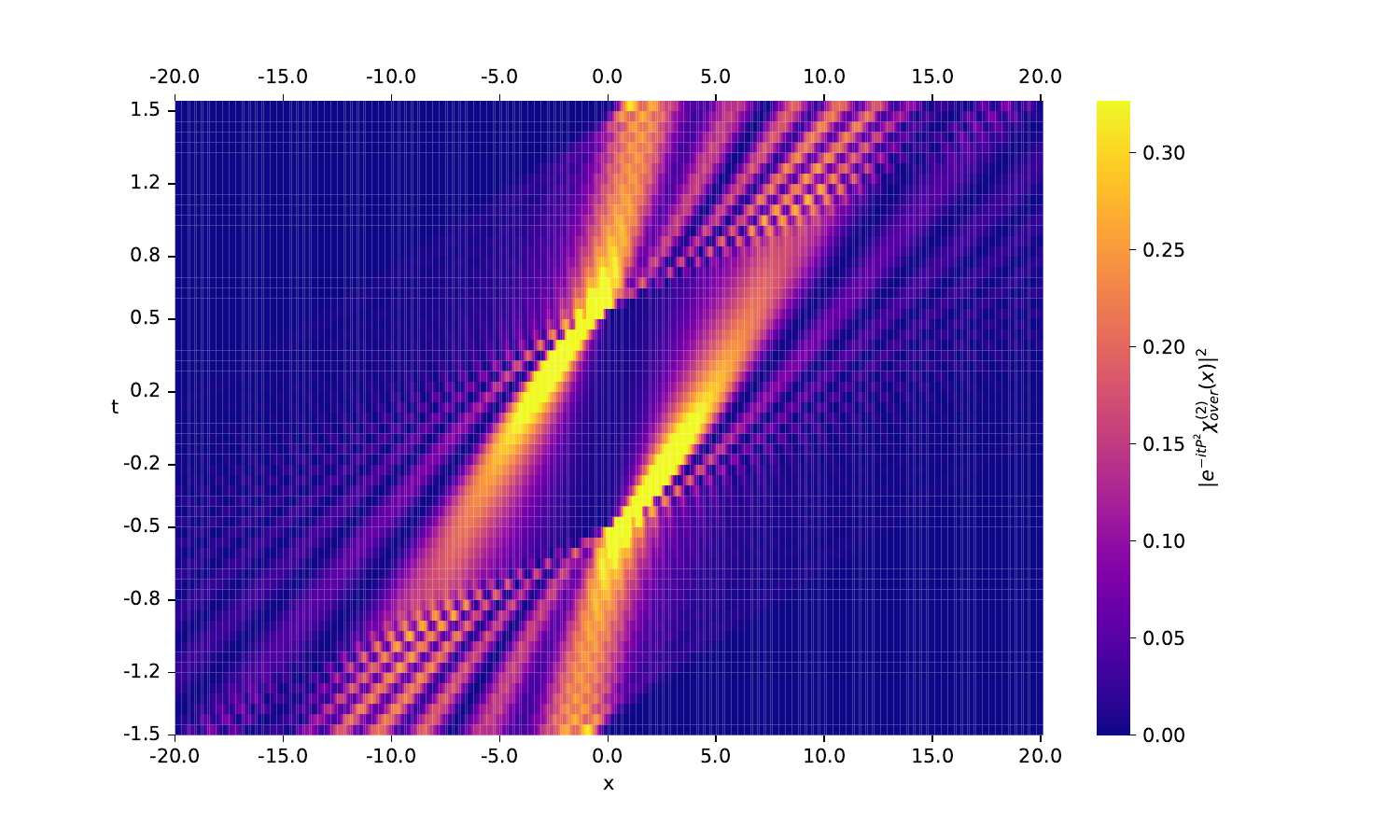}
    \caption{Time evolution of the position probability density for the $M=2$ overflow state $\chi^{(2)}_{\textnormal{over}}$.}
    \label{fig:overvecM2heatmap}
\end{figure}

To gain some understanding into why $M$-fold backflow can exceed the Bracken--Melloy limit, consider a normalised positive momentum state $\psi$ that exhibits backflow $\Delta = \Delta^{(1)}_{\langle t_1,t_2\rangle}(\psi)=(1-\epsilon)c_\textnormal{BM}>0$ in time interval $[t_1,t_2]$, where $0\le \epsilon\ll 1$. Then the time-evolved state $\psi_\tau$ will exhibit the same amount of backflow over the interval $[t_1-\tau,t_2-\tau]$. Assuming $\tau$ is sufficiently large that the two intervals are well-separated, a linear combination
$\alpha\psi+\beta\psi_\tau$ will exhibit a total backflow of approximately $(|\alpha|^2+|\beta|^2)\Delta^{(1)}_{\langle t_1,t_2\rangle}(\psi)$ across the two intervals, also assuming that cross terms between $\psi$ and $\psi_\tau$ can be neglected. Now $\alpha$ and $\beta$ must be chosen so that the linear combination is normalised. If $\psi$ and $\psi_\tau$ were orthogonal, this would require $|\alpha|^2+|\beta|^2=1$ and the total $2$-fold backflow would be no more than $\Delta^{(1)}_{\langle t_1,t_2\rangle}(\psi)$. However, linear combinations in which some destructive interference occurs will require a normalisation with $|\alpha|^2+|\beta|^2>1$ whereupon the $2$-fold backflow exceeds $\Delta^{(1)}_{\langle t_1,t_2\rangle}(\psi)$. Taking $\epsilon \to 0$, it is plausible that $c_\textnormal{BM}$ can be exceeded by $2$-fold backflow. Although this argument leaves much to be desired, it illustrates a point that will be proved rigorously later on. In Section~\ref{sec:QM}, we will also give a simple (and precise) argument why the existence of nontrivial backflow implies the existence of nontrivial $2$-fold overflow, approaching the limit $-1-c_\textnormal{BM}$.

     In more detail, we show in Section \ref{M-fold_Backflow} that, for all $M$, the classical analogue of $\Delta^{(M)}_{\boldvec[M]{t}}$ is constrained to the negative unit interval $[-1,0]$ for classical statistical ensembles with positive momenta. Passing to quantum mechanics, a straightforward adaptation of arguments in~\cite{MR1280379,MR2419566} show that for each strictly increasing list of times $\boldvec[M]{t}$, there is a self-adjoint bounded operator $B^{(M)}_{\boldvec[M]{t}}$ on $L^2(\mathbb{R}^+,dk)$ so that
     \begin{equation}
     \Delta^{(M)}_{\boldvec[M]{t}}(\psi)=
     \ip{\hat{\psi}}{B^{(M)}_{\boldvec[M]{t}}\hat{\psi}}
     \end{equation}
     for normalised positive-momentum state $\psi\in L^2(\mathbb{R},dx)$, where $\hat{\psi}\in L^2(\mathbb{R}^+,dk)$ is the Fourier transform of $\psi$, i.e., its momentum representation. It follows that the set of possible values for the total backflow is the numerical range of the operator $B^{(M)}_{\boldvec[M]{t}}$, and that the supremum and infimum of these values are given by the maximum and minimum points of the spectrum $\sigma(B^{(M)}_{\boldvec[M]{t}})$, generalising the original insight of~\cite{MR1280379}. The unitary equivalence of operators $B^{(M)}_{\boldvec[M]{t}}$ under simultaneous uniform translation of all times, and under simultaneous uniform dilation, implies that the maximum and minimum of $\sigma(B^{(M)}_{\boldvec[M]{t}})$ are functions of the successive ratios of successive differences of the $t_j$'s. That is, one has bounds
     \begin{equation}
         b_{\textnormal{over}}^{(M)}(\srsd\boldvec[M]{t})\leq
         \Delta^{(M)}_{\boldvec[M]{t}}(\psi)\leq
         b_{\textnormal{back}}^{(M)}(\srsd\boldvec[M]{t})
     \end{equation}
    for all positive-momentum states $\psi$,
    where 
    \begin{equation}
        \srsd \boldvec[M]{t}=((t_3-t_2)/(t_2-t_1),\ldots, (t_{2M}-t_{2M-1})/(t_{2M-1}-t_{2M-2}))
    \end{equation}
    and  $b_{\textnormal{back/over}}^{(M)}:\mathbb{R}_{>0}^{2M-2}\to \mathbb{R}$ are \emph{backflow/overflow functions} defined so that
    \begin{equation}
         b_{\textnormal{back}}^{(M)}(\srsd\boldvec[M]{t})
         =\max 
         \sigma(B^{(M)}_{\boldvec[M]{t}}), \qquad
         b_{\textnormal{over}}^{(M)}(\srsd\boldvec[M]{t})
         =\min
         \sigma(B^{(M)}_{\boldvec[M]{t}}). 
    \end{equation}
    The supremum and infimum of $b_{\textnormal{back}}^{(M)}$ and $b_{\textnormal{over}}^{(M)}$ over $\mathbb{R}_{>0}^{2M-2}$ define sequences of backflow 
     and overflow constants, $c^{(M)}_{\textnormal{back}}$ and $c^{(M)}_\textnormal{over}$. For $M=1$ one has $c^{(1)}_\textnormal{back}=c_{\textnormal{BM}}$ and $c^{(1)}_\textnormal{over}=-1$.
     
     Next, in Section~\ref{the_spectrum}, we prove two key results regarding backflow and overflow. First, we show that the backflow and overflow constants satisfy the bounds 
     \begin{equation}
         c^{(M)}_\textnormal{back}\leq M c_{\textnormal{BM}}, \qquad
         c^{(M)}_\textnormal{back}\geq -1-(M-1) c_{\textnormal{BM}}
     \end{equation}
     (see Theorem~\ref{thm: spectral_outer_bounds} for the precise statement and more detail). Second, by considering the example of equal duration intervals separated by gaps of the same duration and using simple sequences of trial positive-momentum states, we prove that
    \begin{equation}
        c^{(M)}_\textnormal{back}\to\infty,\qquad c^{(M)}_\textnormal{over}\to -\infty
    \end{equation}
    as $M\to\infty$, at least as fast as $\mathcal{O}(M^{1/4})$. See Theorem~\ref{thm: spectral_asymptotic_unboundedness} for the detailed statement.

     In Section~\ref{sec:limiting cases}, we study various limiting cases of the backflow functions as some of their parameters are taken to zero or infinity. 
     The main result here, Theorem~\ref{thm:blims}, provides an understanding of the relation between backflow and overflow functions for different $M$, and a proof that the backflow and overflow constants are monotonic in $M$, and obey
     \begin{equation}
     c^{(M+L)}_\textnormal{back}\geq -1 - c^{(M)}_\textnormal{over}, \qquad c^{(M+L)}_\textnormal{over}\leq -1 - c^{(M)}_\textnormal{back}
     \end{equation}
     for all $M,L\in\mathbb{N}$. 
     Various consequences are drawn. In particular it is shown that when two time intervals merge, the spectra of the corresponding backflow operators do not converge to the spectrum of the backflow operator for the merged intervals. 
     For example, as the limit is approached, the overflow function $b_{\textnormal{over}}^{(2)}$ tends to $-1-c_\textnormal{BM}$, while once the intervals have merged the appropriate overflow function is $b_{\textnormal{over}}^{(1)}\equiv -1$. Thus an arbitrarily small `recovery time' between the time intervals, during which the state exhibits backflow, is sufficient to accommodate states with maximal $2$-fold overflow. Technically, the failure of spectra to converge is related to the fact that the backflow operators only converge in the strong topology rather than in norm topology. However, if one restricts attention to positive-momentum states with a fixed momentum cutoff, the maximum backflow/overflow do converge to the values for a single interval as the intervals merge. This is significant because numerical tests will typically involve such a cutoff, and would yield misleading results unless the cutoff is increased  as the merger is approached.
     
     In Section~\ref{section: numerical_calculation}, we present a numerical investigation of
     the situation of $M$ equally spaced backflow periods of equal length. In particular, we aim to estimate the 
     values of $b^{(M)}_\textnormal{back/over}(1,1,\ldots,1)$ for $1 \leq M \leq 4$. Our method is to study the compression
     of the appropriate $M$-fold backflow operator to the $(N+1)$-dimensional subspace 
     $V_{N,a,\delta} = \textnormal{span}(\psi_{0,a,\delta},\dots,\psi_{N,a,\delta})$, where $\psi_{n,a,\delta}(p)\propto p^{n+\delta}e^{-ap}$ are normalised  
     $L^2(\mathbb{R}^+,dp)$ vectors depending on parameters 
     $a$, $\delta$. The maximum/minimum eigenvalue $\lambda_\textnormal{back/over}^{(M)}(a,\delta; N)$ of the compression can be obtained by solving a generalised eigenvalue problem relative to the Gram matrix of 
     $\psi_{0,a,\delta},\dots,\psi_{N,a,\delta}$; high precision is required in the calculation because 
     the minimum eigenvalue of the Gram matrix tends to $0$ geometrically as $N\to\infty$. Fortunately, we are able to derive closed form expressions for all the matrix elements required in terms of special functions implemented in the arbitrary precision library FLINT~\cite{FLINT,MR4000439}. The upshot is that the values
     $\lambda_\textnormal{back/over}^{(M)}(a,\delta; N)$ are 
     computed accurate to tens (and sometimes hundreds) of decimal places. 
     
     For choices of $a$ and $\delta$ that we motivate, we obtain sequences $\lambda_\textnormal{back/over}^{(M)}(N)$ for $1 \leq M \leq 4$ and $N\leq 500$. Up to numerical precision, these sequences converge monotonically to 
     $b^{(M)}_{\textnormal{back/over}}(1,1,\ldots,1)$ as $N\to\infty$, thus providing rigorous bounds on these quantities. The sequences converge quite slowly, but
     numerical acceleration techniques can be used to estimate the limiting values (albeit losing the guarantee of monotone convergence).  
     Our acceleration techniques assume that $\lambda_\textnormal{back/over}^{(M)}(N)$ has an asymptotic expansion in powers of $N^{-1/2}$ -- this is supported  by our results, except in the case of $\lambda_\textnormal{over}^{(1)}$, which converges rapidly to $c_\textnormal{over}^{(1)}=-1$.
     
     The acceleration methods lead to conjectured upper and lower bounds for $b^{(M)}_\textnormal{back/over}(1,1,\ldots,1)$ which can be found in Tables~\ref{tab: backflow_const_bounds} and~\ref{tab: overflow_const_bounds}.
     In particular, our conjectured value for $c_{\textnormal{BM}}$ is $0.0384506$ to 6 significant figures. We will discuss the relation to earlier calculations of~\cite{MR1280379,MR2119354} in Section~\ref{sec:numerical_acceleration}.
     
     We also compute the normalised eigenvectors $\psi_\textnormal{back/over}^{(M)}(N; \cdot) \in V_{N,a,\delta}\subset L^2(\mathbb{R}^+,dp)$ of the compressed backflow operators corresponding to eigenvalues $\lambda^{(M)}_\textnormal{back/over}(N)$.  For $M=1$ the overflow eigenvectors $\psi_\textnormal{over}^{(1)}(N; \cdot)$ resemble Gaussians roughly peaked at $1.5N$ and with broadening width -- see Figure~\ref{fig:psi_over_1}.
     We argue that this behaviour is in line with a result proved in~\cite{MR2419566} that $-1$ lies in the essential spectrum of the single backflow operator. The remaining eigenvectors have an oscillatory structure that becomes more complex as $M$ increases, and appear to decay as $\mathcal{O}(p^{-3/4})$ for $p\to\infty$. The backflow eigenvectors also diverge as $p\to 0$, apparently as $\mathcal{O}(p^{-1/4})$. We  present plots of $\psi_\textnormal{back/over}^{(M)}(N;p)$ for $N=500$, where in most cases the wavefunctions are multiplied by a factor $p^{3/4}$ to remove the decay at large momentum, thus more clearly showing the oscillatory structure. Our results also suggest that $b^{(M)}_{\textnormal{back}}(1,\ldots,1)$ for $1\le M\le 4$ and $b^{(M)}_{\textnormal{over}}(1,\ldots,1)$ for $2\le M\le 4$ are isolated eigenvalues of the $M$-fold backflow operator.
We conclude in Section~\ref{sec:conclusion} with a summary and outlook. The Supplementary Material comprises appendices giving technical details deferred from the main text and additional plots.

\section{$M$-fold backflow}
\label{M-fold_Backflow}
\subsection{Classical statistical mechanics}
Consider a classical ensemble of particles $\mathcal{C}$ of mass $\mu$ under free evolution on the line.  
Let the distribution of particles on phase space $\mathcal{P}=\mathbb{R}^2$ at time $t$ be given by a probability measure $\rho_t$ so that a randomly chosen member of the ensemble has phase space position in Borel subset $S\subseteq \mathcal{P}$ at time $t$ with probability
\begin{equation}
\Prob_\rho((X,P)\in S\mid t)=\rho_t(S),
\end{equation}
where we write $\rho=\rho_0$ for the state at time $t=0$. The ensemble average of an observable $f\in C(\mathcal{P})$ at time $t$ is given by 
$\int f\, d\rho_t = \int f\circ\tau_t \,d\rho$ by Liouville's theorem, where 
\begin{equation}
    \tau_t(x,p)=(x+pt/\mu,p)
\end{equation}
is the 
forwards Hamiltonian time evolution through time $t$. Correspondingly, we find that $\rho_t(S)=\rho(\tau_t^{-1}(S))$. If the ensemble only contains particles with non-negative momentum then one has 
\begin{equation}
    \rho(S)= \rho(S\cap(\mathbb{R}\times [0,\infty)))
\end{equation}
for any Borel subset $S\subseteq\mathcal{P}$, and consequently,
\begin{align}
    \Prob_\rho(X\in (-\infty,0)\mid t) &= \rho_t( (-\infty,0)\times \mathbb{R}) \nonumber \\
    &= 
    \rho(\tau^{-1}_t((-\infty,0)\times \mathbb{R}))\cap (\mathbb{R}\times [0,\infty)) \nonumber\\
     &= \Prob_\rho((X,P)\in  \{(x,p)\in\mathcal{P}: x<-pt/\mu,~p\ge 0\}).
\end{align}
The difference between the probabilities of finding a randomly chosen particle on the left-hand half line at time $t'$ and an earlier time $t$ is therefore
\begin{equation}
   \Prob_\rho(X\in (-\infty,0)\mid t' )- \Prob_\rho(X\in (-\infty,0)\mid t )=-
    \Prob_\rho((X,P)\in S_{t,t'}),
\end{equation}
where
\begin{equation}
S_{t,t'}=   \{(x,p)\in\mathcal{P}: -pt'/\mu\le x<-pt/\mu ,~p\ge 0\}
\end{equation}
for any $t'>t$. Note that we may replace $p\ge 0$ by $p>0$ in the formula $S_{t,t'}$ without loss
because $S_{t,t'}\cap (\mathbb{R}\times\{0\})=\emptyset$.
Accordingly, $-\rho(S_{t,t'})$ measures the probability backflow between times $t$ and $t'$; as it is clearly nonpositive, we see that there is no classical backflow. 

Now consider adding together the backflow for a set of disjoint time intervals. Specifically, 
given $M \in \mathbb{N}$ and times $t_1 < t_2 < \cdot \cdot \cdot < t_{2M-1}<t_{2M}$, the total amount of probability backflow is given by 
\begin{equation}
\begin{aligned}
\label{classical_backflow}
    \Delta^{(M)}_\textnormal{classical}(\rho)&:= \sum_{j=1}^M\left( \Prob_\rho(X \in (-\infty,0)\mid t_{2j})-\Prob_\rho(X \in (-\infty,0)\mid t_{2j-1})\right)\\
    &=-\sum_{j=1}^M\Prob_\rho(X\in S_{t_{2j-1},t_{2j}}) 
=-\Prob_\rho\left(X\in \bigcup_{j=1}^M S_{t_{2j-1},t_{2j}}\right)
\end{aligned}
\end{equation}
because the sets $S_{t_{2j-1},t_{2j}}$ are disjoint
for distinct $j$.

It follows that, for a classical ensemble with non-negative momentum, the total amount of probability backflow over multiple disjoint time periods is bounded by
\begin{equation}\label{eq:classicalbackflowbound}
    -1 \leq \Delta^{(M)}_\textnormal{classical}(\rho) \leq 0
\end{equation}
for every $M$. In particular, a classical ensemble cannot exhibit  positive probability backflow, and the probability transfer
in the forward direction is bounded by unity.  By contrast, it is well known that quantum particles can exhibit backflow over a single time interval, thus violating the upper bound in the analogue of~\eqref{eq:classicalbackflowbound}. One of the main results of this paper is to see that the violation increases with $M$, and that the lower bound of~\eqref{eq:classicalbackflowbound} is also violated in quantum theory for $M\ge 2$.

We remark that if one drops the restriction $p>0$ then a similar argument to the above shows that $\Delta^{(M)}_\textnormal{classical}(\rho)$ can be written as a difference of two probabilities and therefore lies $[-1,1]$. Thus, while backflow is certainly possible (but unsurprising) in this case, classical overflow is always forbidden. 

\subsection{Quantum mechanics}\label{sec:QM}

We turn to the quantum case, using units in which $\hbar = 1$. Consider the motion of a free quantum particle of mass $\mu$ with normalised state vector $\psi_t\in L^2(\mathbb{R}^+)$ at time $t$, whose dynamics
are governed by the free Schr\"{o}dinger equation
\begin{equation}
\label{SchrodingerEquation}
     i \partial_t \psi_t   = -\frac{1}{2\mu} \partial_x^2 \psi_t
\end{equation}
and initial condition $\psi_0=\psi$, so $\psi_t=e^{-ip^2t}\psi$, where $p=-i\partial_x$ and we have chosen units so that
$\mu=\tfrac{1}{2}$, as we shall do from now on for simplicity. The states with nonnegative momentum form the subspace 
\begin{equation}
    \mathscr{H}_+=\{\psi\in L^2(\mathbb{R},dx): \supp \hat{\psi}\subseteq [0,\infty)\},
\end{equation}
where 
\begin{equation}
\hat{\psi}(p) = \frac{1}{\sqrt{2\pi}}\int_{\mathbb{R}} dx\, e^{-ipx}\psi(x)
\end{equation}
is the Fourier transform of $\psi$. 

Let $t_1<t_2$. The probability backflow exhibited by the normalised $t=0$ state $\psi$ between times $t_1$ and $t_2$ is quantified by
\begin{equation}
\Delta_{\langle t_1,t_2\rangle}(\psi) := \Prob_\psi(X<0\mid t_2)-\Prob_\psi(X<0\mid t_1)  .
\end{equation}
where $\Prob_\psi(X\in S\mid t)=\int_S dx\,|\psi_t(x)|^2$ is
the probability of measuring the position of the particle to lie in $S$ at time $t$. There are simple examples of
states $\psi\in\mathscr{H}_+$ for which $\Delta_{\langle t_1,t_2\rangle}(\psi)>0$, showing that quantum mechanical particles can exhibit nontrivial backflow (see e.g.,~\cite{MR3128861}) 
whereas the classical backflow quantity for particles always lies in $[-1,0]$. 

We are interested in the range of values assumed by
$\Delta_{\langle t_1,t_2\rangle}(\psi)$ as $\psi$ varies over
normalised states of $\mathscr{H}_+$. As a difference of
probabilities, its range is contained in $[-1,1]$. Moreover,
because $\mathscr{H}_+$ is invariant under both the free time evolution and dilations, the range is invariant
under $\langle t_1,t_2\rangle\mapsto \langle t_1+\tau,t_2+\tau\rangle$ and 
$\langle t_1,t_2\rangle\mapsto \langle \lambda t_1,\lambda t_2\rangle$ for all $\tau\in\mathbb{R}$ and all $\lambda>0$.
It follows that the range is independent of both $t_1$ and $t_2$, provided $t_1<t_2$, as was observed by Bracken and Melloy~\cite{MR1280379}.
As noted in~\cite{MR2419566}, Dollard's lemma (Lemma 4 of \cite{dollard1969scattering}) shows that
$\Prob_\psi(X<0\mid t_1)\to 1$ as $t_1\to -\infty$ and
$\Prob_\psi(X<0\mid t_2)\to 0$ as $t_2\to +\infty$  
for any fixed $\psi\in \mathscr{H}_+$, thus giving
$\Delta_{\langle t_1,t_2\rangle}(\psi)\to -1$ when both limits
are taken. Thus $-1$ belongs to the closure of the 
range of $\Delta_{\langle t_1,t_2\rangle}$ for any $t_1<t_2$.

More detailed information can be found by reformulating backflow in terms of operators. For normalised $\psi\in\mathscr{H}_+$, a calculation due to Bracken and Melloy \cite{MR1280379} gives $\Delta_{\langle t_1,t_2\rangle}(\psi)$ as a quadratic form  
\begin{equation}
\Delta_{\langle t_1,t_2\rangle}(\psi) =  -\frac{1}{2\pi i}\int_0^\infty dp \: \int_0^\infty dq \:   \frac{e^{i(p^2-q^2)t_{2}}-e^{i(p^2-q^2)t_{1}}}{p-q}\hat{\psi}^\ast(p)\hat{\psi}(q),
\end{equation}
which can be written as
\begin{equation}\label{eq:Delta_from_B}
    \Delta_{\langle t_1,t_2\rangle}(\psi) = \ip{\hat{\psi}}{B_{\langle t_1,t_2\rangle} \hat{\psi}}.
\end{equation}
Here, the operator $B_{\langle t_1,t_2\rangle}$ on $L^2(\mathbb{R}^+)$ has the following properties, which are
established in Section~\ref{appendix:single_backflow} of the Supplementary Material, drawing on arguments in~\cite{MR2198971}. 
\begin{theorem}\label{thm:B_facts}
    For any $t_1<t_2$, $B_{\langle t_1,t_2\rangle}$ is
    a bounded self-adjoint operator with
    $\|B_{\langle t_1,t_2\rangle}\|=1$. There is a unitary equivalence between $B_{\langle t_1,t_2\rangle}$ and the operator $C \in\mathcal{B}(L^2(\mathbb{R}^+,dq))$ with action
   \begin{equation}\label{eq:J1def}
 (C\varphi)(p)=-\frac{1}{2\pi }\int_{0}^\infty dq \frac{
    \sin(p-q)}{p-q}\Big[\Big(\frac{p}{q}\Big)^{1/4}+\Big(\frac{p}{q}\Big)^{-1/4}\Big]\varphi(q)
\end{equation} 
on $\varphi\in L^2(\mathbb{R}^+,dq)$,
where~\eqref{eq:J1def} holds pointwise almost everywhere on $\mathbb{R}^+$. The map $(t_1,t_2)\mapsto B_{\langle t_1,t_2\rangle}$
is strongly continuous on $\{(t_1,t_2)\in\mathbb{R}^2: t_1<t_2\}$.
\end{theorem}

The unitary equivalence in Theorem~\ref{thm:B_facts} is a combination of translations, scale transformation and a change of variables.
It follows from~\eqref{eq:Delta_from_B} that the range of $\Delta_{\langle t_1,t_2\rangle}(\psi)$, as $\psi$ varies over the normalised 
elements of $\mathscr{H}_+$, is the \emph{numerical range} of the operator $B_{\langle t_1,t_2\rangle}$. We recall the definition and main properties of numerical range (see e.g., Section~9.3 of \cite{MR2359869}).
\begin{defn}
\label{defn: num_range}
    Let $A$ be a bounded self-adjoint operator on a Hilbert space $\mathscr{H}$. The \emph{numerical range} $\mathcal{N}(A) \subset \mathbb{R}$ of $A$ is given by
    \begin{equation}
        \mathcal{N}(A) = \left\{ \frac{\langle \phi | A \phi \rangle}{\langle \phi | \phi \rangle} \big| \phi \in \mathscr{H} \setminus \{0\}\right\}.
    \end{equation}
    One has $\mathcal{N}(A)=\mathcal{N}(UAU^\ast)$ for any unitary $U\in\mathcal{B}(\mathscr{H})$, and
    \begin{equation}
    \sigma(A) \subseteq \overline{\mathcal{N}(A)}=
    \Conv \sigma(A),
    \end{equation}
    where $\overline{\,\cdot\,}$ denotes topological closure and $\Conv$ the convex hull. In particular $\sigma(A)$ and
    $\mathcal{N}(A)$ have the same supremum and infimum. 
\end{defn} 
It follows immediately that, for all $t_1<t_2$,
\begin{equation} 
    \overline{\mathcal{N}(B_{\langle t_1,t_2\rangle})}
    =\Conv\sigma(C)\subseteq [-1,1],
\end{equation} 
using $\|C\|=1$. We have already noted that $-1\in \overline{\mathcal{N}(B_{\langle t_1,t_2\rangle})}$.   Defining
the \emph{Bracken--Melloy constant} as the largest spectral point of $C$,
\begin{equation}
    c_{\textnormal{BM}}=\sup \sigma (C)=\max\sigma(C),
\end{equation}
the possible values of the backflow $\Delta_{\langle t_1,t_2\rangle}(\psi)$ obey
\begin{equation}
    -1 \le \Delta_{\langle t_1,t_2\rangle}(\psi) \le c_{\textnormal{BM}}
\end{equation}
for all normalised $\psi \in \mathscr{H}_+$. 

These ideas are readily generalised to consider the total backflow over $M\ge 1$ disjoint time intervals $[t_{2j-1},t_{2j}]$ for $1\le j\le M$, which can be represented by a list $\boldvec[M]{t}=\langle t_1,...,t_{2M}\rangle$ in strictly increasing order.  
\begin{defn}
    For positive integer $M$, let
    \begin{equation}
        \mathcal{T}_M=\left\{ \boldvec[M]{t}=\langle t_1,...,t_{2M}\rangle \big| t_1<t_2<\cdot\cdot\cdot<t_{2M-1}<t_{2M}\right\}.
    \end{equation} 
\end{defn}
The total probability backflow exhibited by the time-zero state $\psi\in L^2(\mathbb{R},dx)$ over the intervals parameterised by $\boldvec[M]{t}\in\mathcal{T}_M$ is defined as
\begin{equation}
\label{deltaM_as_sum}
    \Delta^{(M)}_{\boldvec[M]{t}}(\psi) := \sum_{j=1}^M \left(\text{Prob}_\psi(X<0|t=t_{2j})-\text{Prob}_\psi(X<0|t=t_{2j-1})\right),
\end{equation}
so that, for example, $\Delta^{(1)}_{\boldvec[1]{t}}(\psi)=\Delta_{\langle t_1,t_2\rangle}(\psi)$. 

For normalised $\psi\in\mathscr{H}_+$, one has
\begin{equation}
\Delta^{(M)}_{\boldvec[M]{t}}(\psi) = \ip{\hat{\psi}}{B^{(M)}_{\boldvec[M]{t}} \hat{\psi}},
\end{equation}
where the \emph{$M$-fold backflow operator}
\begin{equation}
\label{eq: FM_as_sum}
    B^{(M)}_{\boldvec[M]{t}} = \sum_{j=1}^{M} B_{\langle t_{2j},t_{2j-1}\rangle}
\end{equation}
is a bounded self-adjoint operator on $L^2(\mathbb{R}^+, dq)$ acting as  
\begin{equation}
\label{defn:multiple_backflow_operatorF}
    \big(B^{(M)}_{\boldvec[M]{t}} \phi\big)(p) = -\frac{1}{2\pi i}\int_0^\infty dq \: \sum_{k=1}^M \frac{e^{i(p^2-q^2)t_{2k}}-e^{i(p^2-q^2)t_{2k-1}}}{p-q} \phi(q).
\end{equation} 
In the same way as before, the set of possible values taken by $\Delta^{(M)}_{\boldvec[M]{t}}(\psi)$ for normalised $\psi\in\mathscr{H}_+$ is the numerical range of $B^{(M)}_{\boldvec[M]{t}}$ and can be studied via its spectrum. As before, if $\Delta^{(M)}_{\boldvec[M]{t}}(\psi)>0$, we say that the state exhibits quantum backflow. We will also investigate states $\psi\in\mathscr{H}_+$ for which $\Delta^{(M)}_{\boldvec[M]{t}}(\psi)<-1$ for $M\geq 2$,
a phenomenon we will call \emph{quantum overflow}.  As noted in the introduction, overflow would also be interesting in the situation where the positive-momentum assumption is dropped, and perhaps exhibits a greater quantum advantage in that case, but this is not pursued here.

The following straightforward argument shows the existence of states exhibiting quantum overflow with $M=2$. Fix $t_1<t_2$ and suppose $\psi \in \mathscr{H}_+$ is such that $\Delta^{(1)}_{\langle t_1,t_2\rangle}(\psi) = \lambda >0$, i.e., there is nontrivial backflow between $t_1$ and $t_2$. Taking $T>\max(|t_1|,|t_2|)$, one can write the total backflow  over $[-T,t_1]\cup[t_2,T]$ as
\begin{equation}
    \Delta_{\langle -T,t_1,t_2,T\rangle}^{(2)}(\psi)=\Delta_{\langle -T,T\rangle}^{(1)}(\psi)-\Delta_{\langle t_1,t_2\rangle}^{(1)}(\psi).
\end{equation}
For any fixed $\epsilon\in (0,\lambda)$, one can find $T$ sufficiently large such that $\Delta^{(1)}_{\langle -T,T\rangle}(\psi)<-1+\epsilon$, because Dollard's lemma~\cite{dollard1969scattering} implies that $\Delta^{(1)}_{\langle -T,T\rangle}(\psi)\to -1$ as $T\to\infty$.
It follows that 
\begin{equation}
    \Delta_{\langle -T,t_1,t_2,T\rangle}^{(2)}(\psi)<-1-\lambda+\epsilon<-1,
\end{equation}
showing that the state $\psi$ exhibits quantum overflow on $[-T,t_1]\cup[t_2,T]$.
As $\psi$ can be chosen to make $\lambda$ arbitrarily close to $c_\textnormal{BM}$ and $\epsilon$ could be chosen arbitrarily small, this already shows that
$2$-fold overflow can approach $-1-c_\textnormal{BM}$ over suitable intervals.
In this example, we chose the time interval after fixing the state, without any control on the size of $T$.
However, in Section~\ref{spectrum}, we will consider $M$ equally spaced intervals of equal width, and show that the total $M$-fold overflow can be made unboundedly negative as $M$ increases, and similarly that the $M$-fold backflow can be made unboundedly positive as $M$ increases.

The invariance of $\mathscr{H}_+$ under time evolution and dilations does not completely fix $\mathcal{N}(B^{(M)}_{\boldvec[M]{t}})$ if $M>1$. However, time translation invariance shows that 
$\mathcal{N}(B^{(M)}_{\boldvec[M]{t}})$ is a function
of the successive differences $(t_2-t_1,t_3-t_2,\ldots, t_{2M}-t_{2M-1})$, and the scaling invariance shows that
it may be expressed in terms of successive ratios of these differences.
For each $M> 1$, then, there are \emph{backflow and overflow functions} $b^{(M)}_{\textnormal{back/over}}:\mathbb{R}_{>0}^{2M-2}\to \mathbb{R}$ (writing $\mathbb{R}_{>0}^{k}=(0,\infty)^{\times k}$) so that
\begin{align}
    \sup \mathcal{N}(B^{(M)}_{\boldvec[M]{t}}) &= 
    b^{(M)}_{\textnormal{back}}\left( \frac{t_3-t_2}{t_2-t_1},\ldots, \frac{t_{2M}-t_{2M-1}}{t_{2M-1}-t_{2M-2}}
    \right) =b^{(M)}_{\textnormal{back}}(\srsd \boldvec[M]{t}) \\
    \inf \mathcal{N}(B^{(M)}_{\boldvec[M]{t}}) &= 
    b^{(M)}_{\textnormal{over}}\left( \frac{t_3-t_2}{t_2-t_1},\ldots, \frac{t_{2M}-t_{2M-1}}{t_{2M-1}-t_{2M-2}}
    \right) = b^{(M)}_{\textnormal{over}}(\srsd \boldvec[M]{t}),
\end{align}
where $\srsd$ is the operation of forming the sequence of successive ratios of successive differences. 
By convention, we write $b^{(1)}_{\textnormal{back}}=c_{\textnormal{BM}}$, 
$b^{(1)}_{\textnormal{over}}=-1$, and $\srsd$ maps any element of $\mathcal{T}_1$ to the empty list. As we now describe, the backflow and overflow functions are semicontinuous.
\begin{lemma}\label{lem:semicontinuity}
 For each $M> 1$, the functions $b^{(M)}_{\textnormal{back/over}}$ are lower/upper semicontinuous on 
 $\mathbb{R}_{>0}^{2M-2}$. That is, for all $u_0\in\mathbb{R}_{>0}^{2M-2}$, one has
 \begin{equation}
     \liminf_{u\to u_0} b^{(M)}_{\textnormal{back}}(u)\geq b^{(M)}_{\textnormal{back}}(u_0),\qquad
     \limsup_{u\to u_0} b^{(M)}_{\textnormal{over}}(u) \le b^{(M)}_{\textnormal{over}}(u_0),
 \end{equation}
as $u\to u_0$ in $\mathbb{R}_{>0}^{2M-2}$. 
\end{lemma}
\begin{proof} 
The function $\boldvec[M]{t}\mapsto B^{(M)}_{\boldvec[M]{t}}$ is strongly continuous on $\mathcal{T}_M$ due to~\eqref{eq: FM_as_sum} and Theorem~\ref{thm:B_facts},
so for each normalised $\phi\in L^2(\mathbb{R}^+,dq)$,
$\boldvec[M]{t}\mapsto \ip{\phi}{B^{(M)}_{\boldvec[M]{t}}\phi} $ is continuous on $\mathcal{T}_M$. As 
$\boldvec[M]{t}\mapsto \sup\mathcal{N}(B^{(M)}_{\boldvec[M]{t}})= b^{(M)}_{\textnormal{back}}(\srsd\boldvec[M]{t})$ is the pointwise supremum over $\phi$ of a family
of continuous functions, it is lower semicontinuous (see, e.g., Chapter~2 of~\cite{Rudin:RealandComplex}). Fixing $\tau>0$, it follows that
$b^{(M)}_{\textnormal{back}}\circ\srsd$
is lower semicontinuous on the subset $\{\langle 0,\tau,t_3,\ldots,t_{2M}\rangle\in\mathcal{T}_M\}$ (in the relative topology). As this subset is mapped homeomorphically to $\mathbb{R}_{>0}^{2M-2}$ by $\srsd$, we conclude that $b^{(M)}_{\textnormal{back}}$ is lower semicontinuous on $\mathbb{R}_{>0}^{2M-2}$.
The proof for the overflow functions is similar. 
\end{proof}

We also introduce sequences of backflow and overflow constants, by
\begin{equation}
\label{eq: backflow_overflow_constants}
c^{(M)}_{\textnormal{back}} = \sup_{\mathbb{R}_{>0}^{2M-2}} b^{(M)}_{\textnormal{back}}, 
\qquad
c^{(M)}_{\textnormal{over}} = \inf_{\mathbb{R}_{>0}^{2M-2}} b^{(M)}_{\textnormal{over}}
\end{equation}
for $M\ge 1$, in which the Bracken--Melloy constant appears as the first backflow constant:
$c_{\textnormal{BM}}=c^{(1)}_{\textnormal{back}}$,
while the first overflow constant is
$c^{(1)}_{\textnormal{over}}=-1$. 
Just as with $c_{\textnormal{BM}}$, the functions
$b^{(M)}_{\textnormal{back/over}}$ and constants
$c^{(M)}_{\textnormal{back/over}}$ are quantities that are fixed by the free quantum dynamics on the line. However, 
as the $b^{(M)}_{\textnormal{back/over}}$ 
are dimensionless functions of dimensionless variables,
they are manifestly independent of Planck's constant. 
In the remainder of this paper, we will initiate the study of
these functions and constants using both analytical and numerical methods, starting by considering the 
spectrum of the operators $B^{(M)}_{\boldvec[M]{t}}$ for
$\boldvec[M]{t}\in\mathcal{T}_M$.

\section{The spectrum of $M$-fold backflow operators}
\label{the_spectrum}
\subsection{Bounds on the spectrum}

Our aim in this section is to obtain estimates
on $\sigma(B^{(M)}_{\boldvec[M]{t}})$ for a variety of $\boldvec[M]{t}\in \mathcal{T}_M$. The simplest estimate arises from the triangle inequality:
since $\|B_{\langle t_1,t_2\rangle}\|= 1$ for all $t_1<t_2$,
we have from \eqref{eq: FM_as_sum} that $\|B^{(M)}_{\boldvec[M]{t}}\|\le M$ and consequently
$\sigma(B^{(M)}_{\boldvec[M]{t}})\subseteq [-M,M]$ for
all $\boldvec[M]{t}\in\mathcal{T}_M$ and $M\ge 1$.

To get a better estimate of the upper bound on the spectrum, we note that any sum of self-adjoint bounded operators obeys 
\begin{equation}\label{eq:sigmasubmultiplicative}
    \sup \sigma\left(\sum_{j=1}^M A_j\right)\le \sum_{j=1}^M \sup\sigma(A_j)
\end{equation}
and therefore
\begin{equation}\label{eq:sigmaupperbound}
    \sup \sigma\big(B^{(M)}_{\boldvec[M]{t}}\big) 
\le \sum_{j=1}^M \sup \sigma( B_{\langle t_{2j-1},t_{2j}\rangle} )=  Mc_\textnormal{BM}.
\end{equation}

To get a lower bound on the spectrum, it is convenient to rewrite $\Delta^{(M)}_{\boldvec[M]{t}}(\psi)$ as 
\begin{align}
    \Delta^{(M)}_{\boldvec[M]{t}}(\psi) &= \text{Prob}_\psi(X<0|t=t_{2M}) - \text{Prob}_\psi(X<0|t=t_{1}) \nonumber\\ &\qquad - 
    \sum_{j=1}^{M-1} \left(\text{Prob}_\psi(X<0|t=t_{2j+1})-\text{Prob}_\psi(X<0|t=t_{2j})\right),
\end{align}
from which it follows that
\begin{equation}
\label{eq: operator_difference_relation}
    B^{(M)}_{\boldvec[M]{t}} = B_{\langle t_1,t_{2M}\rangle} - B^{(M-1)}_{\langle t_2,\ldots,t_{2M-1}\rangle}.
\end{equation} 
Iterating, one obtains the formula
\begin{equation} 
\label{eq: FM_alternating_representation}
    B^{(M)}_{\boldvec[M]{t}} = 
    \sum_{j=1}^M (-1)^{j-1} B_{\langle t_j,t_{2M-j+1}\rangle},
\end{equation} 
which will be used later on.

Returning to~\eqref{eq: operator_difference_relation}, an analogue of~\eqref{eq:sigmasubmultiplicative} gives
\begin{align}
    \inf \sigma\left( B^{(M)}_{\boldvec[M]{t}}\right) &\ge 
    \inf \sigma\left(B_{\langle t_1,t_{2M}\rangle}\right) + \inf \sigma\left(-B^{(M-1)}_{\langle t_2,\ldots,t_{2M-1}\rangle}\right)
    \nonumber\\
    &= -1 - \sup \sigma\left(B^{(M-1)}_{\langle t_2,\ldots,t_{2M-1}\rangle}\right) \nonumber\\
    &\ge -1-(M-1) c_\textnormal{BM},
\end{align}
where we have used~\eqref{eq:sigmaupperbound}. Summarising, we have shown the following.
\begin{theorem}
\label{thm: spectral_outer_bounds}
For any $M\in \mathbb{N}$ and $\boldvec[M]{t}\in\mathcal{T}_M$, the spectrum of $B^{(M)}_{\boldvec[M]{t}}$ satisfies
\begin{equation}
    \sigma(B^{(M)}_{\boldvec[M]{t}})\subseteq [-1-(M-1)c_\textnormal{BM}, Mc_\textnormal{BM}].
\end{equation} 
Thus one has bounds 
\begin{align}
    b^{(M)}_{\textnormal{back}}(u_1,\ldots,u_{2M-2}) & \le c^{(M)}_{\textnormal{back}}  \le Mc_\textnormal{BM} \\
    \label{eq:b_over_bounds}
b^{(M)}_{\textnormal{over}}(u_1,\ldots,u_{2M-2}) & \ge c^{(M)}_{\textnormal{over}} \ge -1-(M-1) c_\textnormal{BM}
\end{align}
for all $(u_1,\ldots,u_{2M-2})\in\mathbb{R}_{>0}^{2M-2}$.
\end{theorem}

Theorem~\ref{thm: spectral_outer_bounds} shows that 
$c^{(M)}_{\textnormal{back/over}}$ do not grow faster than linearly in $M$, but does not answer the question of whether they grow at all. 
In the following subsections we will prove rigorously that for a particular class of backflow operators, both the infimum and supremum of $\sigma(B^{(M)}_{\boldvec[M]{t}})$ tend to $\pm \infty$ as $M \rightarrow \infty$ at least as fast as $\mathcal{O}(M^{1/4})$. This is in stark contrast to the classical case where, irrespective of the number of disjoint backflow intervals, the total backflow always lies in $[-1,0]$. The rigorous results presented will be supplemented by a numerical investigation for values of $1 \leq M\leq 4$ which clearly indicate monotonicity of the infimum and supremum of $\sigma(B^{(M)}_{\boldvec[M]{t}})$. 

\subsection{Unbounded $M$-fold backflow and overflow as $M\to\infty$}
\label{spectrum}

In this subsection we will show that
$c^{(M)}_\textnormal{back}$ and $-c^{(M)}_\textnormal{over}$ grow 
unboundedly with $M$, by considering the $M$-fold backflow operator corresponding to $M$ backflow periods of equal duration, separated by $M-1$ intervals of the same duration. Specifically, for $M \in \mathbb{N}$ and fixed $T>0$ consider the $M$-fold backflow problem for the times
\begin{equation}
\label{eq: B_times}
    \langle \boldsymbol{t}_M\rangle=\left\langle -T\frac{2M-1}{2},-T\frac{2M-3}{2},...,-\frac{T}{2},\frac{T}{2},...,T\frac{2M-1}{2}\right\rangle,
\end{equation}
for which the total backflow is bounded between $b^{(M)}_{\textnormal{over}}(1,\ldots,1)$ and $b^{(M)}_{\textnormal{back}}(1,\ldots,1)$. Note that $b_\textnormal{over}^{(M)}(1,...,1)$ and $b_\textnormal{back}^{(M)}(1,...,1)$ are invariant under the choice of $T>0$.
One can write the associated $M$-fold backflow operator $B^{(M)}_{\langle \boldsymbol{t}_M\rangle}$ as
\begin{equation}
    B^{(M)}_{\langle \boldsymbol{t}_M\rangle}=\sum_{j=1}^M (-1)^{j-1} B^{(1)}_{\langle -T(j-1/2),T(j-1/2)\rangle}
\end{equation}
using  \eqref{eq: FM_alternating_representation}. Each of the $B^{(1)}_{\langle -T(j-1/2),T(j-1/2)\rangle}$ operators has the closed form 
\begin{equation}
    (B^{(1)}_{\langle -T(j-1/2),T(j-1/2)\rangle} \phi)(k)=-\frac{1}{\pi}\int_0^\infty dl \: \frac{\sin\left(T(j-1/2)(k^2-l^2)\right)}{k-l}\phi(l).
\end{equation}
Let $U: L^2(\mathbb{R}^+, dk) \rightarrow L^2(\mathbb{R}^+, dp)$ be the unitary implementing the change of variables $p=k^2T/2$,
$(U\phi)(p) = (2pT)^{-1/4} \phi(\sqrt{2p/T})$,
so that $U^\ast B^{(1)}_{\langle -T/2,T/2\rangle} U=C$ is the operator defined in~\eqref{eq:J1def}. 
Then it is easily seen, by use of the identity 
\begin{equation}
    \sum_{j=1}^M (-1)^{M-j}\sin(2(j-1/2)(p-q))=\frac{\sin 2M(p-q)}{2 \cos(p-q)},
\end{equation}
that the operator $C^{(M)} = U^\ast B^{(M)}_{\langle \boldsymbol{t}_M\rangle} U\in \mathcal{B}(L^2(\mathbb{R}^+,dq))$ is independent of $T$ and has the action
\begin{equation}\label{eq:CMdefn}
    \left(C^{(M)}\varphi\right)(p)=-\frac{1}{4\pi}\int_0^\infty dq \: \frac{\sin(2M[p-q])}{(p-q)\cos(p-q)}\left[\left(\frac{p}{q}\right)^{1/4}+\left(\frac{p}{q}\right)^{-1/4}\right] \varphi(q).
\end{equation}
The following theorem shows that the suprema and infima of $\mathcal{N}(B^{(M)}_{\langle \boldsymbol{t}_M\rangle})=\mathcal{N}(C^{(M)})$ grow arbitrarily large in magnitude with $M$
at least as fast as $\mathcal{O}(M^{1/4})$.
\begin{theorem}
\label{thm: spectral_asymptotic_unboundedness}
For $M \in \mathbb{N}$, let $C^{(M)}\in \mathcal{B}(L^2(\mathbb{R}^+,dq))$ be the backflow operator defined in~\eqref{eq:CMdefn} describing $M$ backflow periods of unit duration separated by $M-1$ periods of unit duration. Then
there is a constant $k>0$ so that    
\begin{equation}\label{eq:supbound}
    \sup \sigma\left(C^{(M)}\right) \geq k M^{1/4}
\end{equation}
for all positive even $M$, and furthermore,
\begin{equation}\label{eq:spectral_asymptotics}
   \liminf_{M\to\infty} M^{-1/4} \sup \sigma\left(C^{(M)}\right) \geq k>0, \qquad 
   \limsup_{M\to\infty} M^{-1/4} \inf \sigma\left(C^{(M)}\right) \leq -k<0,
\end{equation}
where the limits are taken over all $M\in\mathbb{N}$.
\end{theorem}
\begin{proof} 
For even $M\in\mathbb{N}$ and $\epsilon \in (0,\tfrac{\pi}{6}]$, define the intervals $\Ic_0, \Ic_1\subseteq\mathbb{R}$ by
    \begin{equation}
        \Ic_0 = \left[0, \frac{\epsilon}{M}\right], \qquad \Ic_1 = \left[\frac{\pi}{2}-\frac{\epsilon}{2M}, \frac{\pi}{2}+\frac{\epsilon}{2M}\right]
    \end{equation}
    and associated normalised $\psi_M\in L^2(\mathbb{R}^+, dq)$ by
    \begin{equation}\label{eq:psiM}
        \psi_M(q) = \begin{cases}
            \sqrt{\frac{M}{2\epsilon}} \qquad& q \in \Ic_0 \cup \Ic_1\\
            0 \qquad& \text{otherwise}
        \end{cases}.
    \end{equation}
   The supremum of the spectrum is estimated using the bound $\sup \sigma(C^{(M)})\ge \left\langle \psi_M | C^{(M)} \psi_M \right\rangle$. This is computed in Supplementary Section~\ref{appendix: Lemma3.3Proof}, which also details the proofs of~\eqref{eq:supbound} and~\eqref{eq:spectral_asymptotics}.
\end{proof}
Note that \eqref{eq:spectral_asymptotics} are poor bounds on the spectra of $C^{(M)}$. However, they show that simple states can deliver arbitrarily large amounts of backflow and overflow. Theorem~\ref{thm: spectral_asymptotic_unboundedness} has 
an immediate corollary relating to the asymptotics of the backflow and overflow constants.
\begin{corollary}
\label{cor: unbounded_backflow_overflow_consts}
Let $c^{(M)}_\textnormal{back}$ and $c^{(M)}_\textnormal{over}$ be the $M$-fold backflow and overflow constants defined in \eqref{eq: backflow_overflow_constants}. Then
\begin{equation}
    \liminf_{M \to \infty} M^{-1/4} c_\textnormal{back}^{(M)} > 0, \qquad \limsup_{M \to \infty} M^{-1/4} c^{(M)}_\textnormal{over} < 0.
\end{equation}
\end{corollary}
\begin{proof}
    Since $C^{(M)}$ is unitarily equivalent to $B^{(M)}_{\langle\boldsymbol{t}_M(T)\rangle}$, we find that
    \begin{equation}
        \sup \sigma \left(C^{(M)}\right) = b_\textnormal{back}^{(M)}(1,...,1)\leq c^{(M)}_\textnormal{back}, \qquad \inf \sigma \left(C^{(M)}\right) = b_\textnormal{over}^{(M)}(1,...,1)\geq c^{(M)}_\textnormal{over}.
    \end{equation}
    The result follows on multiplying each inequality by $M^{-1/4}$ and employing 
    lemma~\ref{thm: spectral_asymptotic_unboundedness}, 
    together with elementary properties of $\limsup$ and $\liminf$.
\end{proof}

Theorem~\ref{thm: spectral_asymptotic_unboundedness} and Corollary~\ref{cor: unbounded_backflow_overflow_consts} demonstrate the unboundedness of the quantum backflow and overflow effects as $M\to\infty$. In Section~\ref{section: numerical_calculation}, we will find numerical estimates of the spectral extrema for $M$ equally spaced intervals of equal width, with $2 \leq M\leq 4$, showing in particular that
\begin{equation}
    \sup \left(C^{(M)}\right)>c_\textnormal{BM}, \qquad \inf \left(C^{(M)}\right)<-1
\end{equation}
and further give numerical evidence for lower bounds on $c_\textnormal{back}^{(M)}$ as well as plots of vectors in $L^2(\mathbb{R}^+,dp)$ exhibiting
backflow and overflow.

\section{Limiting cases and monotonicity of the backflow constants}
\label{sec:limiting cases}

In Lemma~\ref{lem:semicontinuity} we have already seen that 
the $M$-fold backflow and overflow functions are semicontinuous for limits taken within $\mathbb{R}^{2M-2}_>$. We now consider certain limiting cases, which show how backflow and overflow functions with different numbers of parameters are related and establish monotonicity properties of the backflow constants.
\begin{theorem}\label{thm:blims}
    Let $M\in\mathbb{N}$ and $j,k\in\mathbb{N}_0$ so that $j+k=2L$ is even and positive. 
    Consider any $u\in\mathbb{R}^{(2M-2)}_>$ and any
    sequences $(v_n)_{n\in\mathbb{N}}$  in $\mathbb{R}_>^j$ and $(w_n)_{n\in\mathbb{N}}$ in $\mathbb{R}_>^k$ whose last and first components obey $v_{n,j}\to 0$ and $w_{n,1}\to \infty$. (In the case $M=1$ one omits $u$; similarly, in the cases $j=0$ resp., $k=0$ one omits $v$ resp., $w$.)
    
    If $j$ (and hence also $k$) is even, then 
    \begin{equation}\label{eq:blimseven}
        \liminf_n b^{(M+L)}_{\textnormal{back}}(v_n,u,w_n) \geq b^{(M)}_{\textnormal{back}}(u), \qquad
        \limsup_n b^{(M+L)}_{\textnormal{over}}(v_n,u,w_n) \leq b^{(M)}_{\textnormal{over}}(u)
    \end{equation}
    where we use the shorthand notation
    $(v,u,w)=(v_1,\ldots,v_j,u_1,\ldots, u_{2M-2},w_1,\ldots,w_k)$.
     
    On the other hand, if $j$ (and hence also $k$) is odd, then 
    \begin{equation}\label{eq:blimsodd}
        \liminf_n b^{(M+L)}_{\textnormal{back}}(v_n,u,w_n) \geq -1- b^{(M)}_{\textnormal{over}}(u), \qquad
        \limsup_n b^{(M+L)}_{\textnormal{over}}(v_n,u,w_n) \leq -1-b^{(M)}_{\textnormal{back}}(u).
    \end{equation}
    Consequently, the sequence of backflow constants $c^{(M)}_{\textnormal{back}}$ 
    is nondecreasing, and the sequence of overflow constants $c^{(M)}_{\textnormal{over}}$ is nonincreasing; moreover,
    $c^{(M+L)}_\textnormal{back}\geq -1 - c^{(M)}_\textnormal{over}$ and
    $c^{(M+L)}_\textnormal{over}\leq -1 - c^{(M)}_\textnormal{back}$ for all $L\in\mathbb{N}$.    
\end{theorem}
We remark that the result holds also, with the same proof, if the sequences $v_n$ and $w_n$ are replaced by nets. The proof can be found in Supplementary Section~\ref{blims_appendix}.

Theorem~\ref{thm:blims} has a number of consequences. First, the bound
$c^{(2)}_\textnormal{over}\leq -1-c^{(1)}_\textnormal{back}$ can be combined with~\eqref{eq:b_over_bounds} to give the sharp value
\begin{equation}\label{eq:c2over}
    c^{(2)}_\textnormal{over} = -1-c_{\textnormal{BM}}.
\end{equation}
Second, as the overflow constants are nonincreasing, we see that
for every $M\ge 2$ and $\epsilon>0$ there exists
$\boldvec[M]{t}\in\mathcal{T}_M$ and $\psi\in\mathscr{H}_+$ so that 
\begin{equation}
\Delta^{(M)}_{\boldvec[M]{t}}(\psi) \le -1-c_{\textnormal{BM}}+\epsilon
\end{equation}
which (taking $0<\epsilon<c_{\textnormal{BM}}$) provides another proof of the existence of overflow states for any $M\ge 2$. Similarly, 
we may deduce that there are backflow states for every $M\in\mathbb{N}$.

Third, the proof of Theorem~\ref{thm:blims} can be modified slightly
to give the following result (many other variations are possible) which may be useful for further studies. The notation $\Li$ will be explained after the statement, which is also proved in Supplementary Section~\ref{blims_appendix}.
\begin{theorem}\label{thm:PKlimits}
    For $M\in\mathbb{N}$ and $\boldvec[M]{t} \in \mathcal{T}_M$, one has 
    \begin{equation}
         -1-\mathcal{N}\left(B^{(M)}_{\boldvec[M]{t}}\right) \subseteq \Li_{T_\pm \rightarrow \pm\infty } \mathcal{N}\left(B^{(M+1)}_{\langle T_-,\boldsymbol{t}_M,T_+\rangle}\right)
    \end{equation}
    and 
    \begin{equation}
        \mathcal{N}\left(B^{(M)}_{\boldvec[M]{t}}\right) \subseteq 
        \Li_{\substack{T,T'\rightarrow -\infty\\ T<T'} } \mathcal{N}\left(B^{(M+1)}_{\langle T,T',\boldsymbol{t}_M\rangle}\right).
    \end{equation} 
\end{theorem} 
Here, $\Li$ denotes the \emph{Painlev\'{e}-Kuratowski lower closed limit}, see e.g.~chapter~29 of \cite{MR217751} and chapter~5 of \cite{MR1269778}, defined as follows.
\begin{defn}
\label{defn: Li_defn}
    Given a net of sets $(A_\alpha)_{\alpha \in I}$ in a topological space $X$, where $I$ is a directed set, $\Li A_\alpha \subseteq X$ is the set of points $p\in X$ with the property that, for every neighbourhood $N$ of $p$, there exists $\alpha_0 \in I$ with $A_\alpha \cap N \neq \emptyset$ for all $\alpha > \alpha_0$,
    i.e., every neighbourhood of $p$ is eventually intersected by $A_\alpha$.
\end{defn}

A fourth consequence of Theorem~\ref{thm:blims} is that it provides an understanding of how the spectrum of an two-fold backflow operator 
$B^{(2)}_{\langle t_1,t_2,t_3,t_4\rangle}$  behaves in a limit 
in which the two backflow intervals merge together, i.e., $t_2,t_3\to \tau$ for some
$\tau\in (t_1,t_4)$, with $t_1,t_4$ held fixed and maintaining $t_2<t_3$. An obvious question is whether the spectrum converges in some sense to
that of $B^{(1)}_{\langle t_1,t_4\rangle}$. As we now show, this is not the case. 
Writing $\srsd \langle t_1,t_2,t_3,t_4\rangle=(v,w)$, we have
$v\to 0+$ and $w\to +\infty$ in the limit.
Using the $M=j=k=1$ case of Theorem~\ref{thm:blims}, we find
\begin{equation}
    \limsup  b^{(2)}_{\textnormal{over}}(v ,w ) \le -1- b^{(1)}_{\textnormal{back}} = -1-c_{\textnormal{BM}}
\end{equation}
from~\eqref{eq:blimsodd} in the limit of interest.
Hence $b^{(2)}_{\textnormal{over}}(v ,w )\to -1-c_{\textnormal{BM}}$ because we also have $b^{(2)}_{\textnormal{over}}\ge c^{(2)}_{\textnormal{over}}=-1-c_\textnormal{BM}$. 
Thus, $\sigma(B^{(2)}_{\langle t_1,t_2,t_3,t_4\rangle})$ contains points 
approaching $-1-c_{\textnormal{BM}}$ arbitrarily closely in the limit, while
$\sigma(B^{(1)}_{\langle t_1,t_2\rangle})\subseteq [-1,c_\textnormal{BM}]$.
Therefore the spectra do not coincide in the limit; in fact, we have  
\begin{equation}\label{eq:limiting_Haus_dist}
   \liminf_{\substack{t_2,t_3\to \tau\\t_2<t_{3}}}   d_\textnormal{Haus}\left(\sigma(B^{(2)}_{\langle t_1,t_2,t_3,t_4\rangle}), \sigma(B^{(1)}_{\langle t_1,t_4\rangle})\right)\geq c_{\textnormal{BM}},
\end{equation}
where the Hausdorff distance between compact $A,B \subset \mathbb{R}$ is
\begin{equation}
    d_\textnormal{Haus}(A,B)=\max \{\sup_{a \in A} \inf_{b \in B} |a-b|, \sup_{b \in B} \inf_{a \in A} |a-b| \},
\end{equation} 
and gives the collection of non-empty compact subsets of $\mathbb{R}$ the structure of a complete metric space (see, e.g., theorem 3.2.4 of \cite{MR1269778}).
As one has the inequality $d_\textnormal{Haus}(\sigma(A),\sigma(B)) \leq \|A-B\|$ 
for self-adjoint bounded operators on Hilbert spaces (see e.g., Theorem V.4.10 in \cite{Kato:1980}),  we may infer that the backflow operators $B^{(2)}_{\langle t_1,t_2,t_3,t_4\rangle}$ do not converge in norm to $B^{(1)}_{\langle t_1,t_4\rangle}$ in the limit considered. This can also be seen directly as follows, noting that 
\begin{equation}
\label{eq: double_backflow_decomp}
    B^{(2)}_{\langle t_1,t_2,t_3,t_4\rangle} - B^{(1)}_{\langle t_1,t_4\rangle} = -B^{(1)}_{\langle t_2,t_3\rangle}
\end{equation}
and therefore 
\begin{equation}\label{eq:B2vsB1}
    \|B^{(2)}_{\langle t_1,t_2,t_3,t_4\rangle} - B^{(1)}_{\langle t_1,t_4\rangle}\| = 1
\end{equation}
by recalling that all single backflow operators have unit norm. 
It would be interesting to study other mergers of multiple backflow intervals in a similar way. 

We end this discussion with a cautionary note. For $\Lambda>0$, let $L^2([0,\Lambda],dk)$ be the space of positive momentum states with momentum cutoff $\Lambda$, in momentum representation. Let $\iota_\Lambda:L^2([0,\Lambda],dk)\to L^2(\mathbb{R}^+,dk)$ be the subspace inclusion map, 
whereupon $\iota_\Lambda^*$ is the subspace projection. 
Lemma~\ref{lem:cutoff_norm_continuity} in Supplementary Section~\ref{appendix:single_backflow} shows that the operators 
$\iota_\Lambda^* B^{(1)}_{\langle s,t\rangle}\iota_\Lambda$ are 
norm continuous in $\langle s,t\rangle\in\mathcal{T}_2$ with respect to 
the operator norm $\|\cdot\|_\Lambda$ on $L^2([0,\Lambda],dk)$.
Considering~\eqref{eq: double_backflow_decomp} again, it follows that
\begin{equation}
     \|\iota_\Lambda^*(B^{(2)}_{\langle t_1,t_2,t_3,t_4\rangle} - B^{(1)}_{\langle t_1,t_4\rangle})\iota_\Lambda\|_\Lambda \longrightarrow 0
\end{equation}
as the backflow intervals merge, so
\begin{equation}
    \lim_{t_3-t_2 \rightarrow 0^+} d_\textnormal{Haus}\Big(\sigma\Big(\iota_\Lambda^*  B^{(2)}_{\langle t_1,t_2,t_3,t_4\rangle }\iota_\Lambda\Big),\sigma\Big(\iota_\Lambda^*  B^{(1)}_{\langle t_1,t_4\rangle}\iota_\Lambda\Big)\Big) \longrightarrow 0 
\end{equation}
in this limit. Accordingly, any numerical scheme implementing a fixed momentum cutoff will give the erroneous impression of a convergence of spectra as backflow intervals merge. To see the true behaviour of the backflow operators as in~\eqref{eq:limiting_Haus_dist}, it would therefore be necessary to increase the momentum cutoff as $t_3-t_2$ decreases.

\section{Numerical calculation}
\label{section: numerical_calculation}

The goal of this section is to numerically investigate the multiple quantum backflow and quantum overflow effect over $M$ disjoint intervals of equal duration separated by gaps of the same length, corresponding to the bounded operator $C^{(M)}$ given in \eqref{eq:CMdefn}. Although our method would in principle apply to arbitrary $M$, the computational effort rises quickly with $M$ and we have chosen to restrict to $M\leq 4$.
The numerical results give lower bounds on the magnitudes of $b^{(M)}_{\textnormal{back/over}}(1,\ldots,1)$ for $1\leq M\leq 4$. 
In particular, we will find for $2\leq M\leq 4$ that $\max\sigma \left(C^{(M)}\right) > c_\textnormal{BM}$ and $\min\sigma\left(C^{(M)}\right) < -1$. This complements the analytic results in theorem~\ref{thm: spectral_asymptotic_unboundedness} and Corollary~\ref{cor: unbounded_backflow_overflow_consts} on the large $M$ asymptotics of $\sigma(C^{(M)})$ and the mononotonicity results in Theorem~\ref{thm:blims}. By numerical acceleration methods, we will give improved estimates for $b^{(M)}_{\textnormal{back/over}}(1,\ldots,1)$ for $1\le M\le 4$. In addition, we will investigate the properties of states that come close to maximising the backflow or overflow for these operators.

\subsection{Background theory} 
\label{exact_matrix_elements}

Our numerical calculations are based on the following basic observation. 
(See Section XIII.1 of~\cite{ReedSimon:vol4} for a discussion of other min-max results
and Theorem~VII.12 of~\cite{ReedSimon:vol1} for Weyl's criterion.) Here, $\sigma(B,Q)=\{\lambda\in\mathbb{R}:\det(B-\lambda Q)=0\}$ is the set of generalised eigenvalues for Hermitian matrix $B$ with respect to a positive definite matrix $Q$ of the same dimension. 
\begin{prop}
\label{prop: spectral_recipe}
Let $A$ be a bounded self-adjoint operator on Hilbert space $\mathscr{H}$ and let $(\chi_n)_{n \in \mathbb{N}} \subset \mathscr{H}$ be a sequence of linearly independent vectors with dense span. Define sequences of self-adjoint $N\times N$ matrices 
$(A^{[N]})_{N\in\mathbb{N}}$ and $(P^{[N]})_{N\in\mathbb{N}}$
with matrix elements
\begin{equation}
    A^{[N]}_{mn}=\ip{\chi_m}{A\chi_n} , \qquad P^{[N]}_{mn}=\ip{\chi_m}{\chi_n}
\end{equation} 
for $1\leq m,n\leq N$. Then $\sigma(A^{[N]},P^{[N]})\subseteq \mathcal{N}(A)$, and 
$\max  \sigma(A^{[N]},P^{[N]})$ (resp., $\min  \sigma(A^{[N]},P^{[N]})$) is a bounded nondecreasing (resp., nonincreasing) sequence with
\begin{align}
        \max \sigma(A) &= \lim_{N \rightarrow \infty} \max \sigma(A^{[N]},P^{[N]}) \nonumber\\ 
        \min \sigma(A) &= \lim_{N \rightarrow \infty} \min \sigma(A^{[N]},P^{[N]}).
    \end{align}
For each $N \in \mathbb{N}$, suppose $v^{(N)} \in \mathbb{C}^N$ is a generalised eigenvector obeying
\begin{equation}
    A^{[N]}v^{(N)} = \lambda_N P^{[N]}v^{(N)}, \qquad v^{(N)\dagger} P^{[N]} v^{(N)}=1
\end{equation}
and define $\psi^{(N)} \in \mathcal{H}$ by
\begin{equation}
    \psi^{(N)}=\sum_{n=1}^N v^{(N)}_n \chi_n.
\end{equation}
If $\psi^{(N)}\to \psi \in \mathcal{H}$ in norm then
\begin{enumerate}
        \item \label{st: gen_eigvals}the sequence of generalised eigenvalues $(\lambda_N)_{N \in \mathbb{N}}$ converges;
        \item \label{st: gen_eigvecs}the sequence of vectors $(\psi^{(N)})_{N \in \mathbb{N}}$ is a Weyl sequence for $\lambda = \lim_{N \to \infty} \lambda_N$, i.e., $\|\psi^{(N)}\|=1$ and $\|(A-\lambda I)\psi^{(N)}\|\to 0$;
        \item \label{st: eigenvector}the limiting vector $\psi$ is an eigenvector for $A$ with eigenvalue $\lambda$.
    \end{enumerate}
\end{prop}
The proof of Proposition~\ref{prop: spectral_recipe} may be found in Supplementary section~\ref{spectral_recipe_appendix}. The generalised eigenvalue problems in Proposition~\ref{prop: spectral_recipe} reduce to standard eigenvalue problems if the trial vectors $\chi_n$ are orthonormal. 
 In principle one could always orthogonalise to express the matrix elements in a Gram-Schmidt basis, but this introduces a nontrivial computational overhead and
 the generalised eigenproblem may be preferred.

\subsection{Numerical methodology and results}\label{sec:numerical_methodology_results}

We will apply Proposition~\ref{prop: spectral_recipe} to the operators $C^{(M)}$, but our methodology also applies to general backflow operators. 
Given $\delta>-1/2$ and $a>0$, define the sequence of normalized $L^2(\mathbb{R}^+, dq)$ vectors $(\psi_{n,a,\delta})_{n=0}^\infty$ by
\begin{equation}
\label{eq: almost_polynomial_vectors}
    \psi_{n,a,\delta}(q) = E_n(a,\delta) q^{n+\delta}e^{-aq},
\end{equation}
with the normalization constant $E_n(a,\delta)$ given by
\begin{equation}
\label{eq: normalisation_const}
    E_n(a,\delta)=\frac{(2a)^{n+\delta+1/2}}{\sqrt{\Gamma(2n+2\delta+1)}}.
\end{equation}
The following density result is found in Theorem 5.7.1 of \cite{MR0372517} in the case $a=1/2$, and follows for general $a>0$ follows by a unitary scale change. 
\begin{lemma}
\label{lemma: density_lemma}
    For $a>0$ and $-\tfrac{1}{2}<\delta<\tfrac{1}{2}$, the sequence $(\psi_{n,a,\delta})_{n =0}^\infty$ has a dense span in $L^2(\mathbb{R}^+, dq)$.
\end{lemma}

To apply Proposition~\ref{prop: spectral_recipe} to $C^{(M)}$, we require the matrix elements 
\begin{equation}
\label{M_fold_finite_matrices}
    C^{(M)}(a,\delta)_{mn}=\ip{ \psi_{m,a,\delta}}{C^{(M)} \psi_{n,a,\delta} }, \qquad P(a,\delta)_{mn}=\ip{ \psi_{m,a,\delta} }{\psi_{n,a,\delta}}
\end{equation}
for $m,n\in\mathbb{N}_0$.
The components $P(a,\delta)_{mn}$ admit the closed form
\begin{equation}\label{eq:Pdef}
    P(a,\delta)_{mn}=\frac{\sqrt{B_\textnormal{diag}(m+\delta+1/2)B_\textnormal{diag}(n+\delta+1/2)}}{B(m+\delta+1/2,n+\delta+1/2)},
\end{equation}
where $B(x,y)=\Gamma(x)\Gamma(y)/\Gamma(x+y)$ is the beta function and $B_\textnormal{diag}(x)=B(x,x)=\Gamma(x)^2/\Gamma(2x)$.
The components $C^{(M)}(a,\delta)_{mn}$ (or indeed, those of general multiple backflow operators) also have a closed form expression in terms of incomplete Beta functions, for which the reader is referred to Supplementary Section~\ref{appendix:exactmatrixelements}. For each $M$, and for fixed $a>0$, $\delta\in (-1/2,1/2)$, the matrix elements $P(a,\delta)_{mn}$ and
$C^{(M)}(a,\delta)_{mn}$ can be computed for $0\leq m,n\leq N$, giving  $(N+1)$-dimensional square matrices. The matrix elements must be computed to high precision, for which purpose we use the package FLINT~\cite{FLINT,MR4000439}. For a discussion of the error analysis of the generalised eigenvalues, the reader is guided to Supplementary Section~\ref{appendix: error analysis}.

Further experimentation was used to find the values of $a_M,\delta_M$ that appear to produce the best, i.e., approximately largest magnitude, values of $\lambda^{(M)}_{\textnormal{back/over}}(a,\delta;N)$ for given $M$ across the range of $N$ studied, resulting in choices $a_M=2M/\pi$ and $\delta_M=-1/4$. (We do not claim that these values are precisely optimal.) Our choice relies on the following observations. Each $\psi_{n,a,\delta}(q)$ has a unique global maximum at $q=a^{-1}(n+\delta)$ and hence for each number of backflow periods $M$, we expect that a good choice of $a_M^{-1}$ would be comparable with the gap between consecutive peaks of the top approximate eigenvector $\psi^{(M)}_{\textnormal{back}}(a,\delta;N)$. Early numerical results suggested that for $M=1$, this gap tends to $\pi$ as $N$ increases and the best numerical results were found when using $a_1=2/\pi$. For larger values of $M$, the top approximate eigenvector has $M$ times as many stationary points (see Figures~\ref{fig: psiMAXM2},~\ref{fig: psiMAXM3} and~\ref{fig: psiMAXM4} below) and so we selected $a_M=Ma_1=2M/\pi$. To motivate the choice $\delta_M=-1/4$, consider the action of a vector of the form $\psi(q)=q^{-1/4}f(q)$ under $C^{(M)}$, we find
\begin{equation}
    (C^{(M)}\psi)(p)=-\frac{p^{-1/4}}{4\pi}\int_0^\infty dq\: \frac{\sin(2M(p-q))}{(p-q)\cos(p-q)}f(q)-\frac{p^{1/4}}{4\pi}\int_0^\infty dq\: \frac{\sin(2M(p-q))}{(p-q)\cos(p-q)}q^{-1/2}f(q),
\end{equation}
suggesting that any eigenfunction of $C^{(M)}$ is likely to diverge as $\mathcal{O}(p^{-1/4})$ as $p\to 0^+$. We define $\lambda_\textnormal{back/over}^{(M)}(N)=\lambda_\textnormal{back/over}^{(M)}(a_M,-\tfrac{1}{4};N)$.

We present numerical results for $M \in \{1,2,3,4\}$ in graphical form in Figures~\ref{fig:back} and~\ref{fig:over} for $100 \leq N \leq 500$ and tabulate the values of $\lambda^{(M)}_{\text{back/over}}(500)$ in Table~\ref{tab:N500values}. The values obtained for all $N$ stored at $150$ digits may be found in the research data repository of the University of York~\cite{backflow:dataset}. Note that the $N^\text{th}$ eigenvalue should only be trusted to around $503-0.91N$ decimal places. The plots display the values for $N\geq 100$, over which range they increasingly resemble smooth curves. Each separate eigenvalue sequence is monotone and would tend to the maximum or minimum of the appropriate spectrum $\sigma(C^{(M)})$ as $N\to\infty$. We will discuss the convergence rates and estimated limits with more detail in the next subsection, but it is already clear that the $M=1$ values stay below the previously estimated value of the Bracken--Melloy constant $c_\textnormal{BM}\approx 0.03845$~\cite{MR2119354,MR2419566}, that one has $c^{(M)}_\textnormal{back}\geq b_\textnormal{back}^{(M)}(1,...,1)>c_\textnormal{BM}$ and 
$c^{(M)}_\textnormal{over}\leq b_\textnormal{over}^{(M)}(1,...,1)<-1$
for $2\leq M\leq 4$, and that the value of
$b^{(2)}_\textnormal{over}(1,1)$ appears to be much closer to $-1$ than to the $2$-fold overflow constant $c^{(2)}_{\textnormal{over}}=-1-c_\textnormal{BM}$. 

\begin{figure}
\centering
    \includegraphics[width=\textwidth]{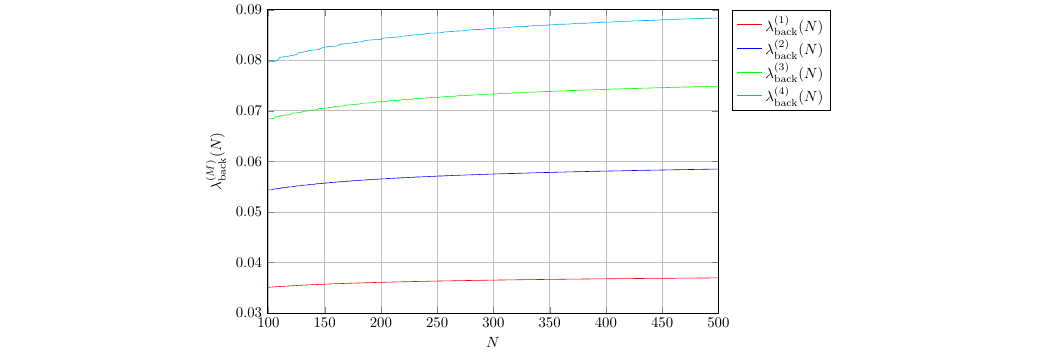}
    \caption{Plot of $\lambda^{(M)}_\textnormal{back}(N)$ for $100 \leq N \leq 500$.}
    \label{fig:back}
\end{figure}

\begin{figure}
\centering
    \includegraphics[width=\textwidth]{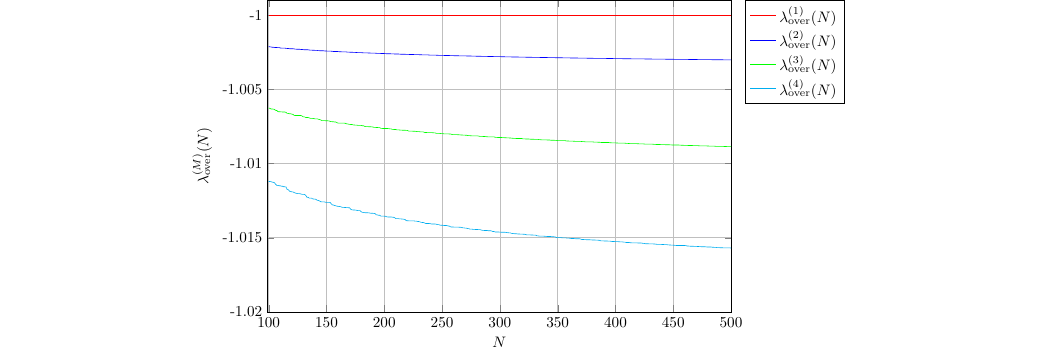}
    \caption{Plot of $\lambda_\textnormal{over}^{(M)}$ for $100 \leq N \leq 500$.}
    \label{fig:over}
\end{figure}

As well as estimates for the backflow and overflow values of $C^{(M)}$, we also computed the associated approximate eigenvectors as follows. As in Proposition~\ref{prop: spectral_recipe}, each vector $v \in \mathbb{R}^{N+1}$ with $v^\dagger P^{[N]}(a,\delta)v=1$ defines an
approximate eigenvector $\psi\in L^2(\mathbb{R}^+)$ for
$C^{(M)}$,
\begin{equation}
\label{eq: eigenvector_generation}
    \psi = \sum_{n=0}^N v_n \psi_n(a,\delta; \cdot).
\end{equation}

Define $\psi_\textnormal{back/over}^{(M)}(a,\delta;N;\cdot)$ as the $L^2(\mathbb{R}^+)$ vector given by applying \eqref{eq: eigenvector_generation} to the generalised eigenvectors associated with the largest/smallest generalised eigenvalues of $(C^{(M)[N]}(a,\delta),P^{[N]}(a,\delta))$. Figure~\ref{fig: psiMAXM1_4p0_10} shows the vectors $\psi_\textnormal{back}^{(M)}(a_M,-1/4;500;p)$ for $p \in [0,10]$. As can be seen, all the $M$-fold backflow maximizing vectors appear to diverge as $p\to 0$. In line with $\delta=-1/4$ giving the best approximate eigenvalue results, we conjecture that all of the $M$-fold backflow maximizing vectors behave like $p^{-1/4}$ for $p\sim 0$.  Recall from Section~\ref{spectrum} that the momentum space wavefunction is obtained by applying the unitary $U^*:L^2(\mathbb{R}^+,dp)\to L^2(\mathbb{R}^+,dk)$ implementing a change of variables $p=k^2/2$, where we set the unit of time equal to the duration and spacing of the backflow intervals and use units of mass and length so that $m=1/2$ and $\hbar=1$. As $(U^*\psi)(k) =\sqrt{k} \psi(k^2/2)$, our conjecture is that the momentum wavefunctions of $M$-fold backflow maximizing vectors have a finite limits as $k\to 0+$.

Figures~\ref{fig: psiMAXM1}--\ref{fig:psi_over_1} show the approximate backflow and overflow eigenvectors for $1 \leq M \leq 4$  as functions of $p$. A common feature is that the envelope of the wavefunction decays as $p^{-3/4}$ for large $p$; in all cases except for
$\psi_\textnormal{over}^{(1)}(a_1,-1/4;500); p)$, we plot the eigenfunction 
multiplied by a factor of $p^{3/4}$ to better illustrate the oscillatory structure. 
From Figures~\ref{fig: psiMAXM2},~\ref{fig: psiMAXM3} and~\ref{fig: psiMAXM4}, one sees that all $M\geq 2$ backflow vectors have higher frequency contributions when compared with the $M=1$ backflow vector. This is not a surprise when considering the integral kernel of the $C^{(M)}$ for $M \geq 2$ in \eqref{eq:CMdefn}. 

At time $t$, the position wavefunction of a state with time-zero momentum wavefunction $\phi(k)$ is 
\begin{equation}
    (\mathcal{R}_t\phi)(x)= \frac{1}{\sqrt{2\pi}} \int_0^\infty dk \, e^{ixk-ik^2t} \phi(k).
\end{equation}
Using a special function representation of $\mathcal{R}_tU^\ast\psi_{m,a,\delta}$, the time evolved position probability densities were computed for backflow/overflow-maximizing states with time-zero position wavefunctions 
\begin{equation}
\chi_\textnormal{back/over}^{(M)} = \mathcal{R}_0 U^\ast \psi^{(M)}_\textnormal{back/over}(a_M,1/4;100;\cdot).
\end{equation}
Here, $N=100$ was used to reduce computational overheads. The corresponding time-zero position probability densities are plotted in Figs.~\ref{fig:backvecsM13PR},~\ref{fig:backvecsM24PR} and~\ref{fig:overvecsM2_4PR}, while the time-evolved probability densities appear in Figs.~\ref{fig:backvecM2heatmap} and~\ref{fig:overvecM2heatmap}
for $M=2$ and the probability fluxes of $\chi_\textnormal{back}^{(M)}$ at $x=0$ are shown in Fig.~\ref{fig:backvecJM1_4}. Further plots may be found in Supplement~\ref{sec:pos_plots}.

\begin{figure}
    \centering
    \includegraphics[width=0.8\textwidth]{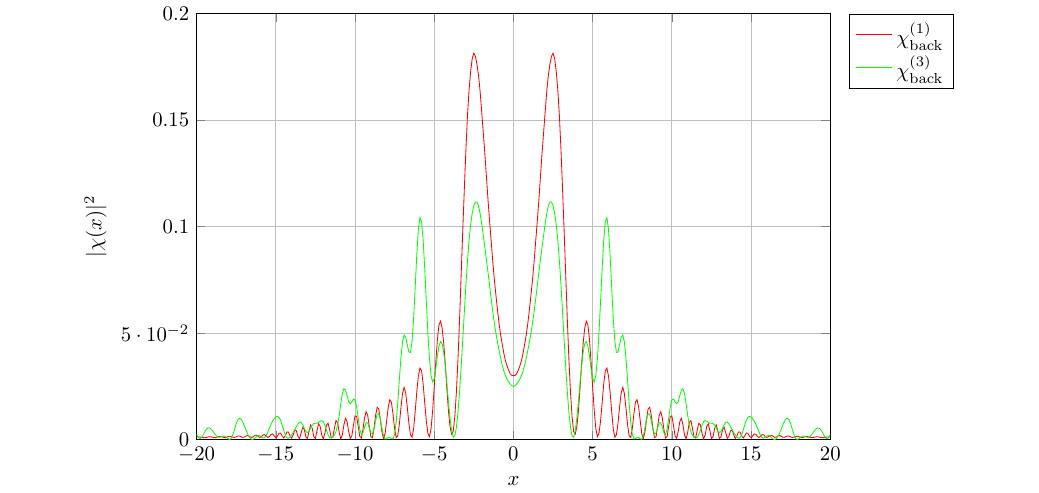}
    \caption{Position probability distribution of the $M=1,3$ backflow states $\chi^{(M)}_\textnormal{back}$ at time $t=0$.}
    \label{fig:backvecsM13PR}
\end{figure}

\begin{figure}
    \centering
    \includegraphics[width=0.8\textwidth]{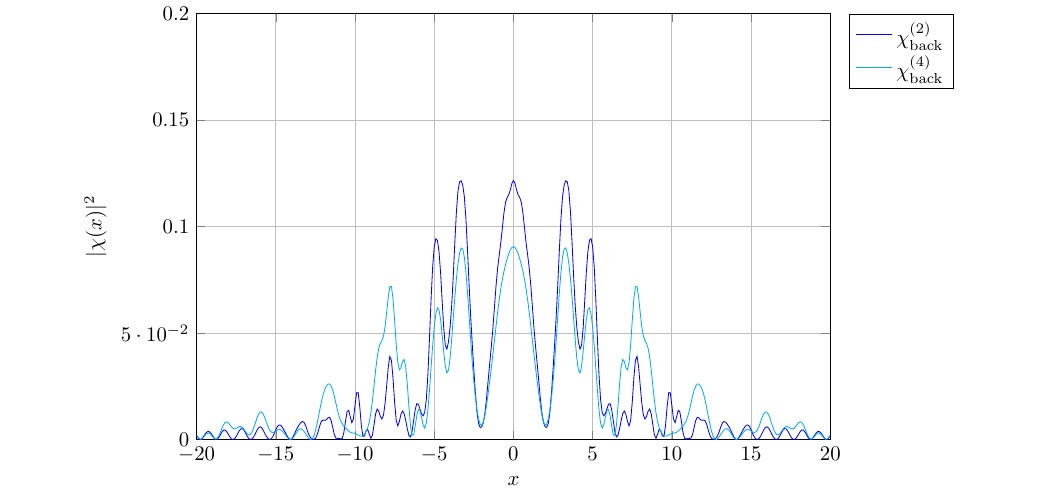}
    \caption{Position probability distribution of the $M=2,4$ backflow states $\chi^{(M)}_\textnormal{back}$ at time $t=0$.}
    \label{fig:backvecsM24PR}
\end{figure}

\begin{figure}
    \centering
    \includegraphics[width=0.8\textwidth]{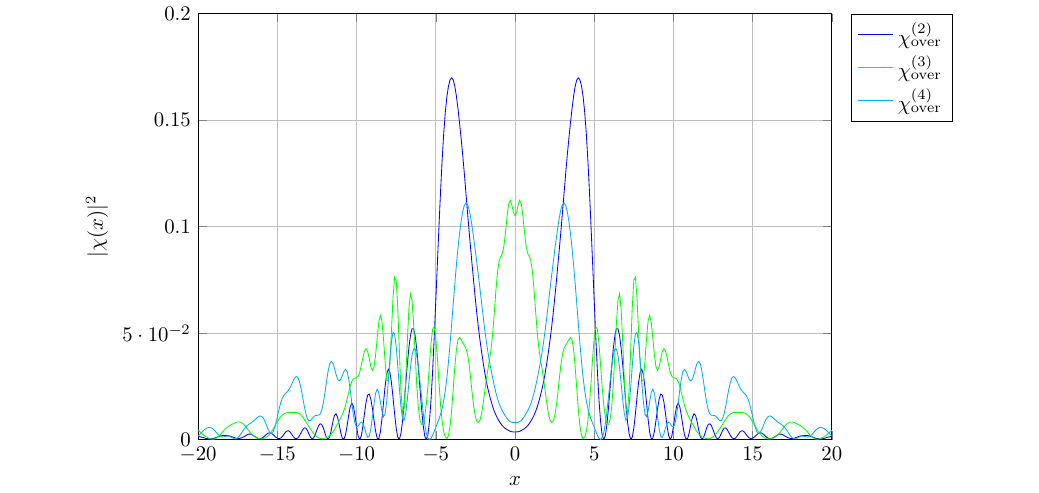}
    \caption{Position probability distribution of the $M=2,3,4$ overflow states $\chi^{(M)}_\textnormal{over}$ at time $t=0$.}
    \label{fig:overvecsM2_4PR}
\end{figure}

As shown in Proposition~\ref{prop: spectral_recipe}, if the sequences   $(\psi^{(M)}_\textnormal{back/over}(a_M,-1/4;N;\cdot))_{N \in \mathbb{N}}$ converge in norm, then they are Weyl sequences and the associated limiting vector and spectral point constitute an eigenpair. For integers $n,m$, let $d^{(M)}_{\textnormal{back/over}}(n,m)$ be defined as
\begin{equation}
    d^{(M)}_{\textnormal{back/over}}(n,m)=\|\psi^{(M)}_\text{back/over}(a_M,-1/4;n; \cdot)-\psi^{(M)}_\text{back/over}(a_M,-1/4;m;\cdot)\|_{L^2(\mathbb{R}^+)}.
\end{equation}
As a crude test, we have computed $d^{(M)}_\textnormal{back}(N,500)$ for $1 \le N \le 499$. For each $M=1,2,3,4$ the values form a decaying nonincreasing sequence with final values tabulated in Table~\ref{tab:N500values} alongside the final value for each of the sequences $\lambda_\textnormal{back/over}^{(M)}$. 
 
\begin{table}
\renewcommand{\arraystretch}{1.2}
\begin{center}
    \begin{tabular}{ | m{1em} | m{8em} | m{8em} | m{8em} |}
        \hline
        & & &\\[-1em]
        $M$ & $\lambda^{(M)}_\textnormal{back}(500)$ & $\lambda^{(M)}_\textnormal{over}(500)$ & $d^{(M)}_\text{back}(499,500)$\\
        \hline
        $1$ & $0.036933$ & $-1$ & $0.0058$\\
        \hline
        $2$ & $ 0.058464$ & $-1.0030$ & $0.0017$\\
        \hline
        $3$ & $0.074860$ & $-1.0089$ & $0.0041$\\
        \hline
        $4$ & $0.088378$ & $-1.0157$ & $0.0084$\\
        \hline
    \end{tabular}
\end{center}
\caption{\label{tab:N500values}Values of $\lambda^{(M)}_{\textnormal{back/over}}(500)$ and $d^{(M)}_\text{back}(499,500)$ for $1 \leq M \leq 4$.}
\end{table}
Similar results are obtained for the norm differences of overflow vectors $d^{(M)}_\text{over}(N,500)$ for $1\le N\le 499$ and $M=2,3,4$. Though far from a rigorous proof, our numerical results suggest that $b^{(M)}_{\textnormal{back}}(1,\ldots,1)=\max\sigma(C^{(M)})$ is an eigenvalue of $C^{(M)}$ for $1\le M\le 4$, while  
$b^{(M)}_{\textnormal{over}}(1,\ldots,1)=\min\sigma(C^{(M)})$ is an eigenvalue of $C^{(M)}$ for $2\le M\le 4$.

\begin{figure}
\centering
\includegraphics[scale=0.9]{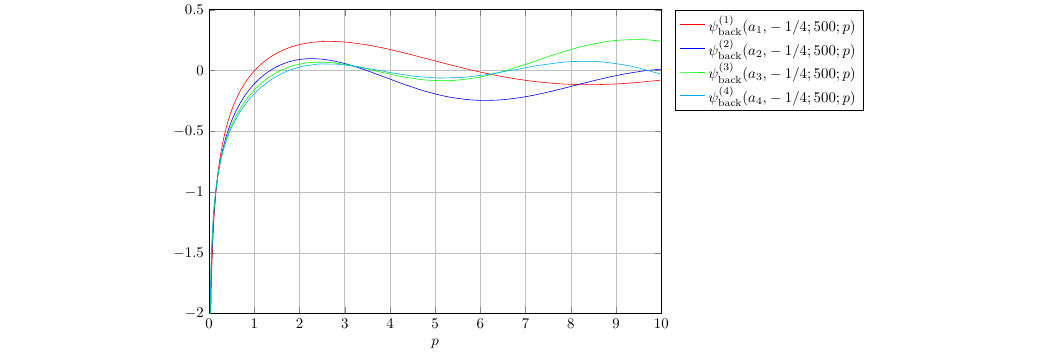}
\caption{Plot of $\psi_\textnormal{back}^{(M)}(a_M,-1/4;500;p)$ for $p \in [0,10]$}
\label{fig: psiMAXM1_4p0_10}
\end{figure}

\begin{figure}
\centering
\includegraphics[width=\textwidth]{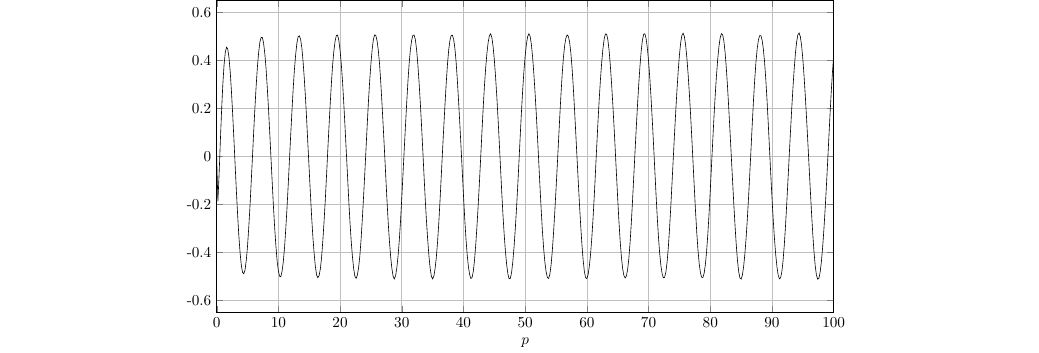}
\caption{Plot of $p^{3/4}\psi_\textnormal{back}^{(1)}(a_1,-1/4;500;p)$ for $p \in [0,100]$}
\label{fig: psiMAXM1}
\end{figure}

\begin{figure}
\centering
\includegraphics[width=\textwidth]{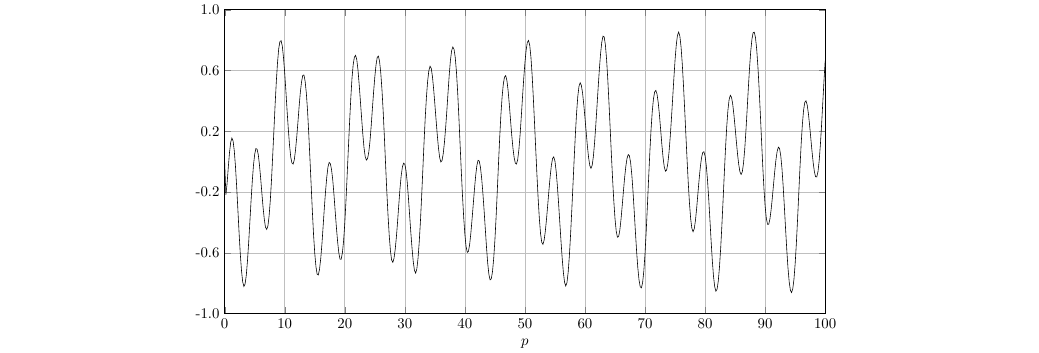}
\caption{Plot of $p^{3/4}\psi_\textnormal{back}^{(2)}(a_2,-1/4;500;p)$ for $p \in [0,100]$}
\label{fig: psiMAXM2}
\end{figure}

\begin{figure}
\centering
\includegraphics[width=\textwidth]{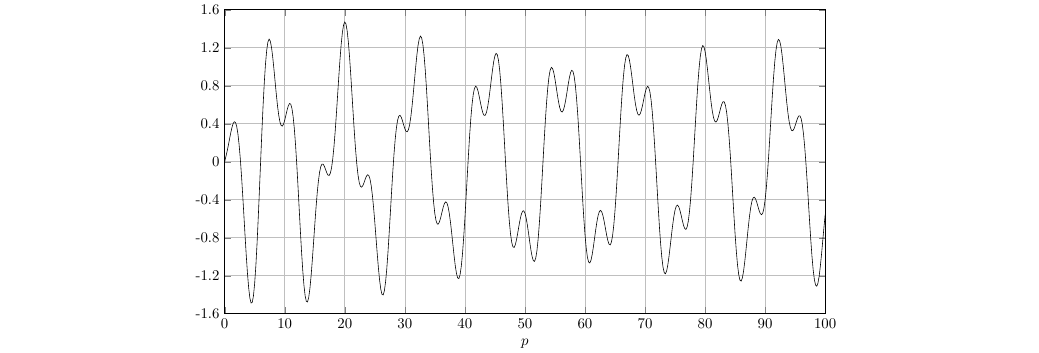}
\caption{Plot of $p^{3/4}\psi_\textnormal{over}^{(2)}(a_2,-1/4;500;p)$ for $p \in [0,100]$}
\label{fig: psiMINM2}
\end{figure}

\begin{figure}
\centering
\includegraphics[width=\textwidth]{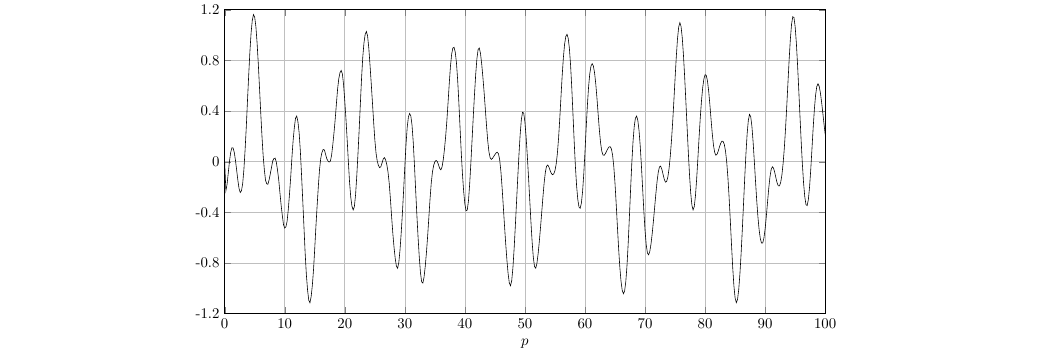}
\caption{Plot of $p^{3/4}\psi_\textnormal{back}^{(3)}(a_3,-1/4;500;p)$ for $p \in [0,100]$}
\label{fig: psiMAXM3}
\end{figure}

\begin{figure}
\centering
\includegraphics[width=\textwidth]{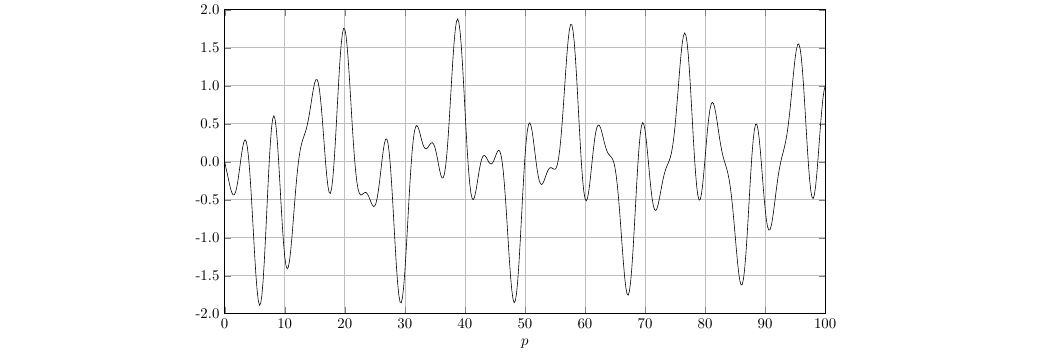}
\caption{Plot of $p^{3/4}\psi_\textnormal{over}^{(3)}(a_3,-1/4;500;p)$ for $p \in [0,100]$}
\label{fig: psiMINM3}
\end{figure}

\begin{figure}
\centering
\includegraphics[width=\textwidth]{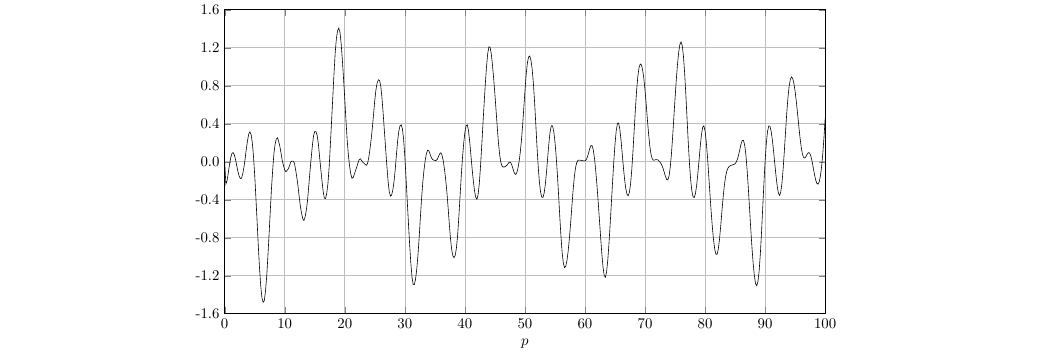}
\caption{Plot of $p^{3/4}\psi_\textnormal{back}^{(4)}(a_4,-1/4;500;p)$ for $p \in [0,100]$}
\label{fig: psiMAXM4}
\end{figure}

\begin{figure}
\centering
\includegraphics[width=\textwidth]{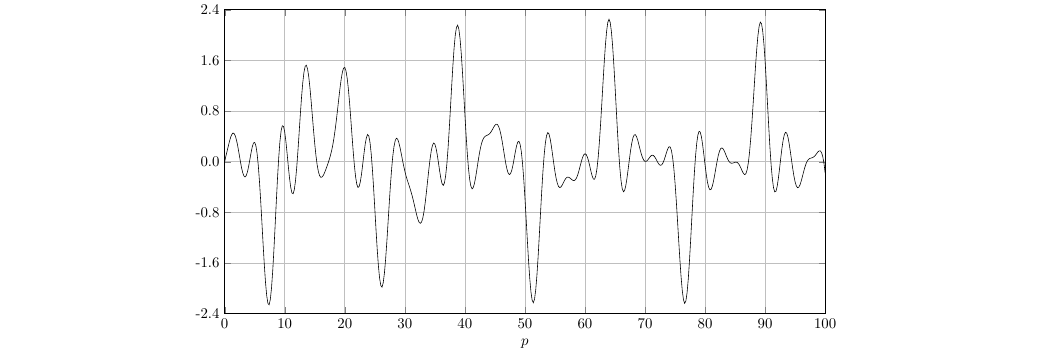}
\caption{Plot of $p^{3/4}\psi_\textnormal{over}^{(4)}(a_4,-1/4;500;p)$ for $p \in [0,100]$}
\label{fig: psiMINM4}
\end{figure}

The infimum of the spectrum of the $M=1$ operator $C^{(1)}$ requires particular discussion. Fig.~\ref{fig:psi_over_1} shows the vectors $\psi_\textnormal{over}^{(1)}(a_1,-1/4;N)$ for $N$ running from $50$ to $450$ inclusive in multiples of $50$, generated from $N_{\textnormal{max}}=500$ data. Each plot consists of a single peak centred roughly at $1.5N$, which broadens as $N$ increases. This suggests a sequence of vectors weakly converging to the zero vector in $L^2(\mathbb{R}^+,dp)$ as $N$ increases, whose $C^{(1)}$-expectation values tend to $-1$. Such a sequence would constitute a singular Weyl sequence for $C^{(1)}$ at $-1$, from which one could deduce that $-1\in\sigma_\textnormal{ess}(C^{(1)})$
(see e.g., theorem 2 in of \cite{MR1192782}) -- a fact that was proved rigorously in~\cite{MR2198971}. Our numerical results are therefore in line with this result.

\begin{figure}
    \includegraphics[width=\textwidth]{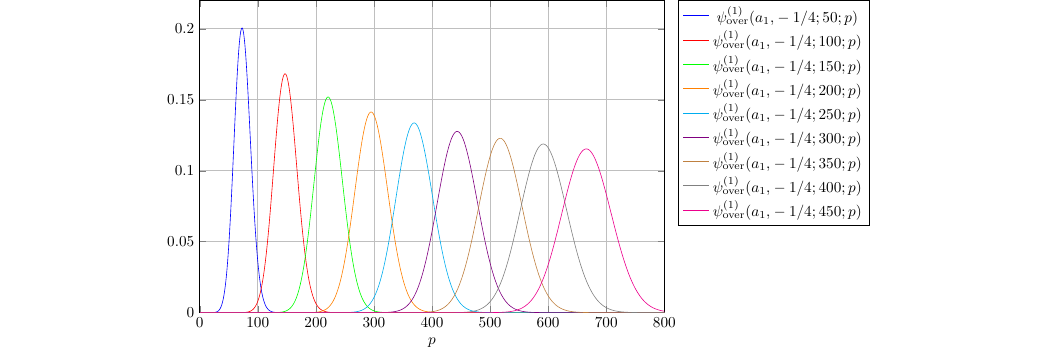}
    \caption{Plots of $\psi_\textnormal{over}^{(1)}(a_1,-1/4;N;p)$ for $N$ in multiples of $50$ and $p \in [0,800]$.}
    \label{fig:psi_over_1}
\end{figure}

\subsection{Numerical acceleration}
\label{sec:numerical_acceleration}

The sequences $\lambda^{(M)}_\textnormal{back/over}$ are guaranteed to converge to the backflow/overflow values $b_\textnormal{back/over}^{(M)}(1,\ldots,1)$, but the rate of convergence is slow. For example, $\lambda^{(1)}_\textnormal{back}(500)=0.03693$ can be compared with the expected value $b_\textnormal{back}^{(1)}=0.03845$ to 4 significant figures. In this section we consider numerical acceleration techniques that modify a convergent sequence, while preserving its limit, but increasing the speed of convergence, at least for sequences of interest. The aim is to provide a better estimate of the eventual limit from the available data. This is done with using sequence acceleration. 

Our aim is to apply (generalised) Richardson accelerators, detailed in Supplementary Section~\ref{appendix: Richardson Accelerator}, to the sequences $\lambda^{(M)}_\textnormal{back/over}$ for $1\le M\le 4$. Briefly, suppose that a sequence $(x_n)_{n\in\mathbb{N}}$ has asymptotic form
\begin{equation}
  x_n =  \sum_{j=1}^k\frac{a_j}{n^{\gamma_j}}  + \epsilon_n, \qquad |\epsilon_n|\le \frac{C}{n^{\gamma_{k+1}}}
\end{equation}
where $0=\gamma_1<\gamma_2<\cdots <\gamma_k<\gamma_{k+1}$, and $a_1,\ldots,a_k$ and $C$ are constant real numbers,
so $x_n\to a_1$ as $n\to\infty$. Writing $\boldsymbol{\gamma}=(\gamma_1,\ldots,\gamma_{k+1})$, the generalised Richardson accelerator $R_{\boldsymbol{\gamma}}$ 
is a linear map on sequences with the property 
that
\begin{equation}
(R_{\boldsymbol{\gamma}}x)_n 
= a_1 + (R_{\boldsymbol{\gamma}}\epsilon)_n,  \qquad 
    |(R_{\boldsymbol{\gamma}}\epsilon)_n|\le \frac{C_{\boldsymbol{\gamma}}}{n^{\gamma_{k+1}}} ,
\end{equation}
where the constant $C_{\boldsymbol{\gamma}}$ can be computed explicitly. Thus the rate of convergence is improved.

As the asymptotics of the sequences $\lambda^{(M)}_\textnormal{back/over}$ are unknown we performed tests that suggest strongly that they have an asymptotic form like the above with half-integer exponents. Another issue is that the sequences $\lambda^{(M)}_\textnormal{back}$ contain oscillatory components with amplitude and period increasing with $M$. These oscillations become apparent in the accelerated sequences, once the dominant power law terms are removed, obstructing attempts to estimate the limit. Details of the oscillatory damping and evidence for the ansatz can be found in Supplementary Section~\ref{appendix: ansatz analysis and damping}.

Using generalised Richardson accelerators, we obtain conjectural bounds on $c_\textnormal{BM}$. Application of $R_{(0.5,1,1.5,2,2.5,3)}$ and $R_{(0.5,1,1.5,2,2.5,3,3.5)}$ to $\lambda^{(1)}_\textnormal{back}$ result in descending and ascending sequences respectively as shown in Figure~\ref{fig: RichmaxestsM1_4}. By selecting the final values of each sequence, we obtain conjectured upper and lower bounds to the Bracken-Melloy constant. To $8$ significant figures, we find
\begin{equation}
    0.038450556 \leq c_\textnormal{BM} \leq 0.038450568
\end{equation}
and thus $c_\textnormal{BM} = 0.0384506$ correct to the first $6$ significant figures. Of particular note is that the upper bound we obtain is strictly below the previously accepted figure of $0.038452$. Possible reasons for this are discussed below. Note that while the accelerators we use result in sequences whose eventual limit is guaranteed to be equal to $c_\textnormal{BM}$, it is conceivable that the limited number of terms available may give a misleading impression. For example, we cannot be sure that the two sequences arising from $R_{(0.5,1,1.5,2,2.5,3)}$ and $R_{(0.5,1,1.5,2,2.5,3,3.5)}$ remain monotonic and provide upper and lower bounds to $c_\textnormal{BM}$. If more terms were computed, these sequences might cross, perhaps many times, before approaching their common limit. However, problems of this sort are common to all numerical computation that is unsupported by rigorous bounds. We are currently conducting an independent calculation of $c_{\textnormal{BM}}$ which it is hoped will serve as a cross-check on these values.

Our estimate for $c_\textnormal{BM}$ agrees with those of~\cite{MR2119354,MR2198971} to 4 significant figures but differs beyond that. In fact the 5th significant figure in~\cite{MR2119354} was stated tentatively, though it was then apparently confirmed and refined in~\cite{MR2198971}. We briefly compare the methods used to argue that our new estimate is likely to be the more accurate of the three,
based on the following two observations. First, in~\cite{MR2119354,MR2198971} the backflow operator (or an operator unitarily equivalent to it) on $L^2(\mathbb{R}^+)$ was numerically diagonalised by truncating to a subspace $L^2([0,\Lambda])$ and then discretising the resulting operator by setting a mesh size $\delta\ll \Lambda$ resulting in a matrix problem of dimension $\mathcal{O}(\Lambda/\delta)$. Although the techniques in~\cite{MR2119354,MR2198971} were different in detail, they share the feature that $c_\textnormal{BM}$ is obtained in a double limit $\delta\to 0$ and $\Lambda\to\infty$.
In practice, a mesh size $\delta(\Lambda)>0$ is selected that appears to give reasonable estimate of the $\delta\to 0$ limit at fixed $\Lambda$, and a sequence of results is obtained by increasing $\Lambda$ through some sequence of cutoff sizes $\Lambda_N$. Consequently, the sequence of results can depend on the sequence $\Lambda_N$ as well as the dependence of $\delta$ on $\Lambda$. By contrast, 
our method is based on closed form expressions for the matrix elements in terms of special functions, avoiding cutoffs and discretisation.  
Once we have chosen our sequence of basis functions, the sequence of eigenvalue estimates is labelled by the single parameter $N$ that indicates how many basis vectors have been used; we are also able to compute these estimates to high precision based on the error analysis in Section~\ref{exact_matrix_elements}.
Second,~\cite{MR2119354,MR2198971} used relatively straightforward fitting or extrapolation methods to estimate the limiting value $c_\textnormal{BM}$, whereas we have used sequence acceleration to systematically remove terms in the asymptotic series resulting in sequences that converge significantly faster to $c_\textnormal{BM}$ than the raw sequence $\lambda_\textnormal{back}^{(1)}$. 

Plots of the accelerated sequences are presented in Figure~\ref{fig: RichmaxestsM1_4}, with the resulting bounds tabulated in Table~\ref{tab: backflow_const_bounds}. As with the $M=1$ case, these `bounds' are not rigorous, but indicate that accelerated sequences appear to sandwich the limiting value. 

\begin{figure}
\centering
    \includegraphics[width=\textwidth]{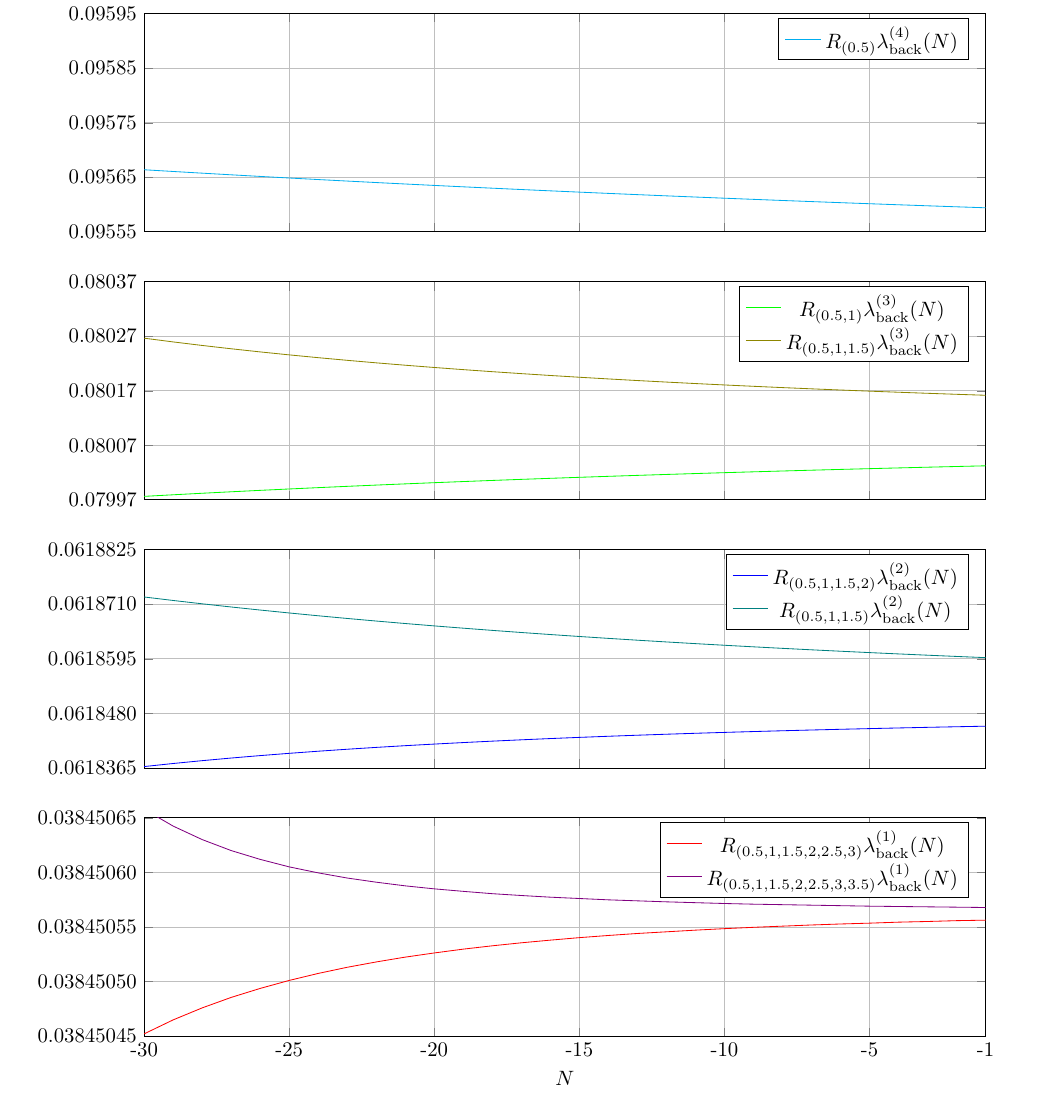}
    \caption{Plots of $R_{\boldsymbol{\gamma}_M} \lambda^{(M)}_\textnormal{back}$ for selected powers $\boldsymbol{\gamma}_M$ and $1 \leq M \leq 4$}
    \label{fig: RichmaxestsM1_4}
\end{figure}

\begin{table}
\begin{center}
    \begin{tabular}{ | m{1em} | p{13em} | m{8em} | m{9em} | }
        \hline
        $M$ & Acceleration parameters & Lower bound (LB)& Upper bound (UB)\\
        \hline
        $1$ & LB: $(0.5,1,1.5,2,2.5,3)$,\newline UB: $ (0.5,1,1.5,2,2.5,3,3.5)$ & $0.0384505563$ & $0.0384505678$\\
        \hline
        $2$ & LB: $(0.5,1,1.5,2)$,\newline UB: $(0.5,1,1.5)$ & $0.0618453024$ & $0.0618597497$\\
        \hline
        $3$ & LB: $(0.5,1)$, UB: $(0.5,1,1.5)$ & $0.0800324618$ & $0.0801617473$\\
        \hline
        $4$ & UB: $(0.5)$ & $-$ & $0.0955940854$\\
        \hline
    \end{tabular}
\end{center}
\caption{\label{tab: backflow_const_bounds} Conjectured upper and lower bounds for $c_\textnormal{BM}$ and $b^{(M)}_\textnormal{back}(1,\dots,1)$ for $1\le M\le 4$.}
\end{table}
 
As a check on our attempts to numerically accelerate $\lambda_\textnormal{back}^{(1)}$, we also made use of theorem~5 from \cite{MR1075161}, which gives a nonlinear acceleration method based on the Raabe-Duhamel (RD) convergence test.
\begin{theorem} 
If a sequence $x$ obeys 
\begin{equation}
\label{eq: RD_condition}
    \lim_{n \to \infty} \left[ (n+1)\frac{x_{n+2}-x_{n+1}}{x_{n+1}-x_{n}}-n\right]<0
\end{equation}
then $x$ converges and the nonlinear sequence transformation 
\begin{equation}
    \RD[x]_n = x_n - \frac{n(x_{n+1}-x_n)^2}{(n+1)(x_{n+2}-x_{n+1})-n(x_{n+1}-x_n)}
\end{equation}
accelerates the convergence of $x$.
\label{thm: RD_acceleration}
\end{theorem}
The sequence $\lambda_\textnormal{back}^{(1)}$ appears to satisfy equation~\eqref{eq: RD_condition} for large $N$, as does $\RD^k[\lambda_\textnormal{back}^{(1)}]$ for $k=1,2,3$. After four applications of theorem~\ref{thm: RD_acceleration} (with some intermediate smoothing to damp oscillations) to $\lambda^{(1)}_\textnormal{back}$, we find $c_\textnormal{BM}\approx 0.0384506$, backing up the values we find in table~\ref{tab: backflow_const_bounds} from the generalised Richardson accelerator. We were only able to apply the accelerator once to each of the sequences $\lambda_\textnormal{back}^{(M)}$ for $M=2,3,4$ and the Richardson results of Table~\ref{tab: backflow_const_bounds} appear to be more stable estimates on the basis of the data for $N=500$, though they rely on our assumption that the asymptotic behaviour mirrors that in the $M=1$ case.

We now turn to the analysis of the minimum spectral estimates $\lambda_\textnormal{over}^{(M)}$ for $2 \leq M\leq 4$. Note that we do not include any analysis of $M=1$ since it is known that $\lambda_\textnormal{over}^{(1)}(N) \rightarrow \min \sigma(C^{(1)})=-1$ as $N \rightarrow \infty$. Our analysis follows that of the maximum spectral estimates and we assume that each $\lambda_\textnormal{over}^{(M)}$ has asymptotics of the same type as $\lambda_\textnormal{back}^{(M)}$. We 
damped oscillations in $\lambda_\textnormal{over}^{(M)}$ by the same methodology as before and then applied the generalised Richardson accelerators to accelerate the convergence of the minimum spectral estimates. 



The final values of the generalised Richardson accelerated overflow estimates are tabulated in Table~\ref{tab: overflow_const_bounds}. The values are written as upper and lower bounds on the overflow values $b^{(M)}_\textnormal{over}(1,\dots,1)$ for $M$ intervals of equal length and spacing, in line with the assumption that the monotonicity of the sequences remain fixed. 
We note in particular that the overflow value $b^{(2)}_{\textnormal{over}}(1,1)$ for two intervals of equal length and spacing is far from the $M=2$ overflow constant $c_\textnormal{over}^{(2)}=-1-c_\textnormal{BM}$ arising from all possible pairs of intervals.
\begin{table}
\begin{center}
    \begin{tabular}{ | m{1em} | p{13em} | m{8em} | m{9em} | }
        \hline
        $M$ & Acceleration parameters & Lower bound (LB)& Upper bound (UB)\\
        \hline
        $2$ & LB: $(0.5,1,1.5)$, \newline UB: $ (0.5,1,1.5,2)$ & $-1.0037916$ & $-1.0037899$\\
        \hline
        $3$ & LB: $(0.5,1,1.5)$, UB: $(0.5,1)$ & $-1.011078$ & $-1.011029$\\
        \hline
        $4$ & UB: $(0.5)$, $(0.5,1)$ & $-1.01947$ & $-1.01934$\\
        \hline
    \end{tabular}
\end{center}
\caption{\label{tab: overflow_const_bounds} Conjectured upper and lower bounds of $b^{(M)}_\textnormal{over}(1,\dots,1)$ for $M=2,3,4$ intervals of equal length and spacing.}
\end{table}

\section{Conclusion}\label{sec:conclusion}
 
We have shown that a free particle in $1$-dimensional quantum mechanics with non-negative momenta can exhibit probability flow in the opposite direction to its momentum over multiple disjoint time intervals. In particular, given a positive integer $M$ number of time intervals we have found that the total amount of backflow can grow as $M^{1/4}$. 

Similarly, we have shown the total probability flow in the same direction as the particle's momentum can exceed unity.
By contrast, we have shown that the total backflow for statistical ensembles of classical particles is bounded within the interval $[-1,0]$. This result follows a host of phenomenon relating quantized observables compared to their classical counterparts. As has been shown, classical observables which have sharp bounds on their values often have quantized counterparts whose values exceed these bounds \cite{MR2119354}. The same is true for the case of total probability flow. As already noted, it could be profitable to study overflow without the restriction to positive momentum states.

Here it is worth making a comment regarding Tsirelson precession, which has already been noted to have a connection with probability backflow in the preprint~\cite{zaw2024all}. Tsirelson showed that if a uniformly precessing coordinate is measured at a randomly chosen angle $\theta \in \{0,2\pi/3,4\pi/3\}$, then the probability $P$ that the measured coordinate is positive is bounded for classical systems by $P_\textnormal{classical} \leq 2/3$~\cite{Tsirelson2006}. However, there exist quantum systems that violate this inequality. In the PhD thesis of Zaw~\cite{Zaw_thesis}, it is shown that for a generalisation of the Tsirelson's protocol to an odd number $2M+1$ of angles, the maximum quantum violation seems to grow like $M^{1/4}$. Further work is being conducted to make the relationship between multiple quantum backflow and this generalisation of Tsirelson's protocol precise. 

We have illustrated our analytical results with numerical calculations of backflow and overflow. In particular this shows that the backflow constants $c_\textnormal{back}^{(M)}$ are strictly greater than $c^{(1)}=c_\textnormal{BM}$ for $M=2,3,4$, and provides a new estimate of the Bracken--Melloy constant $c_\textnormal{BM}\approx 0.0384506$ to 6 significant figures. This estimate relies on numerical acceleration techniques and appears robust. An independent calculation is in progress and will provide a cross-check.

Quantum backflow over a single time interval can be attributed to the ability of the quantum mechanical current of a quantum state to attain arbitrarily negative values. The existence of backflow over multiple disjoint time regions indicates that, ignoring effects of measurement on the state, this phenomena can repeat. Namely, there exist states $\psi$ with non-negative momentum for which a time of arrival detector would repeatedly detect the the state crossing $x=0$. Furthermore, the calculation in Section~3 shows that, as the number $M$ of disjoint time intervals grows, the maximal backflow grows at least as fast as $\mathcal{O}(M^{1/4})$.  This larger effect potentially opens a new avenue to experimentally verify quantum backflow.  However, the fact that this requires measurements of probability differences at many different times presents a challenge that would have to be overcome.

\paragraph{Acknowledgments} 
We thank Reinhard Werner, who posed the question of whether multiple backflow was possible following talks by HKK and CJF at the workshop \emph{Quantum Advantages in Mechanical Systems}, IQOQI, Vienna, September 2023. 
We also thank Lin Htoo Zaw for useful conversations regarding the relation with Tsirelson protocols, and Simon Eveson for comments on parts of the manuscript. Numerical computations were partly performed using the Viking cluster, a high performance compute facility provided by the University of York. We are grateful for computational support from the University of York IT Services and Research IT team. The work of HKK is supported by an EPSRC Studentship.

For the purpose of open access, the authors have applied a creative commons attribution (CC BY) licence to any author accepted manuscript version arising.

\bibliography{biblog}

\pagebreak


\begin{center}
{\Huge Supplemental Material}
\end{center}

\appendix
\section{The Bracken--Melloy operator}

\label{appendix:single_backflow}
Let $L^2(\mathbb{R},dx)$ be the usual Hilbert space of states for the particle on a line and
$L^2(\mathbb{R}^+,dk)$ be the Hilbert space of states with nonnegative momentum, in the momentum representation. The reverse Fourier transform
determines an isometry $\mathcal{R}:L^2(\mathbb{R}^+,dk)\to L^2(\mathbb{R},dx)$ by
\begin{equation}\label{eq:rev_Fourier}
(\mathcal{R}\phi)(x) = \frac{1}{\sqrt{2\pi}} \int_0^\infty dk \, e^{ixk} \phi(k)
\end{equation}
for $\phi\in L^2(\mathbb{R}^+,dk)$. 

For times $t_1<t_2$ let 
\begin{equation}
    \Delta_{\langle t_1,t_2\rangle}(\psi) = \textnormal{Prob}_\psi(X<0|t=t_2)-\textnormal{Prob}_\psi(X<0|t=t_1),
\end{equation}
where $\psi$ evolves according to the free Schr\"odinger evolution, $\psi_t=e^{-iHt}\psi$, and $H\psi=-\psi''$ on its usual domain. We will be interested in nonnegative momentum states of the form $\psi=\mathcal{R}\phi$ for $\phi\in L^2(\mathbb{R}^+,dk)$, for which we have $e^{-iHt}\psi= \mathcal{R}S_t\phi$, where $(S_t\phi)(k) = e^{-ik^2t}\phi(k)$ defines a unitary operator $S_t\in\mathcal{B}(L^2(\mathbb{R}^+,dk))$ for any $t\in\mathbb{R}$. Then, if $\|\psi\|=1$, one has
\begin{equation}
\Delta_{\langle t_1,t_2\rangle}(\psi) = \ip{\hat{\psi}}{B_{\langle t_1,t_2\rangle}\hat{\psi}}
\end{equation}
where
\begin{equation}
    B_{\langle t_1,t_2\rangle} = S_{t_2}^* \mathcal{R}^* \Pi^- \mathcal{R} S_{t_2} - 
    S_{t_1}^* \mathcal{R}^* \Pi^- \mathcal{R} S_{t_1}
\end{equation}
and $\Pi^-$ is the projection onto the subspace $L^2(\mathbb{R}^-,dx)\subset L^2(\mathbb{R},dx)$
of states supported on the negative half-line in position space.
Evidently $B_{\langle t_1,t_2\rangle}$ is the difference of two positive semidefinite self-adjoint bounded operators, each of which is bounded above by the identity. Accordingly, $B_{\langle t_1,t_2\rangle}$ is a self-adjoint bounded operator with spectrum $\sigma(B_{\langle t_1,t_2\rangle})\subseteq [-1,1]$.

To obtain expressions for $B_{\langle t_1,t_2\rangle}$ as an integral operator, we proceed as in~\cite{MR2198971}. First introduce $\Sigma=\Pi^+-\Pi^-=1 -2\Pi^-$ and note that
\begin{equation}
     B_{\langle t_1,t_2\rangle} =-\frac{1}{2} \left(S_{t_2}^* \mathcal{R}^* \Sigma\mathcal{R} S_{t_2} - 
    S_{t_1}^* \mathcal{R}^* \Sigma \mathcal{R} S_{t_1}\right)
\end{equation}
Next, observe that $i\mathcal{R}^*\Sigma \mathcal{R}$ is the restriction to $L^2(\mathbb{R}^+,dk)$ of the Hilbert transform
\begin{equation}
    (\mathcal{H} \phi)(k)=  \frac{1}{\pi}\textnormal{PV}\int_{-\infty}^\infty dl \frac{\phi(l)}{k-l}
\end{equation}
on $L^2(\mathbb{R},dk)$. Here the PV denotes a principal value integral $\lim_{\epsilon\to 0+}\left(\int_{-\infty}^{k-\epsilon}dl+\int_{k+\epsilon}^\infty dl\right)$, which converges in $L^2$, and  pointwise almost everywhere, for $\phi\in L^2(\mathbb{R})$. See Theorem 4.1.12 in~\cite{Grafakos:2014} or Chapter II.3 of \cite{MR0290095}.  
Thus
\begin{equation}
    B_{\langle t_1,t_2\rangle} =-\frac{1}{2i}\left(S_{t_2}^* \mathcal{H} S_{t_2} - 
    S_{t_1}^* \mathcal{H} S_{t_1}\right),
\end{equation}
where we understand the restriction of $\mathcal{H}$ to $L^2(\mathbb{R}^+,dk)$ on the right-hand side.
The integral representation
\begin{equation}
    (B_{\langle t_1,t_2\rangle} \phi)(k) = -\frac{1}{2\pi i}\int_{0}^\infty dl \frac{
    e^{it_2(k^2-l^2)}-e^{it_1(k^2-l^2)}}{k-l}\phi(l)
\end{equation}
then holds pointwise almost everywhere for any $\phi\in L^2(\mathbb{R}^+,dk)$; note that the principal value is no longer needed because the fraction appearing in the integrand is regular. A particular special case is
\begin{equation}
    (B_{\langle -\tau,\tau\rangle} \phi)(k) = -\frac{1}{\pi}\int_{0}^\infty dl \frac{
    \sin(k^2-l^2)\tau}{k-l}\phi(l)
\end{equation}
for any $\tau>0$.

Note also that $t\mapsto S_t^*\mathcal{H}S_t$ is strongly continuous, because $S_t$ and $S_t^*$ are strongly continuous and uniformly bounded. It follows that $(t_1,t_2)\mapsto B_{\langle t_1,t_2\rangle}$ is strongly continuous on 
the set $\mathcal{T}_2=\{(t_1,t_2)\in\mathbb{R}^2: t_1<t_2\}$.

We now turn to some unitary equivalences between operators of the form $B_{\langle t_1,t_2\rangle}$ and related operators. First, one has the elementary covariance relation
\begin{equation}
  S_\tau^* B_{\langle t_1,t_2\rangle} S_\tau= B_{\langle t_1+\tau,t_2+\tau\rangle} 
\end{equation}
for any $\tau\in\mathbb{R}$. Second, for any $\alpha>0$, let $D_\alpha\in \mathcal{B}(L^2(\mathbb{R}^+,dk))$ be the unitary operator $(D_\alpha\phi)(k)=\sqrt{\alpha}\phi(\alpha k)$. It is readily checked that
$S_\tau D_\alpha = D_\alpha S_{\tau/\alpha^2}$ and also that $D_\alpha^*\mathcal{R}^*\Pi^-\mathcal{R}D_\alpha=\mathcal{R}^*\Pi^-\mathcal{R}$. Consequently, we have 
\begin{equation}
    D_\alpha^* B_{\langle t_1,t_2\rangle} D_\alpha= B_{\langle t_1/\alpha^2,t_2/\alpha^2\rangle}.
\end{equation}
Combining, for any $\tau>0$, we have
\begin{equation}
    B_{\langle t_1,t_2\rangle}=S^*_{(t_1+t_2)/2}D^*_{\sqrt{2\tau/(t_2-t_1)}}B_{\langle -\tau,\tau\rangle}
    D_{\sqrt{2\tau/(t_2-t_1)}}S_{(t_1+t_2)/2}
\end{equation}
which shows that all the operators $B_{\langle t_1,t_2\rangle}$ for $t_1<t_2$ are unitarily equivalent and therefore have the same spectrum. 
Third, for any $\tau>0$ we may define a unitary $U_\tau:L^2(\mathbb{R}^+,dk)\to L^2(\mathbb{R}^+,dq)$ by
\begin{equation}
    (U_\tau \phi)(q) = \frac{\phi(\sqrt{q/\tau})}{\sqrt{2}(q\tau)^{1/4}},
\end{equation}
where $q$ is a dimensionless variable (representing action in units of $\hbar$).
Changing variables, one finds that
\begin{equation}
    (U_\tau B_{\langle t_1,t_2\rangle} U_\tau^* \varphi)(p)= 
    -\frac{1}{4\pi i}\int_{0}^\infty dq \frac{
    (e^{it_2(p-q)/\tau}-e^{it_1(p-q)/\tau})}{p-q}\Big[\Big(\frac{p}{q}\Big)^{1/4}+\Big(\frac{p}{q}\Big)^{-1/4}\Big]\varphi(q)
\end{equation} 
for $\varphi\in L^2(\mathbb{R}^+,dq)$
Accordingly, every single backflow operator $B_{\langle t_1,t_2\rangle}$ is unitarily equivalent to
the bounded operator $C=U_\tau B_{\langle -\tau,\tau\rangle} U_\tau^*\in \mathcal{B}(L^2(\mathbb{R}^+,dq))$ given by
\begin{equation}
    (C\varphi)(p)=-\frac{1}{2\pi }\int_{0}^\infty dq \frac{
    \sin(p-q)}{p-q}\Big[\Big(\frac{p}{q}\Big)^{1/4}+\Big(\frac{p}{q}\Big)^{-1/4}\Big]\varphi(q).
\end{equation}
By definition, $c_{\textnormal{BM}}=\max \sigma (C)$, and as $\| C\|\le 1$ it follows that
$\sigma(C)\subset [-1,c_{\textnormal{BM}}]$. All the assertions made in Theorem~\ref{thm:B_facts} have been proved.

Finally, the following result is used in Section~\ref{sec:limiting cases}.
\begin{lemma}\label{lem:cutoff_norm_continuity}
Let $\iota_\Lambda:L^2([0,\Lambda],dk)\to L^2(\mathbb{R}^+,dk)$ be the 
subspace inclusion. Then $(s,t)\mapsto \iota_\Lambda^* B_{\langle s,t\rangle}\iota_\Lambda\in \mathcal{B}(L^2([0,\Lambda],dk))$ is norm continuous 
on $\mathcal{T}_2$.
\end{lemma} 
\begin{proof}
    
For $\phi \in L^2([0,\Lambda])$ we have (almost everywhere)
\begin{equation}
    |(S_{s} \phi- S_{t}\phi)(p)|^2  =\sin^2\Big(\frac{(t -s)p^2}{2}\Big) |\phi(p)|^2 \leq \frac{\Lambda^4}{4}|s -t|^2 |\phi(p)|^2
\end{equation}
and consequently, in the operator norm on $\mathcal{B}(L^2([0,\Lambda],dk),L^2(\mathbb{R}^+,dk))$ 
\begin{equation}
    \left\|(S_{s}-S_{t})\iota_\Lambda\right\| \leq \frac{\Lambda^2|s-t|}{2}.
\end{equation}
Thus the restricted Schr\"{o}dinger time evolution operators $S_t\iota_\Lambda:L^2([0,\Lambda],dk)\to L^2(\mathbb{R}^+,dk)$ are norm continuous in $t$. Using the general inequality $\| A^*C A - B^*C B\|
= \| A^*C(A-B)+(A-B)^*CB\| 
\le  \|C\|(\|A\|+\|B\|)\|A-B\|$, for bounded Hilbert space operators $A,B\in\mathcal{B}(\mathscr{H},\mathscr{K})$ and $C\in\mathcal{B}(\mathscr{K})$ it follows that 
$\iota_\Lambda^* S_t^\ast \mathcal{R}^\ast \Pi^- \mathcal{R} S_t\iota_\Lambda$
is also norm continuous in $t$ and the result follows.
\end{proof}

\section{Proof of Theorem~\ref{thm: spectral_asymptotic_unboundedness}}
\label{appendix: Lemma3.3Proof}

The following Lemma is the key calculation in the proof of Theorem~\ref{thm: spectral_asymptotic_unboundedness}. 
\begin{lemma}
    \label{lemma: lower_bound_supremum}
        For even $M\in\mathbb{N}$ and $\epsilon\in(0,\pi/6]$,
        \begin{equation}
        \label{eq: unbounded_sup}
            \left\langle \psi_M | C^{(M)} \psi_M \right\rangle \geq \frac{2^{7/4}\epsilon^{3/4}}{3\pi^{11/4}}M^{1/4}S\left(\frac{2\epsilon}{\pi M}\right),
        \end{equation}
        where
        \begin{equation}
         S(\eta)=1-\frac{31\pi^2}{80}\eta^{1/4}-\frac{5}{8}\eta-\frac{9\pi^2}{176}\eta^{5/4}.
        \end{equation}
    \end{lemma}
    Clearly $S(\eta)>0$ for sufficiently small $\eta\ge 0$. Fixing $\epsilon\in (0,\pi/6]$ sufficiently small that $S>0$ on $[0,\epsilon/\pi]$, and using continuity of $S$, let
    \begin{equation}
    k = \min_{\eta\in [0,\epsilon/\pi]} \frac{2^{7/4}\epsilon^{3/4}}{3\pi^{11/4}}S(\eta).
    \end{equation}
    Then~\eqref{eq:supbound} follows because $2\epsilon/(\pi M)\in [0,\epsilon/\pi]$ for all even $M\ge 2$. (The value of $k$ can be optimised by varying $\epsilon$.) To establish the first result in~\eqref{eq:spectral_asymptotics}, we need some information about $\sup \sigma(C^{(M)})$ for odd $M\ge 1$.
    Writing
    \begin{equation}
        B^{(M+1)}_{\langle \boldsymbol{t}_{M+1}\rangle} =B^{(M)}_{\langle -T\frac{2M+1}{2},...,-\frac{T}{2},\frac{T}{2},...,T\frac{2M-3}{2}\rangle}
        +B^{(1)}_{\langle T\frac{2M-1}{2}, T\frac{2M+1}{2}\rangle},
    \end{equation}
    the left-hand side is unitarily equivalent to $C^{(M+1)}$, while
    by making time translations and dilations, the first and second terms on the right-hand side are unitarily equivalent to 
    $C^{(M)}$ and $C^{(1)}$ respectively. Consequently
    $C^{(M+1)}$ is a sum of operators unitarily equivalent to $C^{(M)}$ and $C^{(1)}$. As the supremum of the spectrum is a sub-additive function on bounded self-adjoint operators, and spectra are invariant under
    unitary equivalence,
    it follows that
    \begin{align}
        \sup \sigma\left(C^{(M+1)}\right) &\leq \sup \sigma\left(C^{(M)}\right)+ \sup\sigma\left(C^{(1)}\right)
        = c_\textnormal{BM} + \sup \sigma\left(C^{(M)}\right),
    \end{align}
    and therefore $\sup \sigma\left(C^{(M)}\right)\ge k (M+1)^{1/4}-c_\textnormal{BM}$ for all odd $M\ge 1$, using~\eqref{eq:supbound}.
    Combining with the bound~\eqref{eq:supbound} for even $M$, 
    one has $\sup \sigma\left(C^{(M)}\right)\ge k M^{1/4}-c_\textnormal{BM}$ for all $M\in\mathbb{N}$ and
    the first part of~\eqref{eq:spectral_asymptotics} follows immediately.

    Finally, for any $M\ge 1$ $B^{(M+1)}_{\langle \boldsymbol{t}_{M+1}\rangle}$ can be rewritten as 
\begin{equation}
B^{(M+1)}_{\langle \boldsymbol{t}_{M+1}\rangle}=
B^{(1)}_{\langle -T\frac{2M+1}{2},T\frac{2M+1}{2}\rangle}  - B^{(M)}_{\langle \boldsymbol{t}_{M}\rangle }
\end{equation}
which, because all single-backflow operators are unitarily equivalent, gives
\begin{equation}
    C^{(M+1)}= \widetilde{C}^{(1)}-C^{(M)}
\end{equation}
where $\widetilde{C}^{(1)}$ is unitarily equivalent to $C^{(1)}$. Using invariance of the spectrum under unitary equivalence and the fact that $\sup \sigma(-A) = -\inf \sigma(A)$ for any bounded self-adjoint operator $A$, we find
    \begin{equation}
        \sup \sigma\left(C^{(M+1)}\right) \leq \sup \sigma\left( C^{(1)}  \right) - \inf \sigma \left(C^{(M)}\right)
    \end{equation}
and consequently
\begin{equation}
        \inf \sigma \left(C^{(M)}\right) \leq c_\textnormal{BM} - \sup \sigma \left(C^{(M+1)}\right).
    \end{equation}
Thus $\limsup_{M\to\infty}M^{-1/4}\inf \sigma \left(C^{(M)}\right)\leq -\liminf_{M\to\infty} M^{-1/4}\sup \sigma \left(C^{(M+1)}\right)\leq -k$.   
This completes the proof of Theorem~\ref{thm: spectral_asymptotic_unboundedness}, once Lemma~\ref{lemma: lower_bound_supremum} is established below.

We remark that the first positive zero of $S(\eta)$ occurs at $\eta_0=0.0046$ to two significant figures and $S$ is decreasing on the interval $[0,\eta_0]$. Then the constant $k$ in Theorem~\ref{thm: spectral_asymptotic_unboundedness}
may be taken as
\begin{equation}
  k = \max_{\epsilon\in [0,\pi\eta_0]}
  \frac{2^{7/4}\epsilon^{3/4}}{3\pi^{11/4}}S(\epsilon/\pi)=2.1\times 10^{-4},
\end{equation}
computing the maximum numerically to $2$ significant figures. By considering more complicated trial functions, it seems likely that we can increase the coefficient; at any rate, the main point is to establish the growth of at least $\mathcal{O}(M^{1/4})$. For comparison, the numerical results of Section~\ref{section: numerical_calculation} show that
$\sup \sigma(C^{(4)})\gtrsim 0.095$, which far exceeds the lower bound
$4^{1/4}k = 0.0003$ for this case.

\begin{proof}[Proof of Lemma~\ref{lemma: lower_bound_supremum}]
The expectation value of $C^{(M)}$ in the state $\psi_M$ (see~\eqref{eq:psiM}) can be written in the form
\begin{equation}
\begin{aligned}
\label{mat_element_decomp}
    \ip{\psi_M}{C^{(M)} \psi_M } = -\frac{M^2}{4\pi\epsilon}\Big[ I_{00} + 2I_{01} + I_{11}\Big]
\end{aligned}
\end{equation}
where, for $i,j\in\{0,1\}$, 
\begin{equation}
    I_{ij} = \int_{\Ic_i \times \Ic_j} dp \: dq \: g_M(p-q) \Big[\Big(\frac{p}{q}\Big)^{1/4} + \Big(\frac{p}{q}\Big)^{-1/4}\Big]
\end{equation}
with  
\begin{equation}
        \Ic_0 = \left[0, \frac{\epsilon}{M}\right], \qquad \Ic_1 = \left[\frac{\pi}{2}-\frac{\epsilon}{2M}, \frac{\pi}{2}+\frac{\epsilon}{2M}\right]
\end{equation}
and $g_M \in \mathcal{C}^\infty(\mathbb{R})$ is given by
\begin{equation}
    g_M(x) = \frac{\sin(2Mx)}{2Mx\cos(x)}
\end{equation}
except at the isolated zeros of the denominator, where one extends by continuity. 

We first show that $|g_M(x)|\leq 1$ for all $x\in\mathbb{R}$. 
Recall that $\sinc(x)=\sin(x)/x$ for $x\neq 0$ and $\sinc(0)=1$.
By insertion of the factor $\sin(x)/\sin(x)$ in the definition of $g_M$ and using the double angle formula for the $\sin$ function, one finds
\begin{align}\label{g_M_initial_bound}
    g_M(x) &=  \frac{\sin(2Mx)}{2M \sin(x)\cos(x)}\cdot \frac{\sin(x)}{x}
    =   \frac{\sin(2Mx)}{M\sin(2x)}\sinc (x)\\
    &=\frac{1}{M}U_{M-1}(\cos(2x)) \sinc (x),
\end{align}
where $U_n$ are the Chebyshev polynomials of the second kind defined by
\begin{equation}\label{eq:Un_def}
    U_n(\cos(x)) = \frac{\sin((n+1)x)}{\sin(x)}
\end{equation}
and, as shown in Lemma~\ref{U_bound_lemma} of Appendix~\ref{Chebyshev_appendix}, obey $|U_n(\cos(\theta))| \leq n+1$ 
for all $\theta\in\mathbb{R}$. Thus $|g_M(x)|\le |\sinc(x)|\le 1$ as claimed. 

Using $|g_M(x)| \leq 1$ on $\Ic_0$, one can bound $I_{00}$ by
\begin{align}
    I_{00} &\leq \int_{\Ic_0 \times \Ic_0} dp \: dq \: \left[\Big(\frac{p}{q}\Big)^{1/4} + \Big(\frac{p}{q}\Big)^{-1/4}\right]
    = \frac{32\epsilon^2}{15M^2}.
\end{align}
Similarly, 
\begin{equation}\label{I_11_estimate}
    I_{11} \leq \int_{\Ic_1\times\Ic_1} dp \: dq \left[\Big(\frac{p}{q}\Big)^{1/4} + \Big(\frac{p}{q}\Big)^{-1/4}\right] \leq
    \frac{2\epsilon^2}{M^2}\Bigg(\frac{1+\epsilon/(\pi M)}{1-\epsilon/(\pi M)}\Bigg)^{1/4},
\end{equation}
where we use the fact that 
\begin{equation}
    \frac{p}{q} \leq \frac{1+\epsilon/(\pi M)}{1-\epsilon/(\pi M)}
\end{equation}
for $p,q \in \Ic_1$.
One can write
\begin{equation}
    \Bigg(\frac{1+\epsilon/(\pi M)}{1-\epsilon/(\pi M)}\Bigg)^{1/4} = \Big(1+\frac{2\epsilon/(\pi M)}{1-\epsilon/(\pi M)}\Big)^{1/4}
    \leq 1 + \frac{\epsilon/(\pi M)}{2(1-\epsilon/(\pi M))}
\end{equation}
by Bernoulli's inequality. 
Clearly, for $M \geq 2$ and $\epsilon \leq \tfrac{\pi}{6}$ we have $\tfrac{\epsilon}{\pi M} \leq \tfrac{1}{12}$ and hence
\begin{equation}
    \Bigg(\frac{1+\epsilon/(\pi M)}{1-\epsilon/(\pi M)}\Bigg)^{1/4}  \leq 1+\frac{\epsilon/(\pi M)}{2(1-\epsilon/(\pi M))}\leq 1+\frac{6\epsilon}{11\pi M}.
\end{equation}
Combining this inequality with \eqref{I_11_estimate} we find
\begin{equation}
    I_{11} \leq \frac{2\epsilon^2}{M^2}\Big(1+\frac{6\epsilon}{11\pi M}\Big).
\end{equation}
Finally, consider $I_{01}$, the integral over $\Ic_0 \times \Ic_1$.
We claim that 
\begin{equation}\label{eq:gMbound_I01}
g_M(p-q)\leq -\frac{2\sinc(3\epsilon)}{\pi}\left(1-\frac{\epsilon}{\pi M}\right)
\end{equation}
holds for $(p,q) \in \Ic_0 \times \Ic_1$, for which one also has 
\begin{equation}\label{eq:qminusp_range}
    \frac{\pi}{2}-\frac{3\epsilon}{2M} \leq q-p \leq \frac{\pi}{2}+\frac{\epsilon}{2M}.
\end{equation} 
To establish the claim, we first rewrite the even function $g_M(x)$ for $x\neq0$ as 
\begin{equation}
    g_M(x) = \frac{\sin(2Mx)}{2M x\cos(x)}=-\frac{\sin(2M(x-\pi/2))}{2M x\sin(x-\pi/2)}=-\frac{\sinc(2M(x-\pi/2))}{x\:\sinc(x-\pi/2)}
\end{equation}
after recalling that $M$ is even, so  
\begin{align}
    g_M\left(\frac{\pi}{2}+y\right)&=-\frac{\sinc(2My)}{\left(\pi/2+y\right)\:\sinc(y)} 
    =-\frac{U_{2M-1}(\cos(y))}{2M\left(\pi/2+y\right)}.
\end{align}
By Lemma~\ref{U_extreme_values_lemma} in Appendix~\ref{Chebyshev_appendix}, the even expression $U_{2M-1}(\cos(y))$ is nonnegative and nonincreasing on $[0,\tfrac{\pi}{2M}]$. It therefore follows that for $0<\epsilon \leq \tfrac{\pi}{6}$ and $y \in [-\tfrac{3\epsilon}{2M}\tfrac{\epsilon}{2M}]$, we have
\begin{equation}
    g_M\left(\frac{\pi}{2}+y\right)\leq - \frac{\sinc(3\epsilon)}{(\pi/2+y)\sinc(3\epsilon/(2M))} \leq -\frac{2\sinc(3\epsilon)}{\pi}\Big(1-\frac{2y}{\pi}\Big),
\end{equation} 
where the last step uses the elementary inequalities $\sinc(x)\leq 1$ and $-1/(1+x)\leq -(1-x)$ for all $x\in\mathbb{R}$. 
Since $y \in [-\tfrac{3\epsilon}{2M},\tfrac{\epsilon}{2M}]$, we deduce
\begin{equation}
    g_M\left(\frac{\pi}{2}+y\right)\leq -\frac{2\sinc(3\epsilon)}{\pi}\left(1-\frac{\epsilon}{\pi M}\right),
\end{equation}
which establishes the claim~\eqref{eq:gMbound_I01}.
Now we can use~\eqref{eq:gMbound_I01} to obtain an upper bound for $I_{01}$ by
\begin{align}
I_{01} &\leq-\frac{2\sinc(3\epsilon)}{\pi}\Big(1-\frac{\epsilon}{\pi M}\Big)\int_{\Ic_0 \times \Ic_1} dp \: dq \:\Big[\Big(\frac{p}{q}\Big)^{1/4} + \Big(\frac{p}{q}\Big)^{-1/4}\Big]\\
\label{I01_bound}
&< -\frac{32\sinc(3\epsilon)\epsilon^{3/4}}{15\pi M^{3/4}}\Big(1-\frac{\epsilon}{\pi M}\Big)\Big[\Big(\frac{\pi}{2}+\frac{\epsilon}{2M}\Big)^{5/4}-\Big(\frac{\pi}{2}-\frac{\epsilon}{2M}\Big)^{5/4}\Big],
\end{align}
where we have dropped the contribution from the first term in the integrand in the second line.
By the mean value theorem, we have 
\begin{equation}
\label{MVT_application}
    0<\begin{cases}
        \alpha y^{\alpha-1} & 0<\alpha<1\\
        \alpha x^{\alpha-1} & \alpha>1
    \end{cases}
    <\frac{y^\alpha-x^\alpha}{y-x}
    <\begin{cases}
        \alpha x^{\alpha-1} & 0<\alpha<1\\
        \alpha y^{\alpha-1} & \alpha>1
    \end{cases}.
\end{equation}
Applying \eqref{MVT_application} with $\alpha=5/4$, we find
\begin{equation}
\label{eq: 5/4_MVT_bound}
    \Big(\frac{\pi}{2}+\frac{\epsilon}{2M}\Big)^{5/4}-\Big(\frac{\pi}{2}-\frac{\epsilon}{2M}\Big)^{5/4} \geq \frac{5\epsilon}{4M}\Big(\frac{\pi}{2}\Big)^{1/4}\Big(1-\frac{\epsilon}{\pi M}\Big)^{1/4},
\end{equation}
and substituting back into \eqref{I01_bound} gives
\begin{equation}
    I_{01} \leq - 
    \frac{8\sinc(3\epsilon)\epsilon^{7/4}}{3 \cdot 2^{1/4}\pi^{3/4} M^{7/4}}
    \left(1-\frac{\epsilon}{\pi M}\right)^{5/4}.
\end{equation}
An application of Taylor's theorem yields $(1-x)^{5/4} \geq 1-\tfrac{5x}{4}$ for all $x \in (0, 1)$. Recalling that $0<\epsilon \leq \tfrac{\pi}{6}$ so that $\sinc(3\epsilon)\geq \sinc(\pi/2)= \tfrac{2}{\pi}$ we obtain
the final bound
\begin{equation}
    I_{01} \leq -
    \frac{4}{3}\left(\frac{2\epsilon}{M\pi}\right)^{7/4} 
    \left(1-\frac{5\epsilon}{4\pi M}\right).
\end{equation}

Recombining the bounds for $I_{00}, I_{01} \textnormal{ and } I_{11}$, we find
\begin{equation}
\label{supremum_bounds}
    \sup \sigma\left(J^{(M)}\right) \geq \frac{2^{11/4}\epsilon^{3/4}}{3\pi^{11/4}}M^{1/4}S\left(\frac{2\epsilon}{\pi M}\right)
\end{equation}
where
\begin{equation}
    S(\eta)=1-\frac{31\pi^2}{80}\eta^{1/4}-\frac{5}{8}\eta-\frac{9\pi^2}{176}\eta^{5/4}.
\end{equation}
thus completing the proof of Lemma~\ref{lemma: lower_bound_supremum}.
\end{proof}

\section{Chebyshev Polynomials}
\label{Chebyshev_appendix}

This appendix concerns the value and distribution of the extrema of 
Chebyshev polynomials of the second kind as defined in~\eqref{eq:Un_def}. The results are presumably known, though we did not find a reference. The following Lemma bounds the absolute value of the extrema of $U_n(x)$. 
\begin{lemma}
\label{U_bound_lemma}
    For $n\in\mathbb{N}$, the Chebyshev polynomials of the second kind are bounded as
    \begin{equation}
        |U_{n-1}(x)|\leq n
    \end{equation}
    for all $x\in[-1,1]$ and furthermore achieve this extremum for $x=\pm 1$.
\end{lemma}
\begin{proof}
    As $U_0\equiv 1$, we may restrict to $n\ge 2$ and $x\in (-1,1)$, noting that 
    $U_{n-1}(\pm 1) = (\pm 1)^{n-1} n$. The derivative of $U_{n-1}(x)$ is given for $x\in (-1,1)$ by
    \begin{equation}
    \label{U_derivative}
    U_{n-1}'(\cos(\theta))=-\frac{n\sin(\theta)\cos(n\theta)-\cos(\theta)\sin(n\theta)}{\sin^3(\theta)},
    \end{equation}
    so the stationary points of $U_{n-1}(x)$ with $x\in (-1,1)\setminus\{0\}$ occur when $x=\cos(\theta)$ with
    \begin{equation}
    \label{stationary_U}
        \tan(n\theta)=n\tan(\theta).
    \end{equation}
    We can rewrite $U_{n-1}(x)^2$ as
    \begin{equation}
        U_{n-1}(x)^2=\frac{\tan^2(n\theta)}{\tan^2(\theta)} \cdot \frac{1+\tan^2(\theta)}{1+\tan^2(n\theta)}
    \end{equation}
    Suppose $x_0=\cos(\theta_0)$ satisfies \eqref{stationary_U}, then  
    \begin{equation}
    U_{n-1}(x_0)^2=n^2\frac{1+\tan^2(\theta_0)}{1+n^2\tan^2(\theta_0)}=
    \frac{n^2}{n^2+(1-n^2)x_0^2}\in (1,n^2),
    \end{equation}
    using trigonometric identities and 
    recallling that $x_0^2\in(0,1)$. Therefore $|U_{n-1}(x_0)|< n$ at any extremum $x_0 \in(-1,1) \setminus \{0\}$.
    \end{proof}

    As well as an understanding of the best upper bound for $U_n(x)$, we would further like to bound the position of the nearest extremum of $U_n(\cos(\theta))$ to $\theta=0$. 
\begin{lemma}
\label{U_extreme_values_lemma}
    For fixed $n\in\mathbb{N}$ with $n\ge 2$ the function $U_{n-1}(\cos(\theta))$ has no extrema for $\theta\in\left(0,\tfrac{\pi}{n}\right)$ and $U_{n-1}(\cos(\theta))$ is 
    nonnegative and nonincreasing on $\left[0, \tfrac{\pi}{n}\right]$.
\end{lemma}
\begin{proof}
Note that $U_{n-1}(\cos(\theta))\to n>0$ as $\theta\to 0$. 
There are no zeros of $\sin$ on $(0,\pi/n)$ so the extrema of $U_{n-1}(\cos(\theta))$ in this interval satisfy 
\begin{equation}
\label{Chebyshev_extrema}
  t_n(\theta):=  \tan(n\theta)-n\tan(\theta)=0.
\end{equation}
We will show that $t_n>0$ on $(0,\pi/(2n))$ and $t_n<0$ on $(\pi/(2n),\pi/n)$, and it is clear that $t_n(\theta)\to\pm\infty$ as $\theta\to\ \pi/(2n)\mp$. Hence there are no zeros of $t_n$ in $(0,\pi/n)$
and therefore no extrema of $U_{n-1}(\cos(\theta))$. Since $U_{n-1}(\cos(\theta))$ is smooth, achieves its global maximum of $n$ at $\theta=0$, and satisfies $U_{n-1}(\cos(\pi/n))=0$, with no extrema in $(0,\pi/n)$, 
it follows that $U_{n-1}(\cos(\theta))$ is nonnegative and nonincreasing on $\left[0, \pi/n\right]$.

It remains to prove the statements made about $t_n$.  
First, for $\theta\in\left(0, \pi/n\right)\setminus\{\pi/(2n)\}$, the derivative of $t_n$ is given by
\begin{equation}
    t_n'(\theta)=n\Big(\sec^2(n\theta)-\sec^2(\theta)\Big),
\end{equation}
which is positive on $\left(0,  \pi/(2n)\right)$ because 
$\sec^2(\theta)$ is a strictly increasing function on the range $\left(0, \tfrac{\pi}{2}\right)$. As $t_n(0)=0$, this implies $t_n>0$ on $\left(0,  \pi/(2n)\right)$. Second, by trigonometric identities, we may write $t_n'$ equivalently as 
\begin{equation}
    t_n'(\theta)=n \frac{\sin([n+1]\theta)\sin([n-1]\theta)}{\cos^2(\theta)\cos^2(n\theta)},
\end{equation}
from which it is clear that the only zero $t_n'$ has on $\left(0, \pi/n\right)$ is at $\theta=\tfrac{\pi}{n+1}$. At this point, one has
\begin{equation}
    t_n\Big(\frac{\pi}{n+1}\Big) = \tan\Big(\frac{n\pi}{n+1}\Big)-n\tan\Big(\frac{\pi}{n+1}\Big)=-(n+1)\tan\Big(\frac{\pi}{n+1}\Big)<0
\end{equation}
using $\tan(\pi-\theta)=-\tan(\theta)$. Now $t_n(\theta)\to -\infty$ as $\theta\to \pi/(2n)+$ and also 
\begin{equation}
    t_n(\theta) = \tan(n\theta)-n\tan(\theta) \leq \tan(\pi)-n\tan\Big(\frac{\pi}{n+1}\Big)=-n\tan\Big(\frac{\pi}{n+1}\Big)<0
\end{equation}
for $\theta \in \left(\tfrac{\pi}{n+1}, \tfrac{\pi}{n}\right)$
by considering an upper bound for each term individually, using the fact that the tangent function is increasing on  $\left[0,\tfrac{\pi}{2} \right)$ and $\left(\tfrac{\pi}{2},\pi\right]$. As $t_n$ is negative at its unique extremum in $(\pi/(2n),\pi/n)$ it follows that 
$t_n<0$ on this interval, completing the proof.
\end{proof}

\section{Proof of Theorems~\ref{thm:blims} and~\ref{thm:PKlimits}}
\label{blims_appendix}

\begin{manualtheorem}{4.1}\label{thm:blims_app}
    Let $M\in\mathbb{N}$ and $j,k\in\mathbb{N}_0$ so that $j+k=2L$ is even and positive. 
    Consider any $u\in\mathbb{R}^{(2M-2)}_>$ and any
    sequences $(v_n)_{n\in\mathbb{N}}$  in $\mathbb{R}_>^j$ and $(w_n)_{n\in\mathbb{N}}$ in $\mathbb{R}_>^k$ whose last and first components obey $v_{n,j}\to 0$ and $w_{n,1}\to \infty$. (In the case $M=1$ one omits $u$; similarly, in the cases $j=0$ resp., $k=0$ one omits $v$ resp., $w$.)
    
    If $j$ (and hence also $k$) is even, then 
    \begin{equation}\label{eq:blimseven_app}
        \liminf_n b^{(M+L)}_{\textnormal{back}}(v_n,u,w_n) \geq b^{(M)}_{\textnormal{back}}(u), \qquad
        \limsup_n b^{(M+L)}_{\textnormal{over}}(v_n,u,w_n) \leq b^{(M)}_{\textnormal{over}}(u)
    \end{equation}
    where we use the shorthand notation
    $(v,u,w)=(v_1,\ldots,v_j,u_1,\ldots, u_{2M-2},w_1,\ldots,w_k)$.
     
    On the other hand, if $j$ (and hence also $k$) is odd, then 
    \begin{equation}\label{eq:blimsodd_app}
        \liminf_n b^{(M+L)}_{\textnormal{back}}(v_n,u,w_n) \geq -1- b^{(M)}_{\textnormal{over}}(u), \qquad
        \limsup_n b^{(M+L)}_{\textnormal{over}}(v_n,u,w_n) \leq -1-b^{(M)}_{\textnormal{back}}(u).
    \end{equation}
    Consequently, the sequence of backflow constants $c^{(M)}_{\textnormal{back}}$ 
    is nondecreasing, and the sequence of overflow constants $c^{(M)}_{\textnormal{over}}$ is nonincreasing; moreover,
    $c^{(M+L)}_\textnormal{back}\geq -1 - c^{(M)}_\textnormal{over}$ and
    $c^{(M+L)}_\textnormal{over}\leq -1 - c^{(M)}_\textnormal{back}$ for all $L\in\mathbb{N}$.    
\end{manualtheorem}

\begin{proof}
Choose any $\boldvec[M]{t}\in\mathcal{T}_{M}$ so that $u=\srsd \boldvec[M]{t}$ (if $M=1$, any element of $\mathcal{T}_{1}$ can be chosen). With this choice, for each $v\in(\mathbb{R}^+_>)^j$ and $w\in(\mathbb{R}^+_>)^k$
there are unique $r\in \mathbb{R}^j$ and $s\in\mathbb{R}^k$ so that
$\langle\boldsymbol{\hat{t}}_{M+L}\rangle=\langle r_1,\ldots,r_j,t_1,\ldots,t_{2M},s_1,\ldots,s_k\rangle\in\mathcal{T}_{M+L}$ satisfies
$\srsd \langle\boldsymbol{\hat{t}}_{M+L}\rangle=(v,u,w)$. Specifically, the components are defined recursively by
\begin{equation}\label{eq:rec_def}
    r_{a-1}=r_{a}+ \frac{r_{a}-r_{a+1}}{v_{a-1}}, \qquad s_{a+1}=s_a+(s_{a}-s_{a-1})w_{a+1}
\end{equation}
with $s_{0}=t_{2M}$, $s_{-1}=t_{2M-1}$, $r_{j+1}=t_1$, $r_{j+2}=t_2$.
It is easily seen from~\eqref{eq:rec_def} that $v_j\to 0+$ implies $r_j\to -\infty$ (and therefore also $r_a\to-\infty$ for each $1\leq a\leq j$), while if $w_1\to+\infty$ then
$s_1\to+\infty$ (and therefore also $s_a\to+\infty$ for each $1\leq a\leq k$).

 For any $\epsilon>0$, we may choose $\psi\in\mathscr{H}_+$
    so that 
    \begin{equation}
    \Delta^{(M)}_{\langle \boldsymbol{t}_M\rangle}(\psi)> \sup_{\varphi\in\mathscr{H}_+}\Delta^{(M)}_{\langle \boldsymbol{t}_M\rangle}(\varphi)- \epsilon ,
    \end{equation}
    whereupon Dollard's lemma~\cite{dollard1969scattering} gives
    \begin{align}
    \label{eq: MultipleDollardLimit}
        \Delta^{(M+L)}_{\langle\boldsymbol{\hat{t}}_{M+L}\rangle}(\psi) &=
        \sum_{a=1}^j (-1)^a \Prob_\psi(X<0|t=r_a) + (-1)^j \Delta^{(M)}_{\langle \boldsymbol{t}_M\rangle}(\psi) + 
        \sum_{a=1}^k (-1)^{k-a} \Prob_\psi(X<0|t=s_a)\nonumber\\
        &\longrightarrow
        \begin{cases} -1 - \Delta^{(M)}_{\langle \boldsymbol{t}_M\rangle}(\psi) & \textnormal{$j$ odd}\\
        \Delta^{(M)}_{\langle \boldsymbol{t}_M\rangle}(\psi) & 
        \textnormal{$j$ even.}
        \end{cases} 
    \end{align}
    as $v_j\to 0+$ and $w_1\to+\infty$.
    As the limit is approached, in the case where $j$ is even,
    it eventually holds that  $\Delta^{(M+L)}_{\langle\boldsymbol{\hat{t}}_{M+L}\rangle}(\psi)>\Delta^{(M)}_{\langle \boldsymbol{t}_M\rangle}(\psi)-\epsilon$
    and thus the inequalities
    \begin{equation}
    \Delta^{(M+L)}_{\langle\boldsymbol{\hat{t}}_{M+L}\rangle}(\psi)> 
    \sup_{\varphi\in\mathscr{H}_+}\Delta^{(M)}_{\langle \boldsymbol{t}_M\rangle}(\varphi)-2\epsilon
    \end{equation} 
    and 
    \begin{equation}
        b^{(M+L)}_{\textnormal{back}}(v,u,w) = \sup_{\psi\in\mathscr{H}_+} \Delta^{(M+L)}_{\langle\boldsymbol{\hat{t}}_{M+L}\rangle}(\psi)> 
    \sup_{\varphi\in\mathscr{H}_+}\Delta^{(M)}_{\langle \boldsymbol{t}_M\rangle}(\varphi)-2\epsilon = b^{(M)}_{\textnormal{back}}(u)-2\epsilon
    \end{equation}
    are eventually true. Replacing $v$ and $w$ by sequences as in the statement and using the fact that $\epsilon>0$ was arbitrary, we have proved the first statement in~\eqref{eq:blimseven_app}; the second part is proved analogously.

    If $j$ is odd, we argue in a similar way that for any $\epsilon>0$,
    \begin{equation}
    b^{(M+L)}_{\textnormal{back}}(v,u,w)=\sup_{\psi\in\mathscr{H}_+}\Delta^{(M+L)}_{\langle\boldsymbol{\hat{t}}_{M+L}\rangle}(\psi)> 
    -1-
    \inf_{\varphi\in\mathscr{H}_+}\Delta^{(M)}_{\langle \boldsymbol{t}_M\rangle}(\varphi)-2\epsilon = -1-b^{(M+L)}_{\textnormal{over}}(u)-2\epsilon.
    \end{equation}
    is eventually true as the limit is taken, thus giving the first statement in~\eqref{eq:blimsodd_app}; the second is again proved analogously.

    To establish the monotonicity of the backflow constants, for any $M\in\mathbb{N}$, let $\epsilon>0$ be arbitrary and choose $u\in\mathbb{R}^{(2M-2)}_>$ 
    so that $b^{(M)}_\textnormal{back}(u)> c^{(M)}_\textnormal{back}-\epsilon$. Then~\eqref{eq:blimseven_app} with $L=1$ and $j=2$, $k=0$ 
    (or $j=0$, $k=2$) implies that $c^{(M+1)}_\textnormal{back}>c^{(M)}_\textnormal{back}-\epsilon$
    and taking $\epsilon\to 0+$, we deduce that $c^{(M+1)}_\textnormal{back}\ge c^{(M)}_\textnormal{back}$; 
    an analogous argument proves that $c^{(M+1)}_\textnormal{over}\le c^{(M)}_\textnormal{over}$.
    Finally, for any $L\in \mathbb{N}$, we may use~\eqref{eq:blimsodd_app} with e.g.,~$j=1$ and $k=2L-1$ to 
    prove $c^{(M+L)}_\textnormal{back}\ge -1-c^{(M)}_\textnormal{over}$ 
    and $c^{(M+L)}_\textnormal{over}\ge -1-c^{(M)}_\textnormal{back}$ by similar means.
\end{proof}

Recall from the main text (Definition~\ref{defn: Li_defn}) that the Painlev\'{e}-Kuratowski lower closed limit $\Li A_\alpha$ 
of a net of sets $(A_\alpha)_{\alpha \in I}$ in a topological space $X$ 
is the set of points $p\in X$ such that every neighbourhood of $p$ is eventually intersected by the $A_\alpha$. A slight modification of Theorem~\ref{thm:blims_app} gives the following.  
\begin{manualtheorem}{4.2} \label{thm:PKlimits_app}
    For $M\in\mathbb{N}$ and $\boldvec[M]{t} \in \mathcal{T}_M$, one has 
    \begin{equation}
         -1-\mathcal{N}\left(B^{(M)}_{\boldvec[M]{t}}\right) \subseteq \Li_{T_\pm \rightarrow \pm\infty } \mathcal{N}\left(B^{(M+1)}_{\langle T_-,\boldsymbol{t}_M,T_+\rangle}\right)
    \end{equation}
    and 
    \begin{equation}
        \mathcal{N}\left(B^{(M)}_{\boldvec[M]{t}}\right) \subseteq 
        \Li_{\substack{T,T'\rightarrow -\infty\\ T<T'} } \mathcal{N}\left(B^{(M+1)}_{\langle T,T',\boldsymbol{t}_M\rangle}\right).
    \end{equation} 
\end{manualtheorem} 
\begin{proof}
For the first inclusion, fix $\boldvec[M]{t}\in\mathcal{T}_M$ and let $T_-<0<T_+$ so that 
$\langle T_-,\boldsymbol{t}_M,T_+\rangle\in\mathcal{T}_{M+1}$. As in the proof of theorem~\ref{thm:blims_app}, take $j=k=1$ and let $r_1=T_-, s_1=T_+$. By \eqref{eq: MultipleDollardLimit}, 
\begin{equation}
    \Delta_{\langle T_-,\boldsymbol{t}_M,T_+\rangle}^{(M+1)}(\psi) \longrightarrow -1-\Delta_{\boldvec[M]{t}}^{(M)}(\psi) \qquad \text{ for any } \psi \in \mathcal{H}_+
\end{equation}
as $T_\pm\to\pm\infty$.
So for each $\lambda \in \mathcal{N}(B^{(M)}_{\boldvec[M]{t}})$ and $\epsilon>0$, one can find $T_-<0<T_+$ such that 
\begin{equation}
    \left[ -1-\lambda - \epsilon, -1-\lambda+\epsilon \right] \cap \mathcal{N}\left(B^{(M+1)}_{\langle T_-,\boldsymbol{t}_M,T_+\rangle}\right) \neq \emptyset,
\end{equation}
whereupon the first result follows by the definition of the lower closed limit. The second follows in a similar way, taking $j=2,k=0$ and $r_1=$ and $r_1=T,r_2=T'$. 
\end{proof}

\section{Proof of Proposition~\ref{prop: spectral_recipe_app}}
\label{spectral_recipe_appendix}

\begin{prop}{5.1}
\label{prop: spectral_recipe_app}
Let $A$ be a bounded self-adjoint operator on Hilbert space $\mathscr{H}$ and let $(\chi_n)_{n \in \mathbb{N}} \subset \mathscr{H}$ be a sequence of linearly independent vectors with dense span. Define sequences of self-adjoint $N\times N$ matrices 
$(A^{[N]})_{N\in\mathbb{N}}$ and $(P^{[N]})_{N\in\mathbb{N}}$
with matrix elements
\begin{equation}
    A^{[N]}_{mn}=\ip{\chi_m}{A\chi_n} , \qquad P^{[N]}_{mn}=\ip{\chi_m}{\chi_n}
\end{equation} 
for $1\leq m,n\leq N$. Then $\sigma(A^{[N]},P^{[N]})\subseteq \mathcal{N}(A)$, and 
$\max  \sigma(A^{[N]},P^{[N]})$ (resp., $\min  \sigma(A^{[N]},P^{[N]})$) is a bounded nondecreasing (resp., nonincreasing) sequence with
\begin{align}
        \max \sigma(A) &= \lim_{N \rightarrow \infty} \max \sigma(A^{[N]},P^{[N]}) \nonumber\\ 
        \min \sigma(A) &= \lim_{N \rightarrow \infty} \min \sigma(A^{[N]},P^{[N]}).
    \end{align}
For each $N \in \mathbb{N}$, suppose $v^{(N)} \in \mathbb{C}^N$ is a generalised eigenvector obeying
\begin{equation}
    A^{[N]}v^{(N)} = \lambda_N P^{[N]}v^{(N)}, \qquad v^{(N)\dagger} P^{[N]} v^{(N)}=1
\end{equation}
and define $\psi^{(N)} \in \mathcal{H}$ by
\begin{equation}
    \psi^{(N)}=\sum_{n=1}^N v^{(N)}_n \chi_n.
\end{equation}
If $\psi^{(N)}\to \psi \in \mathcal{H}$ in norm then
\begin{enumerate}
        \item \label{st: gen_eigvals_app}the sequence of generalised eigenvalues $(\lambda_N)_{N \in \mathbb{N}}$ converges;
        \item \label{st: gen_eigvecs_app}the sequence of vectors $(\psi^{(N)})_{N \in \mathbb{N}}$ is a Weyl sequence for $\lambda = \lim_{N \to \infty} \lambda_N$, i.e., $\|\psi^{(N)}\|=1$ and $\|(A-\lambda I)\psi^{(N)}\|\to 0$;
        \item \label{st: eigenvector_app}the limiting vector $\psi$ is an eigenvector for $A$ with eigenvalue $\lambda$.
    \end{enumerate}
\end{prop}

\begin{proof} 
If $\lambda\in\sigma(A^{[N]},P^{[N]})$ with $A^{[N]}v=\lambda P^{[N]}v$ 
for nonzero $v$, then 
\begin{equation} 
\lambda= \frac{v^\dagger A^{[N]}v}{v^\dagger P^{[N]}v} = \frac{\ip{\phi}{A\phi}}{\ip{\phi}{\phi}},
\end{equation}
where $\phi = \sum_{r=1}^N v_r \chi_r$, so $\lambda\in\mathcal{N}(A)$, which is a bounded set. Let $S_N=\text{span}(\chi_1,\dots,\chi_N)$, then monotonicity of the sequences of maximum and minimum generalised eigenvalues is guaranteed because $S_N\subset S_{N+1}$. Finally,
as $\bigcup_N S_N$ is dense in $\mathscr{H}$ and $A$ is bounded, we easily compute
\begin{align}
    \max\sigma(A)&= \sup \mathcal{N}(A)=\lim_{N \rightarrow \infty} \sup_{\phi \in S_N\setminus\{0\}} \frac{\ip{\phi}{A \phi}}{\ip{\phi}{\phi}}=\lim_{N \rightarrow \infty} \sup_{v \in \mathbb{C}^N\setminus\{0\}} \frac{v^\dagger A^{[N]} v}{v^\dagger P^{[N]}v}\nonumber \\ 
    &=
    \lim_{N\to\infty}\max \sigma(A^{[N]},P^{[N]}),
\end{align}
where in the final equality we have made use of the (generalised) Courant--Fischer min-max principle -- see, e.g., theorem 2.1 of \cite{MR3445565}. The proof for the minimum of $\sigma(A)$ is analogous. 

For each $N\in\mathbb{N}$, let $Q_N$ be the projection onto $S_N$. First we show that $\psi^{(N)}$ is an eigenvector of $Q_N A$. Since $Q_N A \psi^{(N)} \in S_N$, we must have
\begin{equation}
\label{eq: Q_NApsi_expansion}
    Q_N A \psi^{(N)} = \sum_{k=1}^N w_k^{(N)}\chi_k
\end{equation}
for some $w^{(N)} \in \mathbb{C}^N$. For $1 \leq j \leq N$, the inner product of \eqref{eq: Q_NApsi_expansion} with $\chi_j$ can be written as both
\begin{equation}
    \ip{\chi_j}{Q_N A \psi^{(N)}} = \sum_{k=1}^N \ip{\chi_j}{\chi_k} w^{(N)}_k =\left(P^{[N]} w^{(N)}\right)_j 
\end{equation}
and
\begin{equation}
    \ip{\chi_j}{Q_N A \psi^{(N)}} = \sum_{k=1}^N v^{(N)}_k \ip{\chi_j}{A \chi_k} = \left(A^{[N]} v^{(N)}\right)_j = \lambda_N \left(P^{[N]} v^{(N)}\right)_j.
\end{equation}
By the linear independence of the $\chi_j$, $P^{[N]}$ has trivial kernel and we find $w^{(N)}=\lambda_N v^{(N)}$, from which it follows that
\begin{equation}
    Q_N A \psi^{(N)} = \lambda_N \psi^{(N)}.
\end{equation}
Note that $\|\psi^{(N)}\|=1$ because of the normalisation condition on $v^{(N)}$.
Property~\ref{st: gen_eigvals_app} follows from the fact that
\begin{equation}
    \lambda_N = \ip{Q_N \psi^{(N)}}{A \psi^{(N)}}
\end{equation}
converges to some $\lambda \in \mathbb{R}$ by the norm convergence of the $\psi^{(N)}$ to $\psi$ and boundedness of $A$. To prove property~\ref{st: gen_eigvecs_app}, note that
\begin{equation}
    \|(A-\lambda I)\psi^{(N)}\| \leq \|(A-\lambda_N I)\psi^{(N)}\|+|\lambda_N-\lambda | 
\end{equation}
so since $\lambda_N\to\lambda$, it suffices to note that
\begin{align}
    \|(A-\lambda_N I)\psi^{(N)}\| = \|(A-Q_N A)\psi^{(N)}\| \leq \|(I-Q_N)A\psi\|+\|(I-Q_N)A\| \|\psi^{(N)}-\psi\|\to 0
\end{align}
by the density of $\bigcup_N S_N$ in $\mathcal{H}$ and convergence of $(\psi^{(N)})_{N \in \mathbb{N}}$ proving property~\ref{st: gen_eigvecs_app}. The final property follows from
\begin{equation}
    \|(A-\lambda I)\psi\| \leq \|(A-\lambda I)\psi^{(N)}\| + (\|A\|+|\lambda|) \|\psi^{(N)} - \psi\| \to 0
\end{equation}
and hence $\psi$ is an eigenvector of $A$ with eigenvalue $\lambda$.
\end{proof}

\section{Exact Matrix Element Integrals}\label{appendix:exactmatrixelements}
 
In this appendix we compute closed forms for the matrix elements of an arbitrary backflow operator with respect to the dense sequence $(\psi_{n,a,\delta})_{n \in \mathbb{N}}$, given in \eqref{eq: almost_polynomial_vectors}. Let $C^{(M)}_{\boldvec[M]{t}}=W^\ast B^{(M)}_{\boldvec[M]{t}} W$, where $W$ is the unitary implementing the change of variables $q\mapsto q^{1/2}$, 
with action
\begin{equation}
    \left(C^{(M)}_{\boldvec[M]{t}} \varphi\right)(p)=-\frac{1}{4\pi i}\int_0^\infty dq \sum_{k=1}^M\frac{e^{it_{2k}(p-q)}-e^{it_{2k-1}(p-q)}}{p-q}P(p,q)\varphi(q),
\end{equation}
where $P(p,q)=(p/q)^{1/4}+(p/q)^{-1/4}$. As shown in \eqref{eq: FM_as_sum}, one can write an arbitrary multiple backflow operator as a sum of single backflow operators. Hence it suffices to compute the matrix elements of $C_{\langle s_1,s_2\rangle}^{(1)}$ with respect to $(\psi_{n,a,\delta})_{n=0}^\infty$ for general $s_1 < s_2 $. For $n,m \in \mathbb{N}$ we have
\begin{align}
    \langle \psi_{n,a,\delta} | C^{(1)}_{\langle s_1,s_2\rangle} \psi_{m,a,\delta} \rangle =&-\frac{E_n(a,\delta)E_m(a,\delta)}{4\pi i} \int_0^\infty dp \int_0^\infty dq \: p^{n+\delta} q^{m+\delta}P(p,q)\nonumber\\
        &\qquad\times\frac{e^{is_2(p-q)}-e^{is_1(p-q)}}{p-q} e^{-a(p+q)} \nonumber\\
    =& -\frac{E_n(a,\delta)E_m(a,\delta)}{4\pi}\int_0^\infty dp \int_0^\infty dq \: p^{n+\delta}q^{m+\delta}P(p,q)\nonumber\\
    &\qquad \times\int_{s_1}^{s_2} dt \: e^{it(p-q)}e^{-a(p+q)} \nonumber\\
\label{eq: mat_element_computation}
    =& -\frac{E_n(a,\delta)E_m(a,\delta)}{4\pi}\int_{s_1}^{s_2} dt \: (I_+(n,m, \delta ; t)+I_-(n,m,\delta; t)),
\end{align}
where we have interchanged the order of integration noting that the integral converges absolutely. Here, $E_n(a,\delta)$ is the normalization constant given in \eqref{eq: normalisation_const} and 
\begin{align}
    I_\pm(n,m,\delta;t)&=\int_0^\infty dp \int_0^\infty dq \:p^{n+\delta\pm1/4} q^{m+\delta\mp 1/4}e^{-(a- it)p}e^{-(a+ it)q}\nonumber\\
    &= \Gamma(n+\delta\pm 1/4+1)\Gamma(m+\delta\mp 1/4+1)(a- it)^{-n-\delta\mp 1/4-1}(a+ it)^{-m-\delta\pm 1/4 -1} .
\label{eq: Iplusminus} 
\end{align}
In the definition of $I_\pm$, the branches for noninteger powers are fixed so that $z^\alpha$ is real and positive for $z>0$ and the cut is taken along the negative real axis. Combining \eqref{eq: mat_element_computation} with \eqref{eq: Iplusminus}, we find that
\begin{equation}
\label{single_matrix_element_breakdown2}
   \langle \psi_{n,a,\delta} |C^{(1)}_{\langle s_1,s_2\rangle} \psi_{m,a,\delta}\rangle= -\frac{1}{4\pi}\Big[ D_{nm}(a,\delta)J(s_1,s_2; \alpha_n^+,\alpha_m^-; a)+D_{mn}(a,\delta)J(s_1,s_2;\alpha_n^-,\alpha_m^+; a) \Big],
\end{equation}
where $\alpha^\pm_n=n+\delta\pm\tfrac{1}{4}$ with
\begin{equation}
    D_{nm}(a,\delta)=\pi \sqrt{2}\frac{(2a)^{n+m+2\delta+1}\sqrt{B_\textnormal{diag}(n+\delta+1/2)B_\textnormal{diag}(m+\delta+1/2)}}{B(n+\delta+1/2,3/4)B(m+\delta+1/2,1/4)},
\end{equation}
in which $B$ is the beta function, $B_\textnormal{diag}(x)=B(x,x)$,
and
\begin{equation}\label
{eq: Jdefn}
    J(s_1,s_2; \alpha,\beta; a)=\int_{s_1}^{s_2}dt \: (a - it)^{-\alpha-1}(a+ it)^{-\beta-1}.
\end{equation}

As shown below, the closed form expression for $J$ is given 
for $s_1,s_2\in\mathbb{R}\setminus\{0\}$ by
\begin{equation}
\begin{aligned}
    J(s_1,s_2;\alpha,\beta;a)=\begin{cases}
        J_+(s_1,s_2;\alpha,\beta;a) & s_1 s_2 > 0\\
        J_-(s_1,s_2;\alpha,\beta;a) & s_1<0< s_2 ,
    \end{cases}
\end{aligned}
\end{equation} 
where 
\begin{equation}\label{eq:Jplus} 
    J_+(s_1,s_2; \alpha,\beta; a)=-\frac{i}{(2a)^{\alpha+\beta+1}}\left[e^{-i\pi\text{sgn}(s)(\beta+1)}B\left(\alpha+\beta+1,-\beta; \frac{2a}{a-s i}\right)\right]_{s=s_1}^{s=s_2},
\end{equation}
and
\begin{equation}\label{eq:Jminus}
J_-(s_1,s_2;\alpha,\beta;a) = J_+(s_1,s_2;\alpha,\beta;a) + 
\frac{2\pi}{(2a)^{\alpha+\beta+1}}\frac{\Gamma(\alpha+\beta+1)}{\Gamma(\alpha+1)\Gamma(\beta+1)}.
\end{equation}

To obtain these closed form expressions, first make the transformation of variable $t \mapsto 2at$ in~\eqref{eq: Jdefn} so that 
\begin{equation}
\label{eq: J_integral}
    J(s_1,s_2;\alpha,\beta; a)=(2a)^{-\alpha-\beta-1}\int_{s_1/(2a)}^{s_2/(2a)}dt\: \Big(\frac{1}{2}-it\Big)^{-\alpha-1} \Big(\frac{1}{2}+it\Big)^{-\beta-1}.
\end{equation} 
Next, recall that
the incomplete beta function $B(\mu,\nu;z)$ is given for $\mu>0$, $\nu\in\mathbb{R}$ as a holomorphic function in $z\in\mathbb{C} \setminus ((-\infty, 0] \cup [1, \infty))$ by
\begin{equation}
\label{eq: inc_beta_integral_rep}
    B(\mu,\nu; z)= \int_0^z dv \: v^{\mu-1}(1-v)^{\nu-1},
\end{equation}
where the integral is taken over any contour avoiding $(-\infty,0] \cup [1,\infty)$. Thus, for $t\in \mathbb{R}\setminus\{0\}$,
\begin{align}
    \frac{d}{dt} B\left(\mu,\nu;(1/2-it)^{-1}\right) &= 
    i \left(-\frac{1}{2}-it\right)^{\nu-1}\left(\frac{1}{2}-it\right)^{-\mu-\nu}\nonumber\\
    &= ie^{-i\pi(\nu-1)\sgn(t)}
    \left(\frac{1}{2}+it\right)^{\nu-1}\left(\frac{1}{2}-it\right)^{-\mu-\nu}.
\end{align}
Setting $\mu =\alpha+\beta+1$ and $\nu=-\beta$, we may
evaluate~\eqref{eq: J_integral} in cases where $\alpha+\beta>-1$ and
$s_1,s_2\in\mathbb{R}\setminus\{0\}$ using
the fundamental theorem of calculus to obtain
$J(s_1,s_2;\alpha,\beta;a)=
J_+(s_1,s_2;\alpha,\beta;a)$, where $J_+$ is given by~\eqref{eq:Jplus}.

In the case where $s_1<0$ and $s_2>0$, we can decompose the 
integral $\int_{s_1/(2a)}^{s_2/(2a)}$ in~\eqref{eq: J_integral}
as $\int_{s_1/(2a)}^{s_2/(2a)} = \int_{-\infty}^\infty -\int_{s_2/(2a)}^{\infty}-\int_{-\infty}^{s_1/(2a)}$.
The first integral can be evaluated in closed form (see (5.12.8) in~\cite{NIST:DLMF}) giving~\eqref{eq:fstint} below, while the others are limiting cases of the result just proved. Noting that
\begin{equation}
\lim_{s\to\pm \infty} B\left(\alpha+\beta+1,-\beta; \frac{2a}{a-s i}\right)= 0,
\end{equation}
the overall effect is to add 
\begin{equation}\label{eq:fstint}
\int_{-\infty}^{\infty} dt \: (a - it)^{-\alpha-1}(a+ it)^{-\beta-1}= 
\frac{2\pi}{(2a)^{\alpha+\beta+1}}\frac{\Gamma(\alpha+\beta+1)}{\Gamma(\alpha+1)\Gamma(\beta+1)}
\end{equation}
to $J_+(s_1,s_2; \alpha,\beta; a)$, thus obtaining $J_-(s_1,s_2; \alpha,\beta; a)$ as given in~\eqref{eq:Jminus}.

Combining \eqref{single_matrix_element_breakdown2},\eqref{eq:Jplus}, \eqref{eq:Jminus} with the decomposition of an $M$-fold backflow operator we find
\begin{align}
    \ip{\psi_{n,a,\delta}}{C^{(M)}_{\boldvec[M]{t}}\psi_{m,a,\delta}}&=-\frac{1}{4\pi}\sum_{k=1}^M\left[D_{nm}(a,\delta)J(t_{2k-1},t_{2k},\alpha_n^+,\alpha_m^-;a)\right.\nonumber\\
    &\qquad +\left.D_{mn}(a,\delta)J(t_{2k-1},t_{2k},\alpha_n^-,\alpha_m^+;a)\right]
\end{align}

The incomplete beta function is computed using the hypergeometric representation
\begin{equation}
\label{eq:beta_inc_to_hypergeom}
    B(\mu,\nu;z)=\frac{z^\mu}{\mu}{}_2F_1(\mu,1-\nu;\mu+1;z).
\end{equation}
given in Section 8.17 of \cite{NIST:DLMF}. We also make
use of the identity $J(s_1,s_2; \alpha,\beta; a) = \overline{J(s_1, s_2; \beta, \alpha; a)}$,  readily seen from~\eqref{eq: J_integral}.

\section{Error Analysis}
\label{appendix: error analysis}

The minimum eigenvalue of $P^{[N]}$ tends to zero as $N$ increases, so the generalised eigenvalue problem is increasingly ill-conditioned and requires high-precision calculation to obtain accurate results.
The following corollary shows how the numerical precision of a pair of Hermitian matrices affects the accuracy of the associated generalised eigenvalues.
\begin{corollary}
\label{cor: gen_eig_error} 
    Let $A,B$ be $n$-dimensional Hermitian matrices with $B$ positive definite, and let $\lambda_1\leq\cdot\cdot\cdot\leq\lambda_n$ be the generalised eigenvalues of the matrix pair $(A,B)$. If $n$-dimensional Hermitian perturbations $\Delta A,\Delta B$ obey 
    \begin{equation}\label{eq:Cmax}
        C_\textnormal{max} := n(1+\max(|\lambda_1|,|\lambda_n|)^2)\frac{\|\Delta A\|_\infty^2+\|\Delta B\|_\infty^2}{\min \sigma(B)^2}<\frac{1}{2},
    \end{equation}
    where $\lambda_\textnormal{max}=\max(|\lambda_1|,|\lambda_n|)$ and $\|\cdot \|_\infty$ is the elementwise maximum modulus matrix norm, then the matrix pair $(A+\Delta A,B+\Delta B)$ is definite,
    with generalised eigenvalues $\tilde{\lambda}_1\leq\cdot\cdot\cdot\leq \tilde{\lambda}_n$ satisfying
    \begin{equation}\label{eq:useful_gen_eig_est}
        |\lambda_i-\tilde{\lambda}_i|^2 \leq \frac{C_\textnormal{max}}{1-2C_\textnormal{max}}(1+2\lambda_i^2)
    \end{equation}
    for all $1 \leq i \leq n$.
\end{corollary}
\begin{proof}
First note that $C_\textnormal{max}<1/2$ implies that 
$\|\Delta A\|_\infty^2+\|\Delta B\|_\infty^2<\min \sigma(B)^2\le \gamma(A,B)^2$, using the trivial bound
\begin{equation}
        \gamma(A,B) \geq \min_{\|x\|_2=1} x^\dagger B x = \min \sigma(B).
    \end{equation} 
Then by Corollary VI.3.3 in \cite{MR1061154}$(A+\Delta A,B+\Delta B)$ is definite and 
\begin{equation}
\label{eq: stw_sun_bound}
    \frac{|\lambda_i-\tilde{\lambda}_i|}{\sqrt{(1+\lambda_i^2)(1+\tilde{\lambda}_i^2)}} \leq \frac{\sqrt{\|\Delta A\|_\textnormal{op}^2+\|\Delta B\|_\textnormal{op}^2}}{\gamma(A,B)},
\end{equation} 
where $\lambda_1\leq\cdots\leq\lambda_n$ and $\tilde{\lambda}_1\leq\cdots\leq\tilde{\lambda}_n$ are the generalised eigenvalues of $(A,B)$ and $(A+\Delta A,B+\Delta B)$ respectively. We proceed to obtain~\eqref{eq:useful_gen_eig_est} by estimating and rearranging \eqref{eq: stw_sun_bound}.

Applying the inequality
\begin{equation}
    n^{-1/2}\|M\|_\textnormal{op}\le \|M\|_\infty :=\max_{1\leq i,j\leq n}|M_{ij}|,
\end{equation}
for $n$-dimensional square matrices (see~\cite{HornJohnson:vol1}, pp. 313-314) and \eqref{eq: stw_sun_bound}
\begin{equation}
\label{eq: lambda_diff_bound}
    |\lambda_i-\tilde{\lambda}_i|^2 \leq (1+\tilde{\lambda}_i^2)C_i,
\end{equation}
where
\begin{equation}
    C_i = n(1+\lambda_i^2)\frac{\|\Delta A\|_\infty^2 + \|\Delta B\|_\infty^2}{\min \sigma(B)^2}.
\end{equation}
Now note that $\tilde{\lambda}_i^2  \leq (|\lambda_i| +|\lambda_i-\tilde{\lambda}_i|)^2  \leq 2(\lambda_i^2+|\lambda_i-\tilde{\lambda}_i|^2)$ (using $(a+b)^2\leq 2(a^2+b^2)$). Combining with \eqref{eq: lambda_diff_bound}, we find
\begin{equation}
    |\lambda_i-\tilde{\lambda}_i|^2 \leq C_i (1+2\lambda_i^2+2|\lambda_i-\tilde{\lambda}_i|^2).
    \end{equation}
Rearranging the above, we obtain
\begin{equation}
    |\lambda_i - \tilde{\lambda}_i|^2 \leq \frac{C_i (1+2\lambda_i^2)}{1-2C_i} \leq \frac{C_\textnormal{max}(1+2\lambda_i^2)}{1-2C_\textnormal{max}},
\end{equation}
by making use of the fact that $x \mapsto 2x/(1-2x)$ is monotonically increasing on $[0,1/2]$ and where $C_\textnormal{max} = \max_i C_i$ has the closed form~\eqref{eq:Cmax} with $1-2C_\textnormal{max}>0$. 
\end{proof}

In our application, using~\eqref{eq:Pdef}, numerical experiments showed that $\min \sigma(P^{[N]}(a,\delta))\approx 10^{2-0.94N}$ independently of $\delta$. Further, any generalised eigenvalue $\lambda \in \sigma(C^{(M)[N]},P^{[N]})\subseteq \mathcal{N}(C^{(M)})$ obeys $|\lambda| \leq 1+(M-1)c_\textnormal{BM}\leq M$ by Theorem~\ref{thm: spectral_outer_bounds}. Letting $A=C^{(M)[N]}(a,\delta)$ and $B=P(a,\delta)^{[N]}$ in corollary~\ref{cor: gen_eig_error}, we find that condition~\eqref{eq:Cmax} is satisfied if
the maximum error $\epsilon$ in the matrix elements obeys
\begin{equation}
\epsilon \lessapprox \frac{10^{2-0.94N}}{2\sqrt{(N+1)(1+M^2)}}.
\end{equation}
As we study $1\le M\le 4$, this constraint is satisfied, with some room to spare, by 
computing the matrix elements to $N+30$ digits of precision, whereupon the maximum error $\Delta \lambda_N$ on the elements of $\sigma(C^{(M)[N]},P^{[N]})$ is bounded by
\begin{align}
    |\Delta \lambda_N| &\lessapprox 10^{-0.06N-32} \sqrt{2(N+1)}\sqrt{(1+M^2)(1+2M^2)} 
\end{align}
For the $N=500$, $1 \leq M \leq 4$ calculations we find it is possible -- in theory -- to calculate the generalised eigenvalues to $|\Delta \lambda|\leq 7.5\times 10^{-60}$ or around $59$ digits. 

To calculate the matrix elements of $C^{(M)}(a,\delta)$ and $P(a,\delta)$ to high precision, we make use of the Python 3.11 library python-flint (a wrapper for the FLINT package~\cite{FLINT}) whose function {\tt good} automatically increases working precision to produce results accurate to the required precision.
In particular, python-flint was used to calculate the incomplete beta function in terms of hypergeometric functions, implemented as described in~\cite{MR4000439}.  The generalised eigenvalue spectrum $\sigma(C^{(M)[N]}(a,\delta),P^{[N]}(a,\delta))$ was computed numerically using Maple\textsuperscript{TM} 2023, using $N+20$ digits of precision. Although Corollary~\ref{cor: gen_eig_error} gives a bound on the theoretical error for the generalised eigenvalues, it does not take account of numerical error accumulated in the generalised eigenvalue solver. In practice, we found that, when calculating the first $N_\textnormal{max}$ generalised eigenvalues, the $N^\text{th}$ generalised eigenvalue is accurate to around $N_\textnormal{max}+3-0.91N$ digits. This means for $N_\textnormal{max}=500$, the final eigenvalue is accurate to $48$ digits, rather than the $59$ digits estimated above, but for $N<N_{\textnormal{max}}$, the $N$'th generalised eigenvalue is accurate to higher precision, with over $90$ digits expected for $N\leq 450$.

\section{Richardson Accelerator}
\label{appendix: Richardson Accelerator}
 
Let $l^\textnormal{conv}$ be the vector space of convergent real-valued sequences in $\ell^\infty(\mathbb{N})$. 
\begin{defn}
    A map $A: l^\textnormal{conv} \rightarrow l^\textnormal{conv}$ is called a \emph{sequence accelerator} for $x \in l^\textnormal{conv}$ if $Ax$ and $x$ have the same limit and
    \begin{equation}
       \frac{(A x)_n - x_\infty}{x_n - x_\infty} \longrightarrow 0.
    \end{equation}
    as $n\to\infty$, where $x_\infty=\lim_{n\to\infty}x_n$.
\end{defn}
Richardson accelerators provide a particularly useful class of sequence accelerators; see chapter 1 of~\cite{MR2604135} for a modern overview and~\cite{Sidi} for detailed theory and generalisations.
\begin{defn}
    The \emph{Richardson accelerator} of order $\gamma$ is the linear
    map $R_\gamma$ on $l^\textnormal{conv}$ with action 
    \begin{equation}
        (R_\gamma x)_n = \frac{2^\gamma x_{2n}-x_n}{2^\gamma - 1}.
    \end{equation}
\end{defn}
It is clear that $R_\gamma x$ is convergent with the same limit as $x$. 
Moreover, every sequence $(n^{-\alpha})_{n\in\mathbb{N}}$ is an eigenvector of $R_\gamma$, with eigenvalue $(2^{\gamma-\alpha}-1)/(2^\gamma-1)$, while for any
sequence $x$ with $|x_n|\le n^{-\alpha}$ we have a bound
\begin{equation}
     |(R_\gamma x)_n|  \le   \frac{2^{\gamma-\alpha}+1}{2^\gamma-1} n^{-\alpha}.
\end{equation} 
As $R_\gamma$ is linear,   
if 
\begin{equation}
    x_n = a + \frac{b}{n^\gamma}  + \epsilon_n, \qquad |\epsilon_n|\le \frac{C}{n^\alpha}
\end{equation}
for constants $a$, $b$, $C$, and $\alpha>\gamma$, then
\begin{equation}
R_\gamma x_n = a + (R_\gamma\epsilon)_n, \qquad |(R_\gamma\epsilon)_n|\le 
 \frac{2^{\gamma-\alpha}+1}{2^\gamma-1} \frac{C}{n^{\alpha}}. 
\end{equation}
and it is clear that the rate of convergence has been improved from $O(n^{-\gamma})$ to $O(n^{-\alpha})$; moreover, the error terms are suppressed if
\begin{equation}
    \gamma>1-\log_2 (1-2^{-\alpha}).
\end{equation}
The Richardson acceleration can be iterated, but a significant disadvantage is that the number of the available terms in the resulting sequence is approximately halved on each iteration.
Alternatively, one can generalise the Richardson method as follows -- see~\cite{Sidi} for yet further generalisations and discussion. 
Suppose, for some integer $k\ge 2$, that 
\begin{equation}\label{eq:xn}
    x_n =  \sum_{j=1}^k\frac{a_j}{n^{\gamma_j}}  + \epsilon_n, \qquad |\epsilon_n|\le \frac{C}{n^{\gamma_{k+1}}}
\end{equation}
where $0=\gamma_1<\gamma_2<\cdots <\gamma_k<\gamma_{k+1}$, $a_1,\ldots,a_k$ and $C$ are constant real numbers,
so $x_n\to a_1$ as $n\to\infty$. Let
$1=r_1<r_2<\ldots <r_k$ be integers and write $\boldsymbol{\gamma}=(\gamma_1,\ldots,\gamma_k)$, $\boldsymbol{r}=(r_1,\ldots,r_k)$. Let $v\in\mathbb{R}^k$ be the unique solution (the relevant generalised Vandermonde determinant is nonvanishing by the remark at the end of~\cite{RobbinSalamon:2000}) to 
\begin{equation}
\sum_{j=1}^k \frac{v_j}{(r_j)^{\gamma_i}} = \begin{cases}
1 & i=1 \\ 0 & i\neq 0
\end{cases}. 
\end{equation}
Then
\begin{equation}
(R_{\boldsymbol{\gamma},\boldsymbol{r}}x)_n := \sum_{j=1}^k v_j x_{nr_j}
= a_1 + (R_{\boldsymbol{\gamma},\boldsymbol{r}}\epsilon)_n,  \qquad 
    |(R_{\boldsymbol{\gamma},\boldsymbol{r}}\epsilon)_n|\le \frac{C}{n^{\gamma_{k+1}}}\sum_{j=1}^k \frac{|v_j|}{r_j^{\gamma_{k+1}}}.
\end{equation}
The  accelerator $R_{\boldsymbol{\gamma},\boldsymbol{r}}$ removes $k-1$ power-law terms, while reducing the 
number of available terms by a factor of approximately $r_k$. Provided that $r_k<2^{k-1}$ this represents an advantage over the iterated Richardson accelerator
$R_{\gamma_k}\circ\cdots\circ R_{\gamma_2}$. 
The simplest usage is to set $r_j=j$ for $1\le j\le k$, dropping the vector $\boldsymbol{r}$ from the notation,  which is what we used in our analysis.

\section{Ansatz Analysis and Oscillatory Damping}
\label{appendix: ansatz analysis and damping}

To dampen the oscillations in the sequences $\lambda_\textnormal{back/over}^{(M)}$ for $1 \leq M \leq 4$, Kolmogorov--Zurbenko (KZ) low-pass filters~\cite{MR860209,YangZurbenko:2010} are applied before Richardson acceleration. The KZ filter $\textnormal{KZ}_{m,k}$ can be regarded as $k\in\mathbb{N}$ iterations of a moving average over an odd number of points $m\in\mathbb{N}$, and preserves limits of convergent sequences. We chose $k=5$ for each filter application and then for each $M$ chose the parameters $m_M$ such that the $3^\text{rd}$ successive differences have significantly diminished oscillations. Table~\ref{tab:KZ_parameters} shows the choice of filtering for each sequence of approximate eigenvalues. 
\begin{table}
\begin{center}
    \begin{tabular}{ | m{3em} | m{3em} | m{6em} | m{6em} | m{6em} |}
        \hline
        $M$ & $1$ & $2$ & $3$ & $4$\\
        \hline
        Filter & $\textnormal{KZ}_{7,5}$ & $\textnormal{KZ}_{13,5}\circ \textnormal{KZ}_{7,5}$ & $\textnormal{KZ}_{24,5}\circ \textnormal{KZ}_{7,5}$ & $\textnormal{KZ}_{38,5}\circ \textnormal{KZ}_{7,5}$\\
        \hline
    \end{tabular}
\end{center}
\caption{\label{tab:KZ_parameters} KZ filters applied for $1\leq M \leq 4$.}
\label{tab: KZ_filters}
\end{table} 
With the oscillations significantly dampened, we begin analyzing the non-decreasing backflow sequences by assuming that each sequence $\lambda^{(M)}_\textnormal{back}$ has successive differences of the form 
\begin{equation}
\label{eq: max_sd_ansatz}
    \lambda(N+1)-\lambda(N) = \frac{\beta \gamma}{N^{\gamma+1}}+\mathcal{O}\left(\frac{1}{N^{\gamma+1+\delta}}\right)
\end{equation}
for some $\beta,\gamma,\delta>0$ to be determined. Resumming, this implies that each $\lambda^{(M)}_\textnormal{back}$ has the form 
\begin{equation}\label{eq:lambdaNresummed}
    \lambda(N) = \alpha - \frac{\beta}{N^{\gamma}} + \mathcal{O}\left(\frac{1}{N^{\gamma+\delta}}\right)
\end{equation}
for some constant $\alpha\in \mathbb{R}$, which is the limit of $\lambda(N)$ as $N\to\infty$ and
estimates the supremum of the spectrum of the relevant backflow operator. We will write $\beta^{(M)}_\textnormal{back},\gamma^{(M)}_\textnormal{back}$ to indicate the leading order parameters of the sequences $\lambda^{(M)}_\textnormal{back}$. 

To find the leading order parameters of $\lambda^{(M)}_\textnormal{back}$, we make use of two tools: a $\log\text{-}\log$ plot and the Pearson correlation coefficient. 
Taking the logarithm of \eqref{eq: max_sd_ansatz}, we find that 
\begin{equation}
    \log(\lambda_{N+1}-\lambda_N)=\log(\beta\gamma)-(\gamma+1)\log(N)+\mathcal{O}\left(\frac{1}{N^{\delta}}\right)
\end{equation}
and so one can estimate the parameters $\beta,\gamma$ by the estimating the intercept and gradient of the plot of $\log (\lambda_{N+1}-\lambda_N)$ against $\log N$.

The second tool in our arsenal is the Pearson correlation coefficient. Given two sequences $X,Y$, their Pearson correlation coefficient $\rho(X,Y)\in [-1,1]$ is given by
\begin{equation}
\label{eq: pearson_defn}
    \rho(X,Y)=\frac{\Cov(X,Y)}{\sqrt{\Var(X)\Var(Y)}},
\end{equation}
which measures the extent to which $Y$ is linearly dependent on $X$, with values of $\rho(X,Y)=\pm 1$ for the case of perfect linear correlation with positive/negative coefficient, and values closer to $0$ indicating poor linear correlation. 
For a given $\gamma>0$ and sequence $\lambda$, consider the sequence 
\begin{equation}
\mu(\gamma; \lambda)_N=N^{\gamma+2}(\lambda_{N+1}-\lambda_N).
\end{equation}
Comparing with~\eqref{eq: max_sd_ansatz}, we will choose the value of $\gamma$ so as to maximise the Pearson coefficient.

Using the log-log plots, we obtained Table~\ref{tab:N500maxparameters} of approximate values of $\beta^{(M)}_\textnormal{back},\gamma^{(M)}_\textnormal{back}$ to $3$ significant figures. 

\begin{table}
\renewcommand{\arraystretch}{1.2}
\begin{center}
    \begin{tabular}{ | m{1em} | m{8em} | m{8em} |}
        \hline
        & & \\[-1em]
        $M$ & $\beta^{(M)}_\textnormal{back}$ & $\gamma^{(M)}_\textnormal{back}$\\
        \hline
        $1$ & $0.0329$ & $0.494$\\
        \hline
        $2$ & $0.0739$ & $0.495$\\
        \hline
        $3$ & $0.118$ & $0.502$\\
        \hline
        $4$ & $0.0168$ & $0.514$\\
        \hline
    \end{tabular}
\end{center}
\caption{\label{tab:N500maxparameters}Values of $\beta^{(M)}_\textnormal{back},\gamma^{(M)}_\textnormal{back}$ for $1\ \leq M\leq 4$.}
\end{table}

Meanwhile, computation of the Pearson coefficients all give peaks within $0.02$ of $\gamma=0.5$. 
On natural grounds, we argue that $\gamma^{(M)}_\textnormal{back}=0.5$ for each $M$, which justifies applying the accelerator $R_{0.5}$ to each sequence.
We can then repeat the process of trying to estimate the next term in a putative asymptotic expansion for $\lambda^{(M)}_{\textnormal{back}}$. In the case $M=1$, we obtained good evidence for a power $\gamma=1$, and after applying the combined accelerator $R_{(0.5,1)}$ to the original sequence, we found evidence that the next power is below $2$. This suggests a conjecture that $\lambda_\textnormal{back}^{(1)}$ admits an asymptotic series of the form
\begin{equation}
\label{eq: asymptotic_expansion}
    \lambda_\textnormal{back}^{(1)}(N) \sim c_\textnormal{BM} - \sum_{k \geq 1} \frac{\beta_k}{N^{k/2}}
\end{equation}
for some coefficients $\beta_k$ with $\beta_1 > 0$. The numerical evidence for an asymptotic expansion of $\lambda^{(M)}_\textnormal{back}$ for $M=2,3,4$ beyond the initial $N^{-1/2}$ term is rather weaker, but we conjecture that they have a similar asymptotic expansion to the case $M=1$. However, the results of successive acceleration are limited by the extra smoothing required to limit the impact of the more complicated oscillations in the sequences.

\section{Position Representation and Probability Flux Plots}
\label{sec:pos_plots}

This supplementary section contains plots of the time evolution of the backflow and overflow states as well as their respective probability flux for $M=1,\dots,4$. As discussed in section~\ref{sec:numerical_methodology_results}, we obtained approximate eigenvectors $\psi_\textnormal{back/over}^{(M)}(a_M,-1/4;N;\cdot)$ in the transformed coordinates $k=\sqrt{2p}$ where $p$ is the momentum. For $U^\ast: L^2(\mathbb{R}^+,dp)\to L^2(\mathbb{R}^+,dk)$, we find the position space representations for $N=100$ using
\begin{equation}
    \chi^{(M)}_\textnormal{back/over}=\mathcal{R} U^\ast \psi_\textnormal{back/over}^{(M)}(a_M,-1/4;100;\cdot).
\end{equation} 
Note that the unitaries $\mathcal{R}$ and $U^\ast$ are applied to $\psi_\textnormal{back/over}(a_M,-1/4;100;\cdot)$ term wise, a calculation that can be done exactly by e.g., Maple\textsuperscript{TM} 2024. The time evolved versions are computed as described in the main text. Figures~\ref{fig:backvecM1heatmap_app}--\ref{fig:overvecM4heatmap_app} show heatmaps of the time evolution of $\chi^{(M)}_\textnormal{back/over}$ for $-1.5\leq t\leq 1.5$. The $M=2$ time evolution plots are shown here for completeness despite the fact that they appear in the main document as figures~\ref{fig:backvecM2heatmap} and \ref{fig:overvecM2heatmap}.

\begin{figure}
    \centering
    \includegraphics[width=0.8\textwidth]{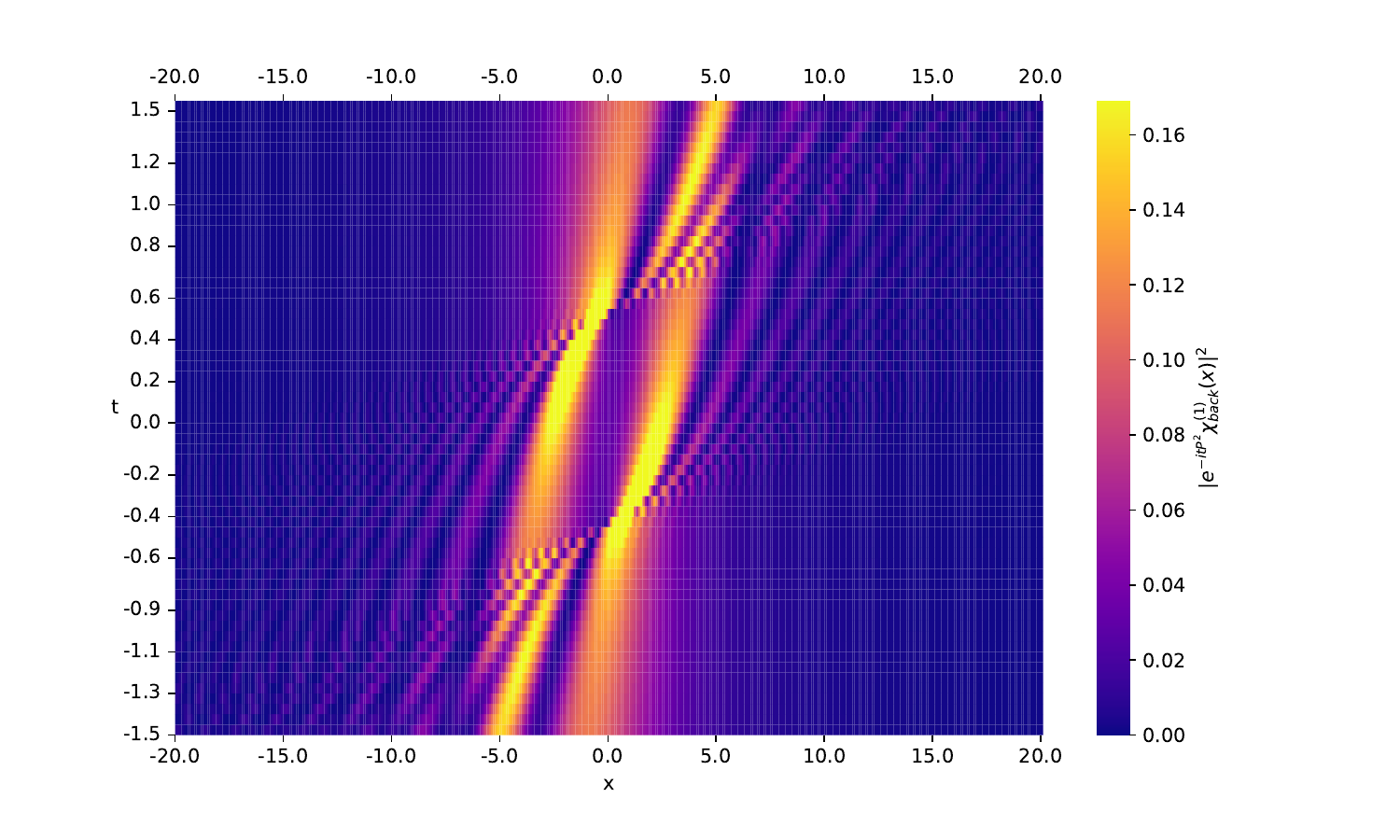}
    \caption{Time evolution of the position probability density for the approximate $M=1$ backflow maximizing state $\chi^{(1)}_\textnormal{back}$.}
    \label{fig:backvecM1heatmap_app}
\end{figure}

\begin{figure}
    \centering
    \includegraphics[width=0.8\textwidth]{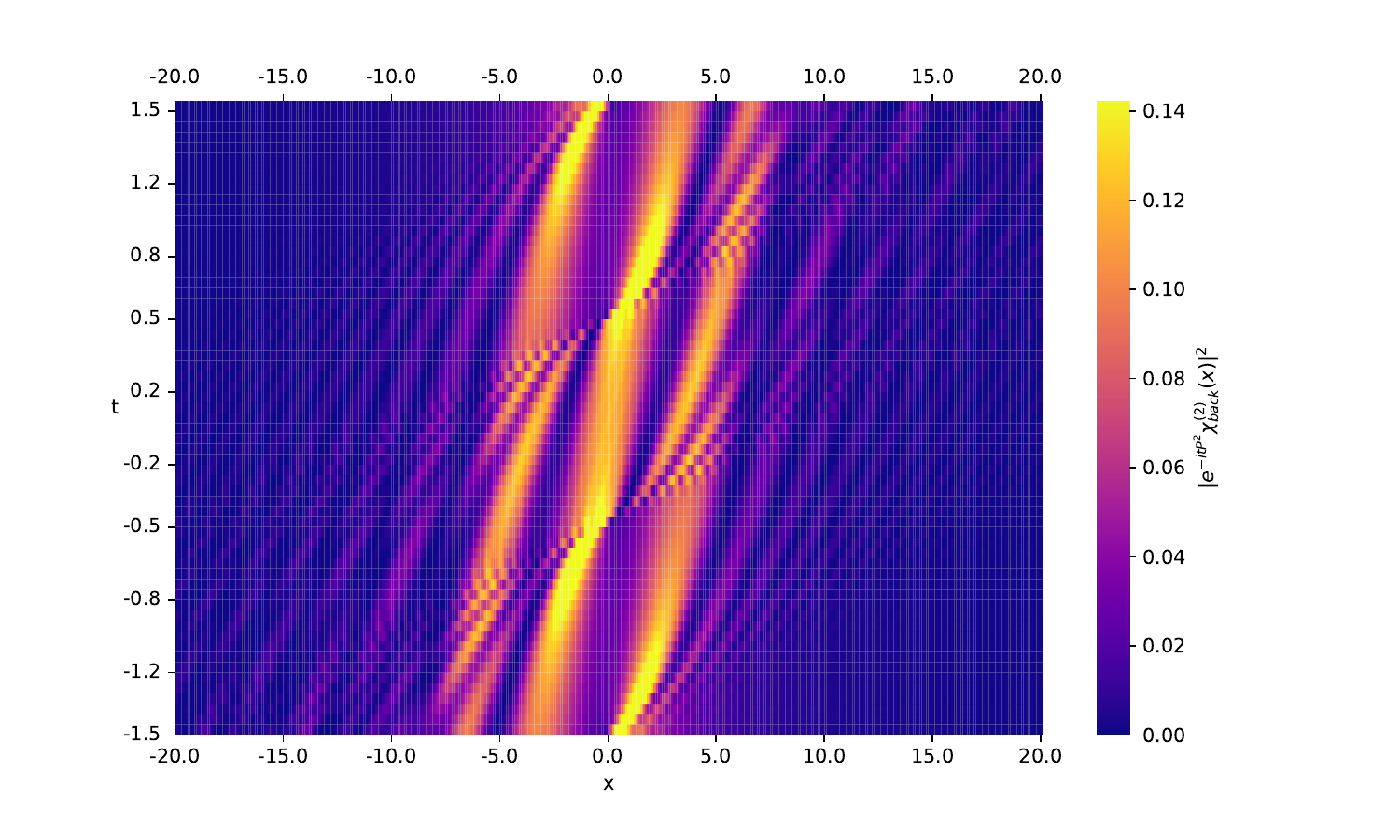}
    \caption{Time evolution of the position probability density for the approximate $M=2$ backflow maximizing state $\chi^{(2)}_\textnormal{back}$.}
    \label{fig:backvecM2heatmap_app}
\end{figure}

\begin{figure}
    \centering
    \includegraphics[width=0.8\textwidth]{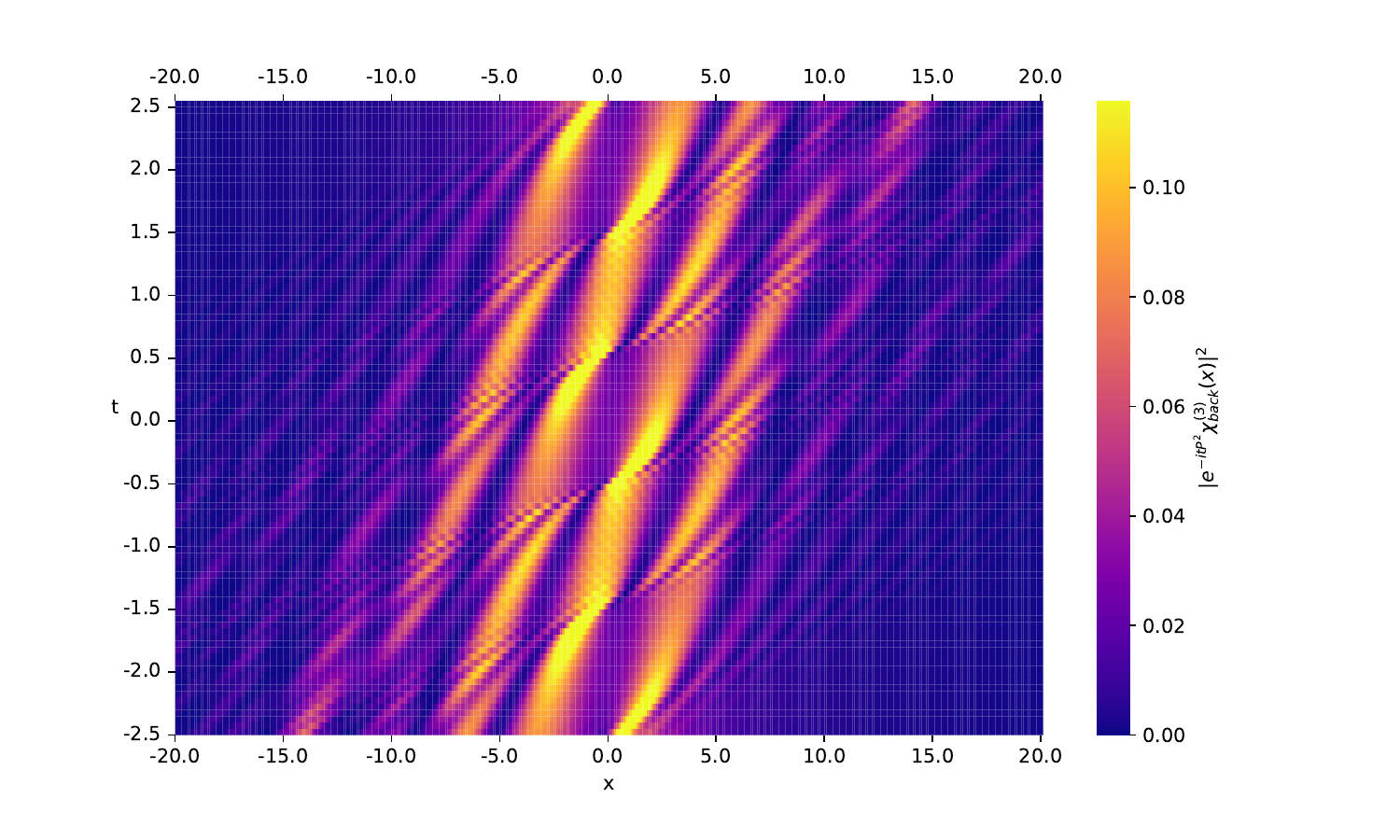}
    \caption{Time evolution of the position probability density for the approximate $M=3$ backflow maximizing state $\chi^{(3)}_\textnormal{back}$.}
    \label{fig:backvecM3heatmap_app}
\end{figure}

\begin{figure}
    \centering
    \includegraphics[width=0.8\textwidth]{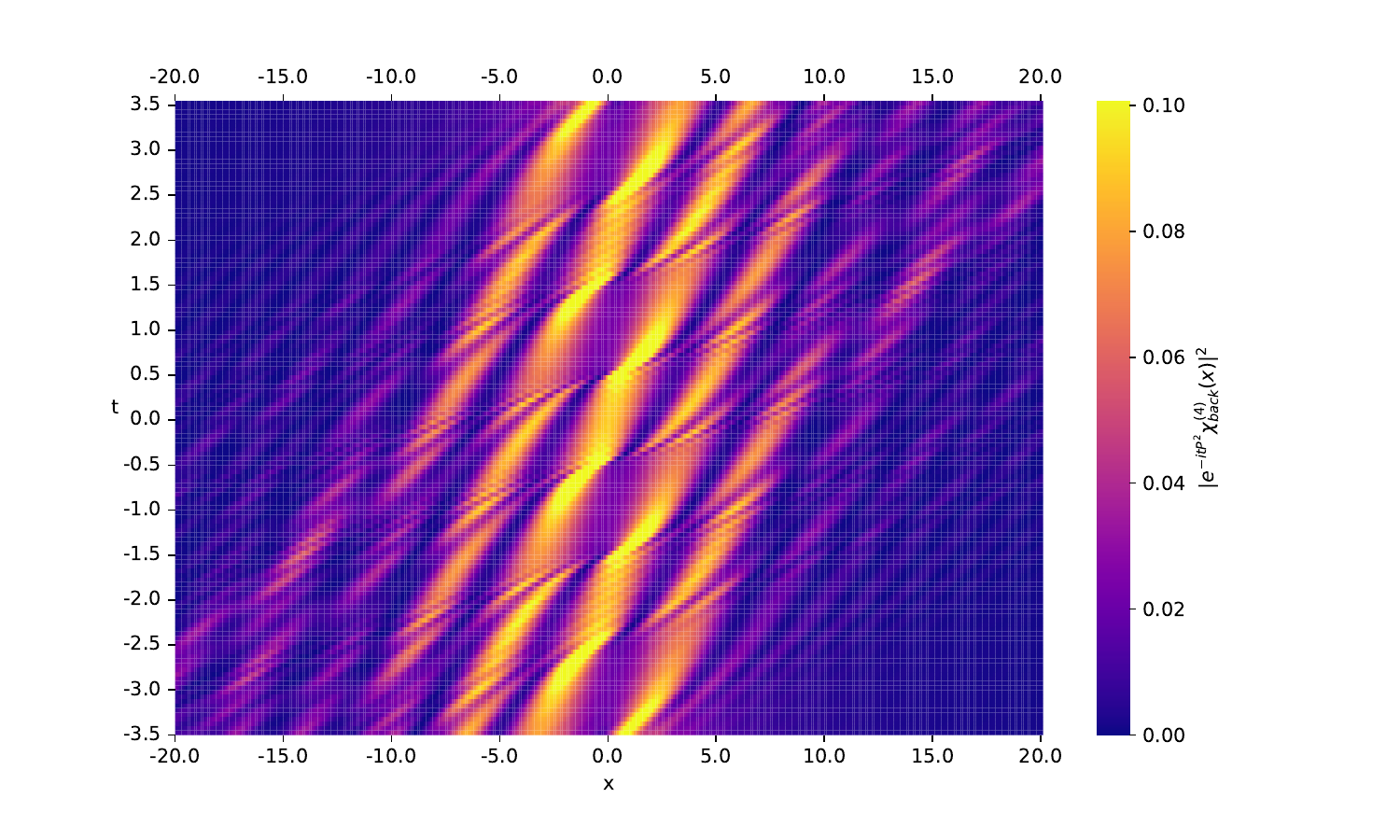}
    \caption{Time evolution of the position probability density for the approximate $M=4$ backflow maximizing state $\chi^{(4)}_\textnormal{back}$.}
    \label{fig:backvecM4heatmap_app}
\end{figure}

\begin{figure}
    \centering
    \includegraphics[width=0.8\textwidth]{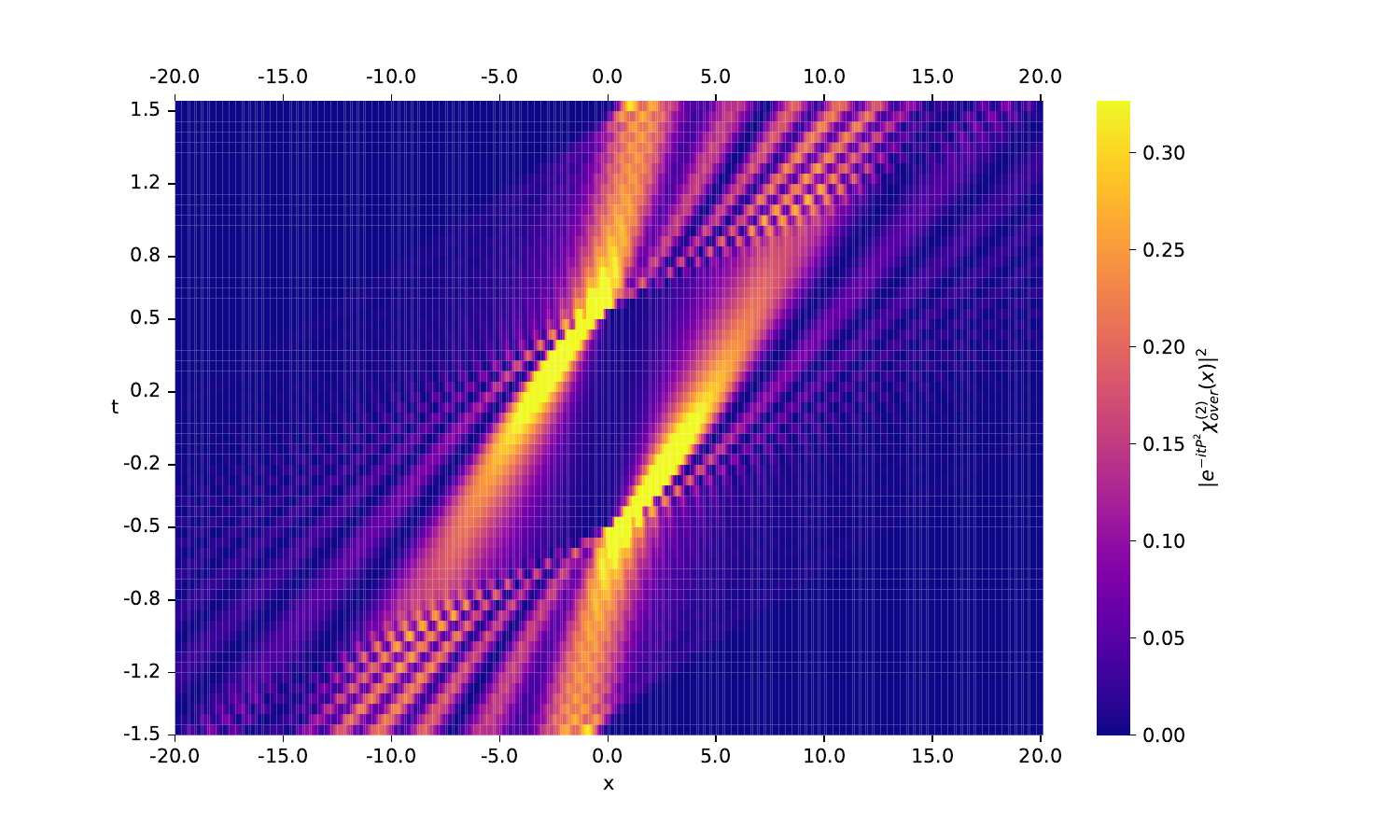}
    \caption{Time evolution of the position probability density for the approximate $M=2$ overflow maximizing state $\chi^{(2)}_\textnormal{over}$.}
    \label{fig:overvecM2heatmap_app}
\end{figure}

\begin{figure}
    \centering
    \includegraphics[width=0.8\textwidth]{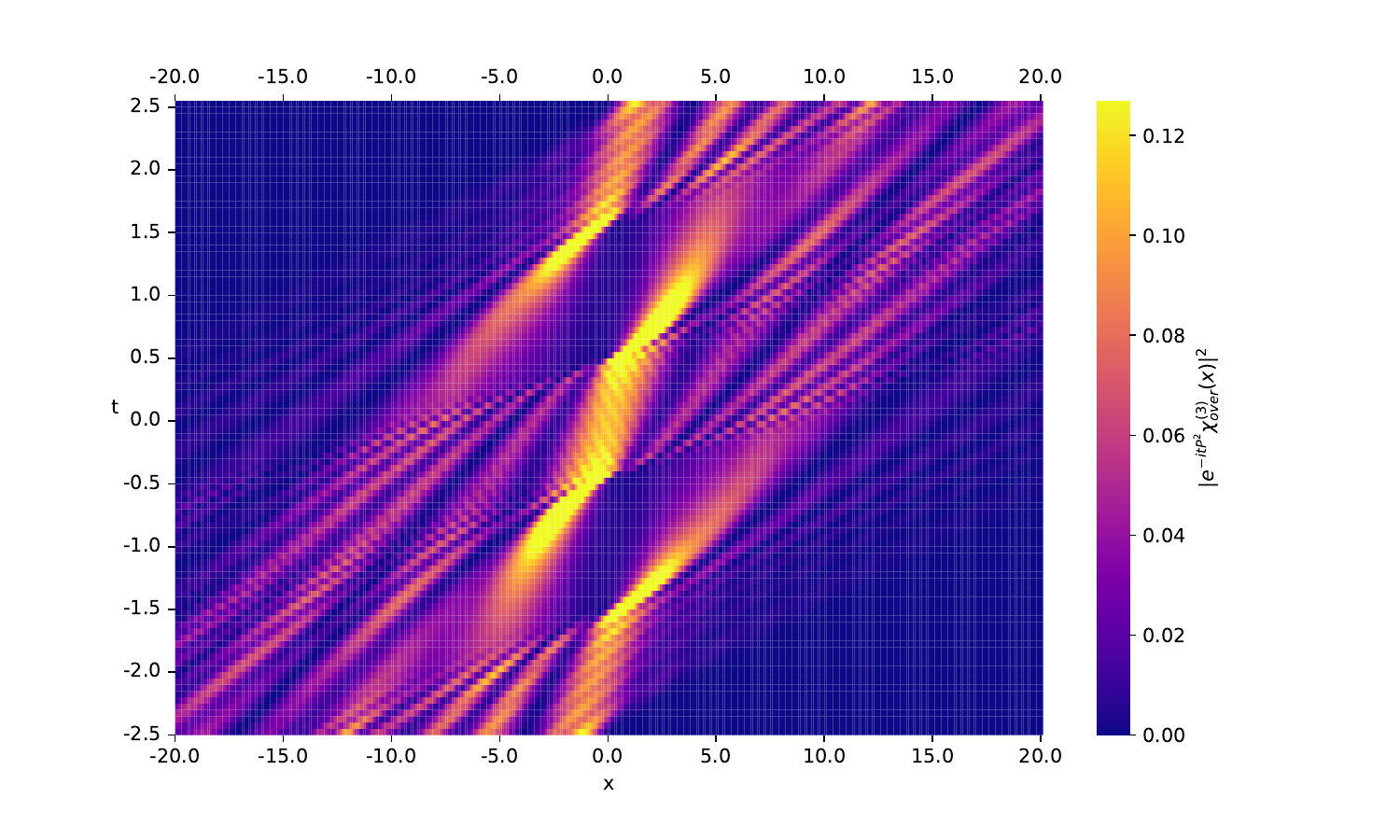}
    \caption{Time evolution of the position probability density for the approximate $M=3$ overflow maximizing state $\chi^{(3)}_\textnormal{over}$.}
    \label{fig:overvecM3heatmap_app}
\end{figure}

\begin{figure}
    \centering
    \includegraphics[width=0.8\textwidth]{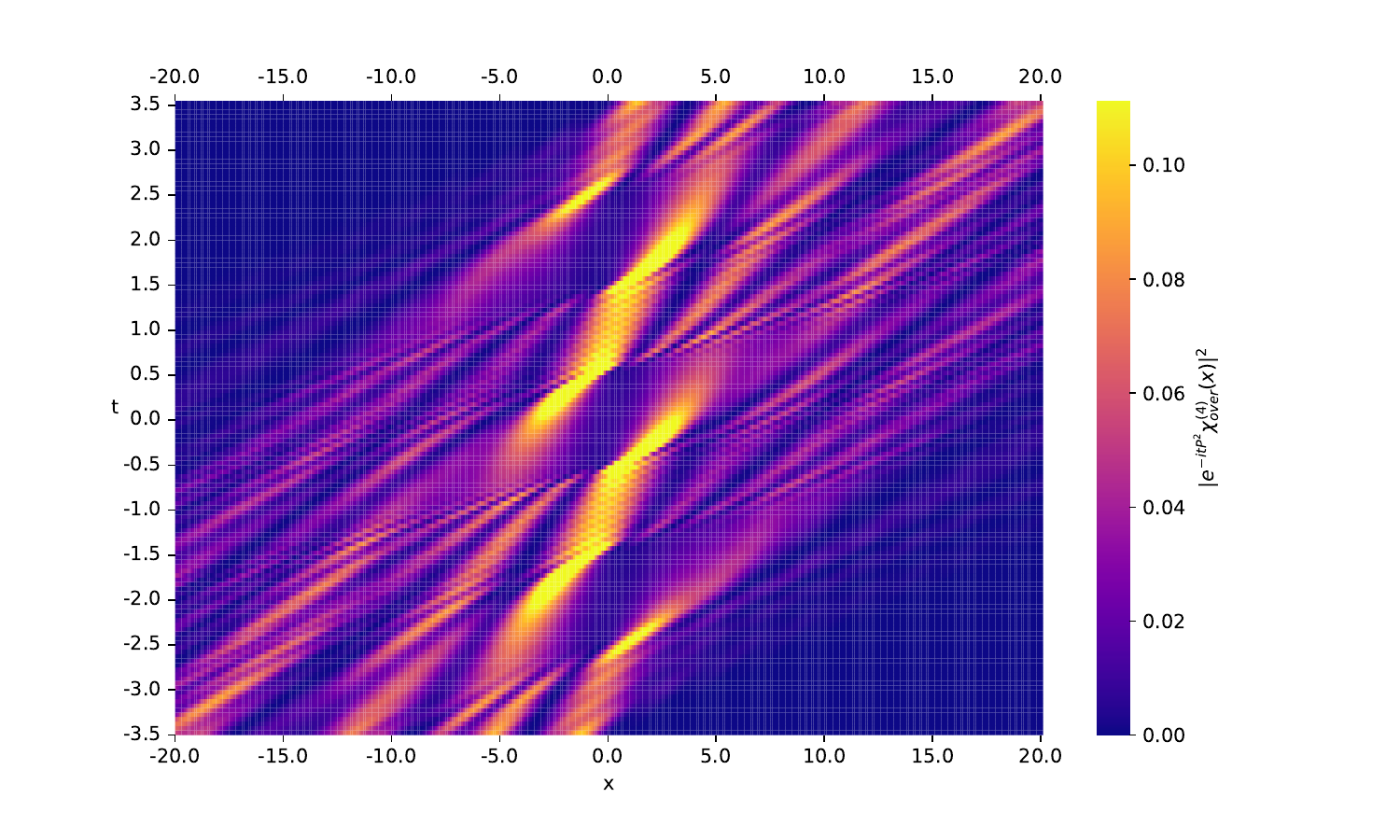}
    \caption{Time evolution of the position probability density for the approximate $M=4$ overflow maximizing state $\chi^{(4)}_\textnormal{over}$.}
    \label{fig:overvecM4heatmap_app}
\end{figure}

From the $\chi^{(M)}_\textnormal{back/over}$ we obtain the time-zero probability flux according to the standard formula
\begin{equation}
    j_\chi(x) = 2\Im \chi^\ast (x)\chi'(x) 
\end{equation}
where the factor of $2$ is a result of the choice of mass units.  Figures~\ref{fig:fluxbackflowM1_4J} and~\ref{fig:fluxoverflowM2_4J} show the probability flux for the backflow/overflow vectors associated with $M$ intervals of backflow/overflow. Specifically the intervals are given by
\begin{gather}
    [-0.5,0.5]\\
    [-1.5,-0.5],[0.5,1.5]\\
    [-2.5,-1.5],[-0.5,0.5],[1.5,2.5]\\
    [-3.5,-2.5],[-1.5,-0.5],[0.5,1.5],[2.5,3.5]
\end{gather}
respectively, for $M=1,2,3,4$. 

\begin{figure}
    \centering
    \includegraphics[width=\textwidth]{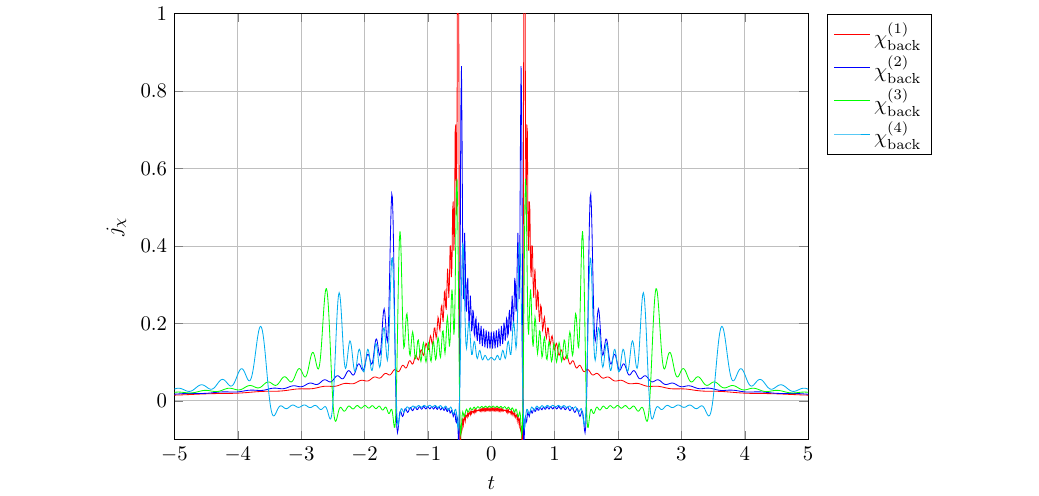}
    \caption{Probability flux at $x=0$ of the backflow states for $M=1,\dots,4$}
    \label{fig:fluxbackflowM1_4J}
\end{figure}

\begin{figure}[h]
    \centering
    \includegraphics[width=\textwidth]{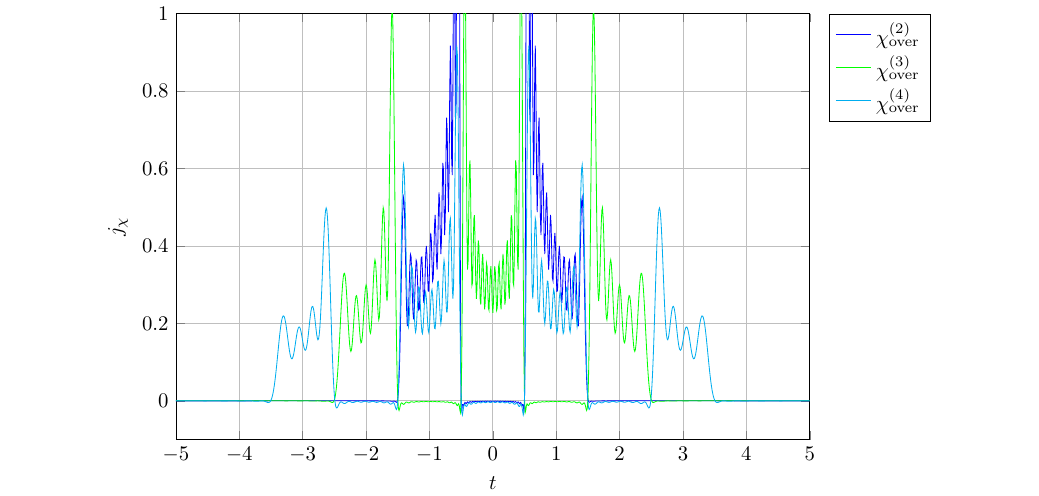}
    \caption{Probability flux at $x=0$ of the overflow states for $M=1,\dots,4$}
    \label{fig:fluxoverflowM2_4J}
\end{figure}
\end{document}